\documentclass[english,11pt]{elsarticle}

\usepackage{epsfig}
\usepackage{subfigure}
\usepackage{amssymb,amsmath,amsfonts,layout,graphicx,amsthm}
\usepackage{makeidx}
\usepackage{babel}
\usepackage{tikz}
\usepackage{sublabel}
\usepackage{cases}
\usepackage{algorithm}
\usepackage{algorithmic}
\usepackage{color}
\usepackage{comment}
\usepackage{booktabs}
\usepackage{longtable}
\usepackage{multirow}
\usepackage[title]{appendix}

\usepackage[colorlinks,
linkcolor=red,
anchorcolor=blue,
citecolor=green
]{hyperref}

\usepackage
[
        letterpaper,
        left=1.91cm,
        right=1.91cm,
        top=1.91cm,
        bottom=2cm,
]
{geometry}

\newtheorem{lemma}{Lemma}

\newtheorem{proposition}{Proposition}

\newtheorem{definition}{Definition}

\newtheorem{remark}{Remark}

\newcommand{\x}{{\bf x}}

  \usepackage{pgfplots}
  \pgfplotsset{compat=newest}
  \usetikzlibrary{plotmarks}
  \usetikzlibrary{arrows.meta}
  \usepgfplotslibrary{patchplots}
  \usepackage{grffile}
  \usepackage{amsmath}

\begin{document}
\begin{frontmatter}
\title{Piggyback on Idle Ride-Sourcing Drivers for Integrated On-Demand and Flexible Intracity Parcel Delivery Services}
\author[hkust_address]{Yang Liu}
\ead{yliuiw@connect.ust.hk}
\author[hkust_address]{Sen Li}
\ead{cesli@ust.hk}

\address[hkust_address]{Department of Civil and Environmental Engineering, The Hong Kong University of Science and Technology, \\Clear Water Bay, Kowloon, Hong Kong}
\begin{abstract}

This paper investigates the spatial pricing and fleet management strategies  for an integrated platform that provides both ride-sourcing services and intracity parcel delivery services over a transportation network utilizing the idle time of ride-sourcing drivers. Specifically, the integrated platform simultaneously offers on-demand ride-sourcing services for passengers and multiple modes of parcel delivery services for customers, including: (1) \emph{on-demand delivery}, where drivers immediately pick up and deliver parcels upon receiving a delivery request; and (2) \emph{flexible delivery}, where drivers can pick up (or drop off) parcels only when they are idle and waiting for the next ride-sourcing request. A continuous-time Markov Chain (CTMC) model is proposed to characterize the status change of drivers under joint movement of passengers and parcels over the transportation network with limited vehicle capacity, where the service quality of ride-sourcing services, on-demand delivery services, and flexible delivery services are rigorously quantified. Building on the CTMC model, incentives for ride-sourcing passengers, delivery customers, drivers, and the platform are captured through an economic equilibrium model, and the optimal spatial pricing decisions of the platform are derived by solving a non-convex profit-maximizing problem. We prove the well-posedness of the model and develop a tailored algorithm to compute the optimal decisions of the platform. Furthermore, we validate the proposed model in a comprehensive case study for San Francisco, demonstrating that joint management of ride-sourcing services and intracity package delivery services can lead to a Pareto improvement that benefits all stakeholders in the integrated ride-sourcing and parcel delivery market under realistic parcel and passenger demand patterns.

\end{abstract}
\end{frontmatter}

\section{Introduction}
The rapid growth of e-commerce worldwide has led to a surge in demand for delivery services over recent years. 
While long-haul trucks have been used for long-distance intercity delivery for decades, logistics companies face significant challenges in short-distance intracity parcel delivery, such as last-mile delivery or same-day delivery services within a city. In an urban context, maintaining a fleet of vans or trucks for delivery is costly due to dispersed customer demand, stringent driver requirements, and high gasoline consumption of vehicles. Logistics companies may opt to maintain a smaller fleet size to save costs, but this may come at the expense of delivery speed. Furthermore, the exponential growth of delivery platforms results in significant negative externalities for city residents, prompting complaints or legal actions. For instance, large delivery vehicles contribute to traffic congestion and accidents within cities due to their size and lack of flexibility \cite{cherry2012truck}. Moreover, upon reaching their destination, trucks must park and wait for all packages to be delivered, potentially obstructing roads, entrances, or violating traffic regulations. Additionally, trucks contribute to air pollution in urban areas \cite{kafle2017design}. Consequently, logistics companies are compelled to explore improved business models to fulfill the growing demand for rapid delivery services \cite{jacobs2019lastmile} and mitigate adverse impacts on urban residents.


To overcome the aforementioned challenges, a co-modality model has been proposed for last-mile delivery or same-day delivery, where ride-sourcing services and parcel delivery services are integrated on the same platform to utilize the same group of gig workers. The success of the model is based on the premise that ride-sourcing and parcel delivery services are highly complementary, and the benefits of integrating these two services can be significant. In particular, co-modality enables the logistic company to access a large number of for-hire drivers that are already working on the existing ride-sourcing platforms rather than attempting to recruit a new fleet of gig drivers. This can avoid competition and reduce investment cost. More importantly,  
the success of the ride-sourcing platform crucially rely on quick response to passenger pickup requests, which relies on a large number of available but {\em idle} for-hire drivers. Consequently, ride-sourcing drivers may face extended periods of inactivity between rides. Since parcel delivery customers tend to be more patient than ride-hailing passengers\footnote{Ride-sourcing companies such as Uber and Lyft provide on-demand services to passengers with short waiting time and travel time, while the delivery services provided by most existing logistic companies such as SF  require one or more days to deliver the parcel.}, parcel delivery services typically exhibit greater flexibility in terms of delivery schedules compared to ride-hailing services. For this reason, idle ride-sourcing drivers can perform package delivery while waiting for the next ride-hailing request, making more efficient use of their time and generating additional income. This can benefit multiple market stakeholders, including the ride-sourcing and parcel delivery platforms, as well as for-hire drivers, without negatively affecting ride-hailing passengers.

Despite the potential benefits, several research gaps must be addressed before the full potential of the integrated business model can be realized. Intuitively, the success of this model hinges significantly on capitalizing on both demand and supply aspects, including effectively leveraging parcel demand flexibility from the demand side and efficiently utilizing ride-sourcing drivers' idle time from the supply side. However, existing co-modality studies largely overlook these factors.  From the demand side, to fully exploit the flexibility in parcel demand, it is crucial to consider pricing incentives that motivates parcel demand users to offer flexibility when using the services. Consequently, a hybrid class of heterogeneously elastic demand should be considered, wherein ride-sourcing passengers are sensitive to price and waiting time, parcel delivery customers are sensitive to price and delivery time, and simultaneously, ride-sourcing passengers exhibit a higher degree of impatience compared to parcel delivery customers. However, to the best of our knowledge, none of the existing works have considered differentiated pricing based on the level of flexibility offered by demand in the co-modality model, nor do they take into account demand elasticity (encompassing both ride-sourcing passengers and parcel delivery customers) with respect to service fare or service quality (i.e., waiting time, delivery time, etc.). From the supply side, to efficiently utilize the idle time of drivers, it is essential to explicitly model the idle time by characterizing the drivers' status changes as they transition between distinct tasks (e.g., cruising, ride-sourcing, parcel delivery, etc.) across the transportation network, so as to pinpoint how the ride-sourcing driver's time is allocated for delivery services, where it is devoted on the transportation network, and how it affects the quality of the ride-sourcing services. We find these crucial modeling components absent in the existing literature. 

This paper aims to fill this research gap by investigating the integration of ride-sourcing and parcel delivery service utilizing the idle time of ride-sourcing drivers, encompassing both on-demand and flexible parcel delivery options. In particular, we consider an integrated platform that provides on-demand ride-sourcing services for passengers and multiple modes of delivery services for customers, including: (1) \emph{on-demand delivery}, where drivers immediately pick up and delivery the parcels upon a delivery request is placed; and (2) \emph{flexible delivery}, where drivers can pick up (or drop off) the parcel only when they are idle and  waiting for the next ride-sourcing request, and can carry parcels along with passengers while delivering the passenger from origin and destination. In this context, the on-demand package delivery services offers immediate responses to customer request, and can be used to deliver urgent and time-sensitive parcel at small scale, whereas flexible delivery services allows the drivers to pick up or drop off the package during their idle time, enabling better integrating with ride-sourcing services and significantly lower delivery fees for customers. To assess the economic impacts of this integrated platform, we develop a mathematical model that characterizes the interactions between ride-sourcing and package delivery services. The performance of the integrated platform is compared with that of standalone ride-sourcing services, and we show that the integrated platform can lead to Pareto improvement in certain conditions.  The major contributions of this paper are summarized as follows:

\begin{itemize}
\item We investigate the differentiated spatial pricing strategies and fleet management strategies for an integrated platform that combines ride-sourcing and parcel delivery services, encompassing both on-demand and flexible parcel delivery options that are priced at different rates. A continuous-time Markov Chain (CTMC) is proposed to characterize the status changes of drivers as they transition between distinct tasks (e.g., cruising, ride-sourcing, parcel delivery, etc.)  over the transportation network under limited vehicle capacity, where the service quality of ride-sourcing services, on-demand delivery services, and flexible delivery services are rigorously quantified. Building on the CTMC model, we capture the incentives of ride-sourcing passengers, delivery customers, drivers, and the platform through an economic equilibrium, and derive the optimal spatial pricing and fleet management decisions of the platform using a non-convex profit-maximizing problem. {\em To the best of our knowledge, the aforementioned co-modality model has not been examined, and the existing literature on co-modality lacks a mathematical model that explicitly addresses parcel demand flexibility (under an elastic demand model) or ride-sourcing driver's idle status (through CTMC)}.

\item To address the non-convex nature of the profit-maximization problem, we delve into the constraints' structure and identify the interdependencies among different endogenous variables in our model. Leveraging these structural insights, we proceed by establishing the conditions necessary for the existence of these endogenous variables given a set of decision variables. Subsequently, we develop a customized algorithm tailored to solving the platform's profit-maximization problem. This algorithm involves a reformulation of the original problem, resulting in a significantly reduced number of decision variables and constraints. Notably, our tailored algorithm exhibits accelerated convergence rates and showcases more robust numerical performance in comparison to standard interior-point methods.

\item We validate the proposed model and algorithm through simulation studies based on a combination of real and synthesized data in San Francisco. We demonstrate through numerical studies that the proposed algorithm can improve the computation time and numerical stability (with respect to initial guess) in solving the profit-maximization problem compared to the benchmark method. We also evaluated the integrated model under distinct parcel demand patterns, including  (1) the first case where the spatial distribution of parcel delivery demand is synthesized based on the real data; and (2) the second case where the spatial distribution of ride-sourcing demand follows the opposite pattern as that of parcel delivery demand. We show that in both cases, the integration of ride-sourcing services and delivery services can bring significant benefits to the platform, passengers, and drivers.
\end{itemize}

\section{Literature Review}

This section reviews the existing literature that is related to our research.
We categorize the related research into three classes: (1) ride-sourcing services; (2) parcel delivery services; and (3) co-modality where passengers and parcels are jointly transported.
The differences between our work and other related literature will be discussed, and the contributions of our work will be summarized.

\subsection{Ride-Sourcing Services}\label{sec:lite_ride}

A significant body of literature exists on the modeling and optimization of ride-sourcing platforms, with pricing as a recurring topic \cite{sun2019optimal,wang2016pricing,zha2016economic,xu2021generalized}. Current strategic pricing problems for ride-sourcing platforms can be divided into static pricing, dynamic/surge pricing, and spatial pricing. The first class of literature focuses on the optimal pricing strategies under the responses of the market equilibrium considering interactions of prices and various endogenous factors. The common methodology for this kind of problem is to formulate the pricing strategies as an profits-maximizing optimization problem subject to the market equilibrium under distinct considerations. For instance, \cite{bai2019coordinating} investigated optimal pricing strategies for on-demand service platforms, taking into account earning-sensitive providers and time and price-sensitive customers using a queuing model. The study demonstrated that platforms can raise prices in response to increased demand, but the optimal price need not be monotonic with respect to increasing waiting costs or provider capacity. A similar setting was explored by \cite{taylor2018demand}, who also considered uncertainties when examining optimal pricing strategies in on-demand services with waiting-time sensitive customers and independent providers. The research indicated that delay sensitivity increases the optimal price when customer valuation uncertainty is moderate, while reducing the wage when the provider has high opportunity cost uncertainty and a moderate expected opportunity cost. \cite{hu2020price} compare the optimal pricing strategies and platform profits when the platform determines the price and wage independently and when the platform determines the price subject to a fixed commission contract. A lower bound is established to help the platform make decisions on whether implement the fixed-commission contract.
Furthermore, spatial pricing \cite{zhu2021mean,zha2018geometric,chen2021spatial} and dynamic pricing \cite{cachon2017role,banerjee2015pricing,zha2017surge} have been extensively explored. In the spatial dimension, \cite{li2021spatial} analyzed optimal pricing strategies for ride-sourcing platforms on transportation networks with congestion charges. A platform's profit-maximization optimization problem is developed and an algorithm was proposed that approximates the optimal pricing strategy with a tight upper bound. A bi-level optimization problem is formulated by \cite{tang2022bi}, where the upper level captures the optimal prices and wages determined by the platform with spatial heterogeneity, and the lower level captures the trip distribution model. 
In the temporal dimension, \cite{chen2020dynamic} assessed the impact of dynamic pricing on the ride-sourcing market using a dynamic vacant car-passenger meeting model characterizing the effects of short-term fluctuations and disturbances in demand and supply on the waiting number of vacant vehicles and passengers. Additionally, \cite{nourinejad2020ride} introduced a dynamic non-equilibrium model to characterize the time-varying macroscopic states of ride-sourcing markets, developing a model-predictive control approach to investigate the optimal dynamic pricing strategies for ride-sourcing platforms. For a comprehensive review, please refer to \cite{wang2019ridesourcing}.

 Aside from ride-sourcing services for single passenger, the pricing and fleet management strategies of ride-pooling services have also attracted increasing attention. For instance, \cite{zhang2021pool} studied the market equilibrium, pricing strategy and the impacts of regulation policies for the ride-pooling services provided by the  e-hail platforms. Both platform profits maximization case and social welfare maximization case are evaluated and the impacts of the minimum wage and congestion tax on the ride-pooling market are investigated. \cite{bahrami2022optimal} investigates the optimal pricing strategies, vehicle occupancy and fleet size of solo and pool services for on-demand ride-sourcing services with both high-value and low-value passengers. It examined whether providing ride-pooling services is profitable or not under distinct elastic parameters of the meeting function, illustrating that offering both solo and pool services together generates more profit only if the meeting function has decreasing returns-to-scale.  
\cite{ke2020pricing} presented a comparison of optimal pricing strategies in ride-pooling markets under two market equilibrium scenarios: ride-pooling market and non-pooling market. It shows that the decrease in trip fare can attract more passengers in ride-pooling market than in non-pooling market by reducecd detour time.
It is further extended by \cite{ke2022coordinating}, which examined the optimal pricing strategies of a ride-sourcing markets with both ride-pooling and non-pooling services considering traffic congestion externality.  A deductive model was proposed to estimate the pool-matching probability and expected waiting time. It illustated that the increase of background traffic congestion will hurt the profits of the platform, but its impacts on the optimal fleet size is non-monotimic.

Our paper differs from the above literature as we consider the optimization for an integrated platform that offers both ride-sourcing and delivery services. This case differs in nature from  pure ride-sourcing platforms because the integrated platform needs to consider how pricing decisions influence the ride-sourcing and parcel delivery demand, and how these two services are intimately correlated with each other over the transportation network. These elements have not been considered in the aforementioned works.


\subsection{Parcel Delivery Services}\label{sec:lite_parcel}
In addition to passenger transportation, a growing body of literature investigates parcel delivery services \cite{taniguchi2000evaluation,quak2009delivering,crainic2009models,cattaruzza2017vehicle}. For example, \cite{taniguchi2004intelligent} examined the dynamic vehicle routing and scheduling problem with time windows for city logistics, incorporating a dynamic traffic simulation providing real-time information. The study demonstrated that the proposed framework reduces costs for couriers and alleviates traffic congestion. \cite{liu2019minimizing} explored vehicle routing problems in urban logistics, aiming to minimize total completion time (including both travel and assembly time). A chance-constrained programming approach was developed to account for uncertainty in assembly time, with its distribution characteristics estimated using a statistical learning method. The research illustrated that the proposed methods can schedule vehicle departure times to decrease waiting time on the road.
The uncertainty in travel time is considered in \cite{gross2019cost}, where an interval travel time (ITT) defined by an upper bound and lower bound of travel times is proposed to capture the variation in travel time. 
The research showed that cost-efficient and reliable vehicle routing can be achieved based on the ITT.
\cite{gayialis2022city} developed an information system that can effectively address the  scheduling and vehicle routing problems of city logistics, with an integration of operational research methods and information technologies that connects research with practice.
\cite{janjevic2020characterizing} proposed a conceptual framework that analyzes the distribution strategies for urban last-mile e-commerce deliveries.
The framework can be used to perform a comprehensive comparison between distribution strategies in mature and emerging markets. \cite{afeche2016optimal} studied the pricing strategies and lead-time menu design for a service provider to maximize the revenue from a hybrid class of patient and impatient customers, and evaluated three different features in the optimal price menu design, including pricing out the middle of the delay cost spectrum, pooling, and strategic delay.
A summary of the opportunities and challenges faced by urban logistics can be found in \cite{savelsbergh201650th}, while a review of relevant methods and algorithms for city logistics problems is presented in \cite{konstantakopoulos2020vehicle}.

Recently, a novel logistics business model utilizing crowd-resourced drivers has drawn considerable attention \cite{archetti2016vehicle,raviv2018crowd,yildiz2021express,dayarian2020crowdshipping,ahamed2021deep}. For instance, \cite{ahamed2021deep} developed a deep reinforcement learning-based framework for solving assignment problems in crowdsourced delivery, presenting a numerical analysis that showcased the efficiency of the deep Q network algorithm and the improvement in solution quality. \cite{mousavi2022stochastic} evaluated a new logistic network where parcels are delivered by occasional couriers during their regular journeys. The study examined the optimal locations for mobile depots and customer assignments in a two-tier last-mile delivery model, considering the uncertainty of crowd-shipper availability. However, it did not account for stochastic delivery demand and could not guarantee delivery services for all customers due to the randomness of crowd-shipper supply. \cite{nieto2022value} addressed these limitations by proposing a crowd-sourced city logistics framework considering both stochastic crowd capacity and stochastic demand, introducing a group of dedicated professional couriers to ensure the fulfillment of all orders. The combination of occasional drivers and dedicated drivers was also studied in \cite{arslan2019crowdsourced}, which tackled the dynamic matching problem for crowdsourced delivery, utilizing the capacity of ad hoc drivers for package delivery services during their regular journeys. Numerical studies demonstrated that employing ad hoc drivers significantly reduces operating costs for the platform compared to using dedicated drivers alone.  However, a limitation of occasional couriers for parcel delivery lies in their high uncertainty and management challenges \cite{savelsbergh2022challenges}.

\subsection{Towards Co-Modality}\label{sec:co-modality}
Another class of research investigates the joint delivery of people and goods using the same fleet. The first stream of study focuses on the use of public transit for the transportation of both people and parcels \cite{fatnassi2015planning,pimentel2018integrated,azcuy2021designing,hu2022mass}. For instance, \cite{masson2017optimization} examined the assignment problem in a two-tiered city logistics system, utilizing buses and personal city freighters to transport passengers and goods. An optimization problem was formulated to minimize the number of vehicles required for delivery, and a solution algorithm based on adaptive large neighborhood search was developed. The study demonstrated that the efficiency of transshipment from buses to city freighters is crucial for the performance of the mixed system. \cite{behiri2018urban} investigated scheduling strategies for parcel delivery services reliant on the passenger rail network, developing a decision-support tool to evaluate the impacts of freight transport on various core components of the rail network. 
\cite{kizil2023public} proposed a last-mile city logistics system built on the public transit network.
A two-stage stochastic programming is formulated and a decomposition branching algorithm is developed to find optimal public transit connections for the installation of parcel lockers. 
It is indicated that the system can improve the efficiency of delivery by increasing demand and reducing required vehicles.
Similarly, \cite{machado2023integration} analyzed the optimal selection of bus services to be adapted for freight transportation considering demand uncertainty, which is formulated as an integer linear programming problem.
The research shows that the last mile operator's capacity has the most significant impact on the required number of adapted bus services.
However, there are limitations to using public transit for parcel delivery, such as fixed timetables and predetermined routes, resulting in a lack of flexibility.

A more flexible solution for intracity delivery is to rely on ride-sourcing services. For example, \cite{qi2018shared} proposed an analytical model and economic analysis for a new logistic planning model where last-mile delivery services are fulfilled by ride-sourcing drivers. The authors provided managerial insights for platform operation and evaluated the environmental impacts of the new business model, showing that crowdsourcing shared mobility is less scalable than a truck-only system but can offer potential economic advantages and reduce emissions. \cite{liu2023economic} investigated managerial strategies for a third-party platform providing both ride-sourcing services and on-demand food-delivery services, introducing a class of hybrid drivers for both services. The study showed that integrating the two services can lead to Pareto improvement, benefiting all stakeholders in the market. However, both studies assume that drivers can either transport people or deliver parcels at one time, not considering people-parcel pooling. On the other hand, a growing body of research focuses on assignment and vehicle routing strategies for ride-sourcing platforms, taking ride-parcel pooling problems into account \cite{ghilas2013integrating,perboli2021simulation,fehn2021ride}. For example, \cite{li2014share} and \cite{chen2016crowddeliver} investigated the optimal order assignment and vehicle routing problem for city logistics networks where people and parcels share taxis. The former formulated a mixed-integer linear programming model to characterize optimal assignment and routing, while the latter proposed a two-phase approach combining offline trajectory mining and online package routing. Both studies demonstrated that their frameworks can successfully fulfill delivery orders within short delivery times through taxi-sharing. \cite{li2016adaptive} proposed a heuristic based on adaptive large neighborhood search to efficiently solve the Share-a-Ride Problem, considering the simultaneous transportation of people and parcels in a taxi network, which was further extended by \cite{li2016share} to account for stochastic travel times and delivery locations. In \cite{manchella2021flexpool}, the authors developed a solution algorithm based on double deep Q-learning networks to determine dispatching policies for passengers and goods on a ride-sharing platform. Simulation results demonstrated that the proposed algorithm leads to better solutions with higher operational efficiency and lower environmental footprints compared to other model-free settings. {\cite{bosse2023dynamic} explored the operational policy to joint on-demand services that combines passenger transportation and goods delivery, where a dedicated fleet is introduced that serves passengers with priority.
A Bayesian optimization approach is employed to find a dynamic priority policy that determines the time-dependent size of the priority fleet.
A similar setting is considered in \cite{fehn2023integrating}, where the authors developed an agent-based simulation approach to investigate the integrated transportation of passengers and parcels with the same fleet of on-demand vehicles assuming that passengers has priority over parcels.
It illustrates that the proposed business model has no negative impacts on the ride-sourcing services and can serve nearly all parcels when the parcel-to-passenger demand ratio is not too large.} In addition to conventional ride-sourcing vehicles, joint parcel and passenger delivery problems have been investigated in the context of autonomous vehicles \cite{beirigo2018integrating}, electric vehicles \cite{lu2022combined}, and a combination of ride-hailing vehicles and electric motorcycles \cite{zhan2023ride}.

Our work differs from the aforementioned works in the following ways: (a)  we consider a novel business model where the flexibility in parcel demand is exploited to enable the integration of flexible parcel delivery services into the existing ride-sourcing systems without incurring extra negative impacts on passengers. This can be realized by only utilizing the idle time of ride-sourcing drivers to pick up and drop off flexible parcel orders while they cruise for new passengers, and transport parcels between zones along with passengers at no additional cost, while prioritizing passengers over parcel in vehicle dispatch. Technically, this requires explicitly modeling the idle time of ride-sourcing drivers by characterizing their status changes as they transition between distinct tasks (e.g., cruising, ride-sourcing, parcel delivery, etc.) across the transportation network, and characterizing how the joint movement of passenger and parcel affects the delivery time, which has not been considered in the existing literature; (b) we consider differentiated pricing strategies to unlock the value of flexibility in parcel demand, enabled by an elastic demand model where ride-sourcing passengers are sensitive to ride fare, waiting time, and trip time, while parcel delivery customers are sensitive to service fare, waiting time, and delivery time. The elastic demand model enables us to derive the optimal spatial pricing for the integrated services, while capturing the heterogeneously elastic demand, where ride-sourcing passengers exhibit a higher degree of impatience compared to parcel delivery customers. Although this is important for evaluating the synergy between ride-sourcing services and parcel delivery services, none of the existing works on integrated services has explicitly considered demand elasticity with respect to service fare or service time.

\section{The Mathematical Model}
Consider a city divided into $M$ zones, and denote the transportation network of the city as a graph $\mathcal{G}(\mathcal{N},\mathcal{A})$, where the set of nodes $\mathcal{N}$ corresponds to zones, and the set of links $\mathcal{A}$ represents the connections between distinct zones that are adjacent to each other.
On the transportation network, an integrated platform offers both ride-sourcing services to transport passengers and delivery services to transport parcels,  using the same fleet of for-hire drivers.  Each ride-sourcing/delivery trip is assigned an origin $i\in\mathcal{N}$ and a destination $j\in\mathcal{N}$. Similar to Uber, Lyft and Didi, ride-sourcing services are on-demand, which has to be served immediately by the closest idle driver\footnote{The main reason for matching the closest idle driver to a passenger is to guarantee that the integration of logistics does not negatively affect passengers. There exists more sophisticated matching strategies where the passenger may be matched to a vehicle with a parcel that has similar OD, but it may also increase the cost for passengers as it leads to longer waiting time before the passenger can be picked up. For simplicity, this paper focuses on matching passengers to the closest idle vehicle.}. However, delivery services can be either on-demand or flexible: in the case of \emph{on-demand delivery},  drivers need to immediately pick up and delivery the goods when the order is dispatched; while in the case of \emph{flexible delivery}, drivers can pick up (or drop off) the goods only when they are idle\footnote{Vehicles not in the process of picking up or delivering passengers or packages that request to be served on-demand.} and close to the origin (or destination), and carry them along with the on-demand orders. Intuitively, on-demand delivery services offer the fastest possible delivery time for customers at a higher cost, whereas flexible services can better utilize the idle time of ride-sourcing drivers and simultaneously transport people and parcel,  creating significant synergy between ride-sourcing and package delivery services and thus substantially reducing the overall operational cost. In this paper, we consider a platform that offers a combination of on-demand ride-sourcing service, on-demand delivery service and flexible delivery service.
In this case, to maintain responsiveness of on-demand services, the flexible delivery service has lower priority, and may be interrupted by ride-sourcing or on-demand delivery services even if the vehicle is already on the way to pickup or drop-off a flexible package.  In addition, when both types of delivery services are available, customers will choose which service to patronage based on the trade-off between delivery time and service fare. The aforementioned dispatching mechanism is illustrated in Figure \ref{fig:process}, and the rest of this section will present a mathematical model that captures the incentives of ride-sourcing passengers, delivery customers, for-hire drivers, and their interactions under the profit-maximizing platform. 

\begin{figure}[htb!]
    \centering
    \includegraphics[width = 0.9\textwidth]{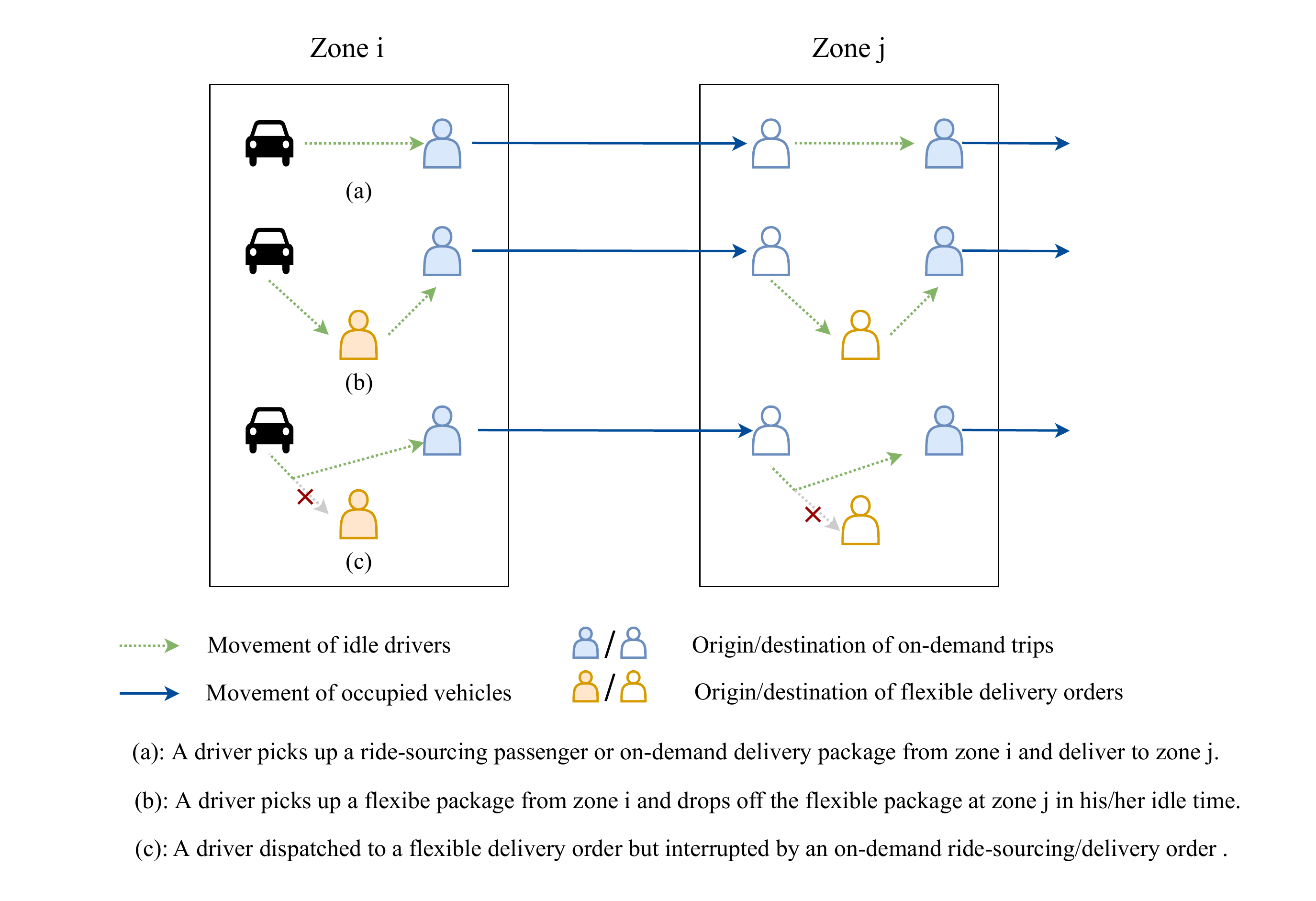}
    \caption{The dispatching of ride orders, flexible delivery orders, and on-demand delivery orders to drivers.}
    \label{fig:process}
\end{figure}

\subsection{Incentives of Ride-Sourcing Passengers}
Passengers choose their travel modes from several available options such as ride-sourcing, public transit and walking, by comparing the respective costs. These options can be regarded as substitute goods for  passengers.
The generalized cost of a passenger traveling from zone $i$ to zone $j$ using ride-sourcing service is defined as a weighted sum of waiting time and trip fare:
\begin{align} \label{eq:demand_p}
    c_{ij}^r = \alpha_r w_i^r+r_i^rt_{ij},
\end{align}
where $w_i^r$ is the average waiting time for passengers in zone $i$ (from sending a request to being picked up), $r_i^r$ is the per-time fare for ride-sourcing trips starting from zone $i$, $t_{ij}$ is the average travel time along the shortest path from zone $i$ to zone $j$, and $\alpha_r$ is the model parameter which represents the passengers' value of time.

Let $F_r$ denote the proportion of potential passengers who choose to take ride-sourcing services.
We treat the cost of other alternative travel options as exogenous and assume that $F_r(c)$ only depends on the average generalized cost $c_{ij}^r$.
Let $\lambda_{ij}^{r,0}$ denote the potential arrival rate of passengers traveling from zone $i$ to zone $j$, then the arrival rate of passengers from zone $i$ to zone $j$ choosing ride-sourcing services is determined by:
\begin{align} \label{demand_ride}
    \lambda_{ij}^r = \lambda_{ij}^{r,0}F_r(c_{ij}^r),
\end{align}
where $F_r(c)$ is an decreasing function satisfying $0<F_r(c)<1$ and $\lim_{c\rightarrow\infty}F_r(c) = 0$.
In (\ref{demand_ride}), the arrival rate $\lambda_{ij}^r$ as a function of $c_{ij}^r$ is the demand function for ride-sourcing services, which delineate the relationship between generalized travel cost and travel demand on the platform.

\subsection{Incentives of Delivery Customers}
Customers choose how to send parcels from origin to destination among several available options, such as on-demand delivery or flexible delivery provided by the integrated platform, or  other outside options\footnote{These options can be treated as substitute goods for customers who want to send packages.}.
The decisions of the customers depend on the respective costs for those options. 
The generalized cost of a customer who chooses to deliver the parcel from zone $i$ to zone $j$ using the flexible delivery services is defined as a weighted sum of waiting time, delivery time, and the delivery fee:
\begin{align}\label{eq:demand_flexible_delivery}
    c_{ij}^{d_f} = \alpha_dw_i^{d_f}+p_d(t_{ij}^{d_f})+r_{ij}^{d_f},
\end{align}
where $w_i^{d_f}$ is the waiting time for flexible delivery customers in zone $i$ (from sending the request to the parcel being picked up),
$t_{ij}^{d_f}$ is the average delivery time for a parcel to be delivered from zone $i$ to zone $j$ by flexible delivery services (from the courier picking up the parcel to dropping off the parcel), $r_{ij}^{d_f}$ is the delivery fare per unit time per parcel from zone $i$ to zone $j$ by flexible delivery services, $\alpha_d$ represents the customer's value-of-time when waiting for pick-up, and $p_d(\cdot)$ is a nonlinear function that represents the customer's disutility associated with the package delivery time\footnote{For on-demand services, the delivery time or trip time $t_{ij}$ is small, e.g., a few minutes, thus we can approximately characterize the disutility of passengers (or customers) associated with  $t_{ij}$ through a linear relation. However, for flexible delivery services, $t_{ij}^{d_f}$ can be much larger (i.e., a couple of hours) and may significantly increase after certain threshold (e.g., one day), thus we choose to use the nonlinear function $p_d(\cdot)$ to characterize the disutility of customers.}. Since delivery customers are less sensitive to the time waiting to be pick-up compared to ride-sourcing passengers, we have $\alpha_d<\alpha_r$.

For on-demand delivery services, the waiting time for customers is the same as the waiting time for passengers (i.e. $w_i^r$ for each zone $i$), and the average delivery time for on-demand delivery services is equivalent to the average travel time along the shortest path from zone $i$ to zone $j$. We assume that the pricing scheme for on-demand delivery services is also the same as ride-sourcing services, then the average generalized costs for on-demand delivery services from zone $i$ to zone $j$ can be defined by
\begin{align} \label{eq:demand_ondemand_delivery}
    c_{ij}^{d_o} = \alpha_dw_i^r+p_d(t_{ij})+r_i^rt_{ij}.
\end{align}

\begin{remark}
We assume that the on-demand ride-sourcing services and on-demand delivery services share the same price. This is similar to Uber and Uber Connect \cite{Uber2023Guide} (an on-demand parcel delivery services provided by Uber), which has similar prices. We believe this is a reasonable assumption because the quality of services provided by the platform and the cost for delivering these services are almost identical  (e.g. same distance, same time and similar fuel consumption) regardless of whether people or parcel is delivered. Such assumption can simplify the notation without loss of essential insights. 
\end{remark}

Let $\lambda_{ij}^{d,0}$ be the  arrival rate of all packages to be delivered from zone $i$ to zone $j$ regardless of the delivery modes, $F_d^k$ be the proportion of packages to be delivered by flexible delivery services ($k=1$) or on-demand delivery services ($k=2$). Customers make their decisions by comparing the costs, thus $F_d^k$ depends on both $c_{ij}^{d_f}$ and $c_{ij}^{d_o}$ when treating the costs for outside options as exogenous.
The arrival rate of flexible delivery customers and on-demand delivery customers from zone $i$ to zone $j$ can be modeled as the following demand functions:
\begin{align}
    \lambda_{ij}^{d_f} &= \lambda_{ij}^{d,0}F_d^1(c_{ij}^{d_f},c_{ij}^{d_o}),\\
    \lambda_{ij}^{d_o} &= \lambda_{ij}^{d,0}F_d^2(c_{ij}^{d_f},c_{ij}^{d_o}),
\end{align}
where $F_d^1(c_{ij}^{d_f},c_{ij}^{d_o})$ is decreasing in $c_{ij}^{d_f}$ and increasing in $c_{ij}^{d_o}$; $F_d^2(c_{ij}^{d_f},c_{ij}^{d_o})$ is decreasing in $c_{ij}^{d_o}$ and increasing in $c_{ij}^{d_f}$. Note that the above demand model includes the well-established logit model as a special case.

\subsection{Incentives of Drivers}

Drivers determine whether to join the ride-sourcing platform or not based on the prospective wages.
Let $q$ denote the average wages for drivers offered by the ride-sourcing platform in the overall transportation network.
Let $F_d$ be the proportion of potential drivers choosing to work for ride-sourcing.
We treat the average wage of other job opportunities for drivers as exogenous and assume that $F_d(q)$ only depends on the average wage of ride-sourcing services.
Let $N_0$ be the number of total drivers in the transportation network, then the total driver supply who choose to work for the ride-sourcing platform is characterized by:
\begin{align}\label{eq:supply}
    N=N_0F_d(q).
\end{align}
where $F_d(q)$ is an increasing and continuous function, and $0\leq F_d(q)\leq 1$.  In (\ref{eq:supply}), the number of drivers as a function of wage $q$ can be viewed as the supply function for the platform, which delineates the relationship between driver supply and driver payment. 

Each ride-sourcing driver has three operating modes:
(a) carrying a passenger or a on-demand parcel,
(b) on the way to pick up the passenger or on-demand parcel, and
(c) cruising with empty seats to look for the next passenger or on-demand parcel\footnote{We do not categorize the drivers carrying flexible parcels or on the way to pick up a flexible parcel into the operating modes, because they  belong to the category of cruising vehicles.}. 
Drivers in the third mode (i.e., cruising for new passengers or on-demand parcel) are treated as idle, and their corresponding waiting time in this status is denoted as $w_i^I$ if they are idle and cruising in zone $i$.
Let $N_i^{c}$ be the average number of drivers carrying a passenger or an on-demand parcel, $N_i^p$ be the average number of drivers on the way to pick up a passenger or an on-demand parcel, and $N_i^I$ be the average number of idle drivers cruising with empty seats. By Little's Law, we have
\begin{align}\label{eq:Nc}
    N_i^c &= \sum_{j=1}^M (\lambda_{ij}^r+\lambda_{ij}^{d_o})t_{ij},\\
    N_i^p &=  w_i^r\sum_{j=1}^M(\lambda_{ij}^r+\lambda_{ij}^{d_o}),\\
    N_i^I &= w_i^I\sum_{j=1}^M (\lambda_{ij}^r+\lambda_{ij}^{d_o}).
\end{align}
where $w_i^I$ is the average waiting time for idle drivers in zone $i$ before receiving the next on-demand ride-sourcing/delivery order. Therefore,  the total number ride-sourcing drivers $N$ should satisfy
\begin{equation}
\label{conservation}
N=\sum_{i=1}^M \sum_{j=1}^M (\lambda_{ij}^r+\lambda_{ij}^{d_o})t_{ij}+\sum_{i=1}^M \sum_{j=1}^M w_i^r (\lambda_{ij}^r+\lambda_{ij}^{d_o})  +  \sum_{i=1}^M \sum_{j=1}^M   w_i^I (\lambda_{ij}^r+\lambda_{ij}^{d_o}).
\end{equation}
Note that vehicles that are on the way to pick up or drop off flexible parcels are already included in the last term of (\ref{conservation}), because they can be interrupted by on-demand services at any moment. 

\subsection{Matching between On-Demand Services and Drivers}
In the demand model for ride-sourcing passengers (\ref{eq:demand_p}) and  on-demand delivery customers (\ref{eq:demand_ondemand_delivery}), the waiting time $w_i^r$ implicitly depends on the number of waiting customers and the number of idle drivers cruising for a new customer on the platform. 
The matching process between on-demand customers and drivers can be characterized by a well-established Cobb-Douglas meeting function \cite{yang2011equilibrium}, which dictates that the matching rate ($m_i^{p-t}$) between on-demand customers and drivers in zone $i$ should satisfy:
\begin{align} \label{eq:matching_ondemand}
    m_i^{p-t} = M_i(N_i^W,N_i^I) = A_i(N_i^W)^p(N_i^I)^q,
\end{align}
where $A_i$ is a scaling parameter that depends on the geometry of the ride-sourcing and the on-demand delivery market, $N_i^W$ is the total number of waiting passengers and on-demand delivery customers in zone $i$, $N_i^I$ is the number of idle drivers in zone $i$, $p$ and $q$ are the elasticity parameters within range $[0,1]$. We note that the matching rate $m_i^{p-t}$ should satisfy that $m_i^{p-t}=\sum_{j=1}^M (\lambda_{ij}^r+\lambda_{ij}^{d_o})$, and by Little's Law, the number of waiting passengers should 
satisfy $N_i^W=w_i^r\sum_{j=1}^M (\lambda_{ij}^r+\lambda_{ij}^{d_o})$.
In the rest of paper, we let $p=1$ and $q = 0.5$, which yields the well-known `square-root law' that  can capture the matching in ride-sourcing market \cite{arnott1996taxi, li2019regulating}.
At the stationary state, the matching rate equals to the arrival rate of ride-sourcing customers and on-demand customers, then we have
\begin{align}
    w_i^r = \frac{L_i}{\sqrt{N_i^I}},
\end{align}
where $L_i$ is a model parameter that depends on the geometry of zone $i$, which can be calibrated for each zone.

\subsection{Matching between Flexible Parcels and Drivers}

In equation (\ref{eq:demand_flexible_delivery}), the customer waiting time for a flexible delivery order to be picked up, denoted as $w_i^{d_f}$, implicitly depends on the matching process between flexible parcels and idle drivers, which further depends on the number of flexible customers, the number of drivers, the successful rate of parcel pick-up and drop-off, and the capacity limit of the vehicles. In this subsection, we model the matching  between flexible parcels and drivers to delineate the dependence of $w_i^{d_f}$ on other endogenous variables.

Since flexible parcel has a lower priority compared to other on-demand service requests,  idle drivers that are dispatched to a flexible delivery order may be interrupted by another on-demand ride-sourcing or package delivery request, and thus fail to pick up or drop off the parcel. In this case, whether the idle driver can successfully pick-up or drop-off a flexible parcel depends on whether the time for drivers to pick up or drop off the parcel is shorter than the waiting time for idle drivers to meet a new ride-sourcing passenger or on-demand delivery customer. Based on this observation, the probability of success can be characterized by the following definition: 

\begin{definition} \label{def:successful_rate}
Assume that the waiting time for drivers to receive an on-demand ride-sourcing/delivery order in zone $i$ ($W_i^I$) is subject to certain distribution with expectation $w_i^{I}$, the waiting time for drivers to receive a flexible delivery order in zone $i$ ($W_i^{dg}$) is subject to certain distribution  with expectation $w_i^{dg}$, the time for a driver in zone $i$ to drop off the flexible parcel ($T_i^g$) is subjective to certain distribution with expectation $t_i^g$, and the time for a driver in zone $i$ to pick up the flexible parcel ($\bar T_i^g$) is subject to certain distribution with expectation $\bar t_i^g$, 
\begin{enumerate}
    \item The successful rate of dropping off a flexible parcel (i.e. the probability for a driver to successfully drop off a flexible parcel) is :
    \begin{align}\label{eq:success_rate_drop}
        p_{i,drop}^{succ} = \mathbb{P}(T_i^g<W_i^I),
    \end{align}
    which captures the probability that the drop-off time is shorter than the driver's cruising time, and can be denoted as $p_{i,drop}^{succ}(t_i^g,w_i^I)$ to emphasize its dependence on $t_i^g$ and $w_i^I$.
    \item The successful rate of picking up a flexible parcel (i.e. the probability for a driver to successfully pick up a flexible parcel) is:
    \begin{align}\label{eq:success_rate_pick}
        p_{i,pick}^{succ} = p_{i,pick}^{succ,w}\cdot p_{i,pick}^{succ,t} =  p_{i,pick}^{succ,w}(\bar t_i^g,w_i^I)\cdot p_{i,pick}^{succ,t}(w_i^{dg},w_i^I)=\mathbb{P}(\bar T_i^g < W_i^I)\cdot\mathbb{P}(W_i^{dg} < W_i^I) ,
    \end{align}
    where $p_{i,pick}^{succ,w}(\bar t_i^g,w_i^I)$ captures the probability that the driver's waiting time to be dispatched to a flexible pick-up order is shorter than the cruising time until matched with a new on-demand order, and $p_{i,pick}^{succ,t}(w_i^{dg},w_i^I)$ captures the probability that the pick-up time is shorter than the driver's cruising time. We have $\lim_{w_i^{dg}\rightarrow \infty}p_{i,pick}^{succ,w} = 0$ and $\lim_{\bar t_i^g\rightarrow \infty}p_{i,pick}^{succ,t} = 0$.
\end{enumerate}
\end{definition}

\begin{remark}
Note that in equation (\ref{eq:success_rate_drop}) and (\ref{eq:success_rate_pick}), $t_i^g$ is exogenous variables that depends on the geometry of the zone\footnote{Note that $T_i^g$ or $\bar{T}_i^g$ are defined as the time needed to pick up or drop off the package if there were no interruptions from on-demand orders, which mainly depend on the distribution of the demand and the zone geometry,  but do not depend on whether the trip is interrupted by on-demand order or not.} and the average speed of the drivers, but $w_i^I$, $w_i^{dg}$ and $\bar t_i^g$ are endogenous variables that depend on the arrival rate of passengers/delivery orders and the number of available idle drivers. These endogenous variables are intimately correlated, and we will characterize their relations hereinafter. 
\end{remark}

Since drivers carry flexible packages in their trunks\footnote{In contrast, on-demand delivery orders can be carried in passenger seats because there is no passenger in the car when the parcel is being delivered. This may raise concerns regarding safety, insurance, and passenger comfort. On the safety side, the platform typically sets some restrictions on the items allowed for delivery services. Regarding insurance, platform can offer insurance coverage for customers who are willing to pay extra insurance rates, which can be treated as exogenous factors. Regarding comfort, most of the on-demand packages are small and are transported while there is no other passengers in the cars, which should not impose significant discomfort to others.}, whether idle drivers can successfully pick up a flexible parcel not only depends on the success rate characterized by Definition \ref{def:successful_rate}, but also depends on the vehicle's carrying capacity, denoted by $C_a$. To evaluate the impact of the capacity limit, we will model the transition of driver's position and number of carried package as a \emph{Continuous-Time Markov Chain} (CTMC). The CTMC has three essential attributes, including the states, the state transition probability matrix, and the amount of time spent in each state. For a transportation network with $M$ zones and the delivery capacity $C_a$, the states of the CTMC is formulated as the set $S_c=\{(z,n)|z\in \{1,\dots,M\}, n\in\{0,\dots,C_a\}\}$, where the first element captures the current zone the driver is in, and the second element captures the number of packages the driver is currently carrying. The transition probability for a driver in state $(z,n)$ to transit to state $(z',n')$ will be denoted as $P^c_{(z,n)(z',n')}$.
The transition from $z$ to $z'$ is determined by the geographical movement of drivers, while the transition from $n$ to $n'$ depends on whether the driver successfully pick up/drop off a flexible package or not.
We first characterize the transition from $z$ to $z'$ by tracking the vehicle movement. For simplicity, we ignore driver repositioning and assume that 
any idle driver cruising in zone $i$ will not leave zone $i$ until he/she is matched to an on-demand order (either a ride-sourcing request or an on-demand delivery request) with a destination other than $i$\footnote{We do not consider repositioning because after the vehicle drops off the passenger, it can actually utilize its idle time in the same zone to pick up or drop off packages, in which case the idle time of the drivers can be effectively utilized even if the driver stays. A more detailed model that characterizes the repositioining decisions of idle drivers in the integrated platform is left for future work. }. 
In this case, for any driver within zone $i$, his/her next destination is contingent on the destination of the dispatched on-demand order, thus the probability for the driver to transit to zone $j$ after dispatched an on-demand order is equivalent to the probability to be matched with an on-demand order with destination $j$.
Let $P_{ij}$ denote the probability that a driver in zone $i$ move to zone $j$ after dispatched an on-demand order, i.e. the driver picks up a passenger/on-demand parcel with origin $i$ and destination $j$, then we have
\begin{align} \label{eq:define_P_ij}
    P_{ij} = \frac{\lambda_{ij}^r+\lambda_{ij}^{d_o}}{\sum_{k=1}^M\left(\lambda_{ik}^r+\lambda_{ik}^{d_o}\right)}.
\end{align}

Next, we characterize the transition from $n$ to $n'$ by tracking the number of packages in the vehicles. Note that the transitions from $n$ to $n'$  can be divided into two categories: (1) transition to a state with the same number of carried flexible packages; and (2) transition to a state with distinct number of carried flexible packages.
Which case it belongs to is determined by whether a driver can successfully pick up or drop off a package or not.
For simplicity, we assume that if a driver arrives at zone $i$ carrying a flexible package with destination $i$, then he/she will prefer to drop off the package first rather than waiting to pick up a new flexible order\footnote{This assumption is practical as in most instances, dropping off can be performed immediately without waiting, while at the same time it can both fulfill the delivery task and release space for future parcels. However, we acknowledge that if a flexible delivery order is just on the way a driver moving to the drop-off location, and the vehicle happens to have free capacity, then it should pick it up first before continuing to the drop-off location.
To capture this possibility of pooling, we can add an additional term to (\ref{eq:p_pick}) defined as the probability of the existence of a pick-up order that is exactly on the shortest route of the flexible order to be dropped off in the same zone, multiplied by $p_{i,pick}^{succ}$. This additional parameter can be calibrated using more detailed parcel demand data, but it would not affect the proposed model or solution framework in this paper.}.
In this case,  the probability for a driver in zone $i$ with $n$ flexible packages to be dispatched to drop off a flexible package is equal to the probability for the driver to carry at least one flexible parcel with destination $i$ (note that otherwise there is no need to drop off), which can be characterized as:
\begin{align}
\label{definition_flex}
    p_{flex,i}^n = 1-\left(1-\frac{\sum_{k=1}^M \lambda_{ki}^{d_f}}{\sum_{j=1}^M\sum_{k=1}^M\lambda_{kj}^{d_f}}\right)^n,
\end{align}
where $\frac{\sum_{k=1}^M \lambda_{ki}^{d_f}}{\sum_{j=1}^M\sum_{k=1}^M\lambda_{kj}^{d_f}}$ denotes the probability that a parcel has a destination in zone $i$. By combining (\ref{definition_flex}) with (\ref{eq:success_rate_drop}) and (\ref{eq:success_rate_pick}), we can derive the probabilities for a idle driver with $n$ flexible packages to successfully pick up the package at zone $i$ (denoted by $p_{i,pick}^n$) or successfully drop off the package at zone $i$ (denoted by $p_{i,drop}^n$), which can be written as:
\begin{align} \label{eq:p_pick}
    p_{i,pick}^n &= \begin{cases}
    p_{i,pick}^{succ} & n=0 \\
    p_{i,pick}^{succ}(1- p_{flex,i}^n)& 0<n<C_a \\
    0 & n=C_a
    \end{cases},\\\label{eq:p_drop}
    p_{i,drop}^n &=\begin{cases}
    0 & n=0 \\
    p_{i,drop}^{succ}p_{flex,i}^n & n>0 \\
    \end{cases}.
\end{align}

Equations (\ref{eq:p_pick}) and (\ref{eq:p_drop}) indicate that when the driver has no flexible parcels on his/her vehicle (i.e., $n=0$), he/she can only be dispatched to pick up new parcels with successful probability $p_{i,pick}^n$. When the driver has flexible vehicles in his/her vehicle, the driver will be dispatched to drop off if there is any parcel on the vehicle with destination $i$. The proportion of drivers that can successfully drop off a parcel is $p_{i,drop}^{succ}p_{flex,i}$, which equals the proportion of drivers that has a parcel to be dropped off times the successful rate to drop off the parcel. When the drivers can do both pick-up or drop-off (i.e., $0<n<C_a$), $p_{flex,i}^n$ of them has a parcel with destination $i$ and is dispatched to drop off, and $1-p_{flex,i}^n$ of them can be dispatched to pick up new parcels. In this case, the proportion of idle drivers that can successfully pick up a new parcel is $p_{i,pick}^{succ}(1- p_{flex,i}^n)$. When the driver has already reached the carrying capacity (i.e., $n=C_a$), he/she can only be dispatched to drop off and $p_{i,pick}^n$ = 0.

So far, we have discussed the transition from $z$ to $z'$, and the transition from $n$ to $n'$, respectively. By integrating these two, we can formulate the transition process from state $(z,n)$ to $(z',n')$. In particular, we have the following two cases:
\begin{itemize}
    \item If $n'\neq n$, it indicates that the driver successfully pick up or drop off a flexible package before being matched with an on-demand order and $z'=z$.
    Therefore, $n'$ must be equal to $n+1$ or $n-1$, which represents that the driver successfully pick up a flexible package or successfully drop off a flexible package, respectively. From (\ref{eq:p_pick}) and (\ref{eq:p_drop}), we have:
    \begin{align} 
        P_{(z,n)(z,n+1)}^c &= p_{z,pick}^n,\quad n = 0,\dots,C_a-1;\\
        P_{(z,n)(z,n-1)}^c &= p_{z,drop}^n,\quad n = 1,\dots,C_a.
    \end{align}

    \item If $n' = n$, then the driver must be dispatched to serve an  on-demand order before he/she picks up or drops off a flexible package.
    The probability for a state $(z,n)$ to transit to state $(z',n)$ is:
    \begin{align}
        P_{(z,n)(z',n)}^c=(1-p_{z,pick}^n-p_{z,drop}^n)P_{zz'},
    \end{align}
    where $1-p_{z,pick}^n-p_{z,drop}^n$ is the probability that a driver was interrupted by an on-demand order when he/she is on the way to pick up or drop off a flexible package, and $P_{zz'}$ is the probability that a driver in zone $z$ is dispatched an on-demand order with destination $z'$, as defined in (\ref{eq:define_P_ij}).
\end{itemize}
Overall, the transition probabilities of the CTMC can be summarized as:
\begin{align}\label{eq:tran_prob_CTMC}
    P_{(z,n)(z',n')}^c=\begin{cases}
    p_{z,pick}^n & z'=z,\ n = 0,\dots,C_a-1,\ n'=n+1 \\
    p_{z,drop}^n & z' = z,\ n = 1,\dots,C_a,\ n' = n-1 \\
    (1-p_{z,pick}^n-p_{z,drop}^n)P_{zz'} & n' = n\\
    0 & \text{otherwise}
    \end{cases}.
\end{align}
Taking a transportation network with 2 zones and $C_a=2$ as an example, then the corresponding CTMC is shown in Figure \ref{fig:CTMC}.

\begin{figure}[!htb]
    \centering
    \includegraphics[width=0.3\textwidth]{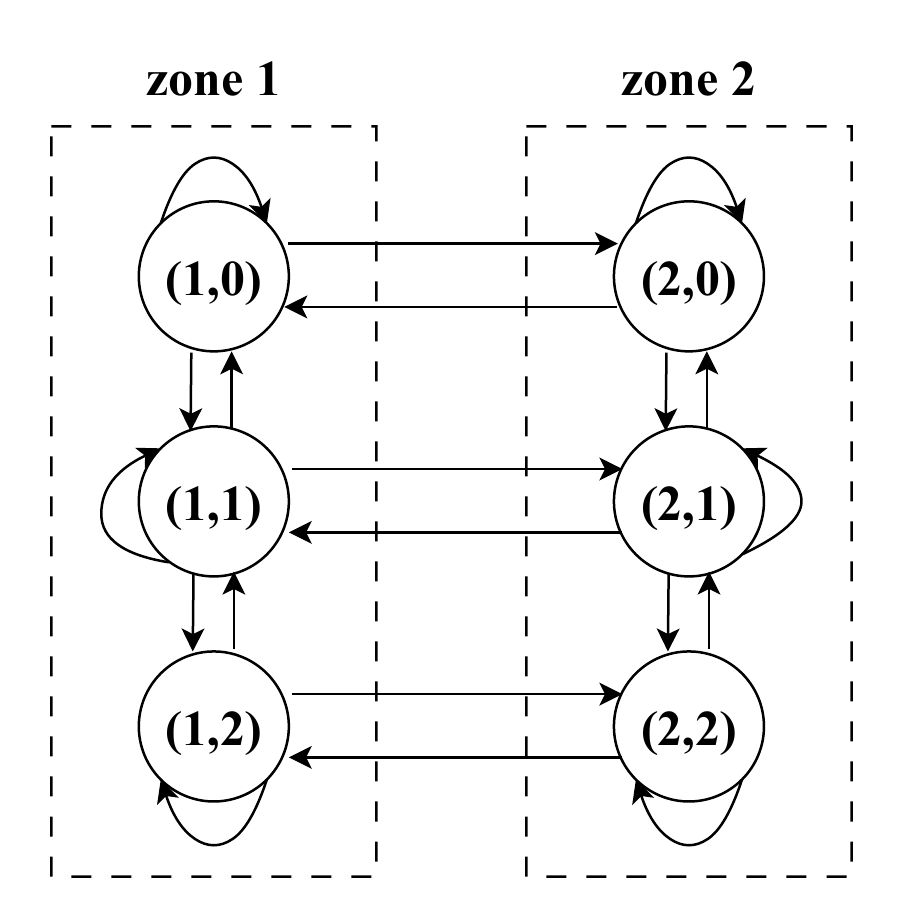}
    \caption{The continuous-time Markov chain for a transportation network with $M=2$ and $C_a=2$.}
    \label{fig:CTMC}
\end{figure}

\begin{remark}
    The proposed model can be  extended to considering ``bundled'' parcel orders where drivers pick up more than one parcels at the origin with different destinations.  The matching function between drivers and bundled parcels remains essentially the same as the matching function for single parcels, with the only difference being that bundled parcels occupy more in-vehicle space. When considering bundled parcels, a driver who successfully picks up a flexible delivery can transition to any state within the same zone that has a larger number of carried parcels (e.g., the driver in state $(1,0)$ in Figure \ref{fig:CTMC} can directly transition to $(1,1)$ and $(1,2)$). For sake of simplification, we only consider solo parcel pick up in this paper.
\end{remark}

Next, we will characterize the average amount of time a driver spends in each state, so that the limiting probability of the CTMC can be derived. 
For a driver in state $(z,n)$ who transits to the same zone with different number of carried flexible packages, the amount of time spent in $(z,n)$ equals to the time period from the driver reaches  $(z,n)$ to he/she successfully pick up or drop off a flexible package in zone $z$ before being matched to an on-demand order.
For a driver who transits to a state with the same number of carried packages, the transition is triggered by being matched to an on-demand order. Therefore,  the time spent in $(z,n)$ equals to the idle time for the driver in zone $z$. 
Note that the CTMC requires the amount of time a driver in a state is independent of which state is next visited, thus we approximately estimate the average time an idle driver spend in state $(z,n)$ by assuming that it is exponentially distributed with rate $\nu_{(z,n)} = 1/\bar t_{(z,n)}$, where $\bar t_{(z,n)}$ denotes the average time spend in state $i$  while drivers are idle and avaiable for the next on-demand order. In this case, for the state $(z,n)$, we have
\begin{align}\label{eq:spend_time_CTMC}
    \bar t_{(z,n)}=\dfrac{1}{\nu_{(z,n)}} = P^c_{(z,n)(z,n-1)}t_z^g+ P^c_{(z,n)(z,n+1)}(w_z^{dg}+\bar t_z^g) 
    +\sum_{z'\in \{1,\dots,M\}} P^c_{(z,n)(z',n)}w_z^I,
\end{align}
where $P_{(z,0)(z,-1)}^c$ and $P_{(z,C_a)(z,C_a+1)}^c$ are defined to be zero, and the three terms in the right-hand side of (\ref{eq:spend_time_CTMC}) represent the average time spend in state $(z,n)$ if the vehicle successfully drops off a flexible parcel, successfully picks up a flexible parcel, or exists zone $z$ without performing any flexible delivery order before entering another state, respectively. 
Given $\bar t_{(z,n)}$, we can derive the limiting probabilities for the CTMC, denoted as $P_{(z,n)}^c=\lim_{t\rightarrow \infty} P_{(z',n')(z,n)}^c(t)$, where $(z',n')$ is any state and $P_{(z',n')(z,n)}^c(t)$ denotes the probability that the CTMC presently in state $(z',n')$ will be in state $(z,n)$ after an additional time $t$.
In this case, the long-run proportion of time the process is in state $(z,n)$ is give by the following proposition:
\begin{proposition}\label{prop:limit_prob}
$\{P^c_{(z,n)}\}$ are the unique nonnegative solution of
\begin{align} \label{eq:limit_prob}
\nu_{(z,n)}P_{(z,n)}^c=\sum_{(z',n')}\nu_{(z',n')}P_{(z',n')}^cP^c_{(z',n')(z,n)}\quad \text{and}\quad \sum_{(z,n)} P_{(z,n)}^c = 1,
\end{align}
and $P_{(z,n)}^c$ also equals the long-run proportion of time the CTMC is in state $(z,n)$.
\end{proposition}

The proof of Proposition \ref{prop:limit_prob} can be found in \cite[p.~251]{ross1996stochastic}. It provides the limiting probability of the CTMC, which is essential for characterizing the average time for drivers to pick up a flexible order in zone $i$, denoted as $\bar t_i^g$, the average driver waiting time for receiving a flexible order in zone $i$, denoted as $w_i^{dg}$, and the customer waiting time for a flexible delivery order to be picked up in zone $i$, denoted as $w_i^{d_f}$.

To characterize  $\bar t_i^g$, we first note that $\bar t_i^g$ depends on the average number of drivers that are available to be dispatched a pick-up order, denoted by $\bar N_i^I$. Among all the idle drivers in zone $i$, there are two kinds of drivers that can not be matched with a flexible pick-up order: (i) the drivers on the way to drop off a flexible package; and (ii) the drivers that are full of flexible packages and can not pick up more due to the capacity limitation.
Thus the number of the available idle drivers for flexible package pick-up can be formulated as:
\begin{align} \label{eq:def_bar_N_I}
    \bar N_i^I = N_i^I-t_i^g\sum_j \lambda_{ji}^{d_f}-N_{(i,C_a)}^I(1-p_{flex,i}^{C_a}),
\end{align}
where the $t_i^g\sum_j \lambda_{ji}^{d_f}$ captures  the average number of drivers that on the way to drop off a flexible package; $N_{(i,C_a)}^I(1-p_{flex,i}^{C_a})$ captures the average number of idle drivers that are full of flexible packages and can not be dispatched for pick-up;
$N_{(i,C_a)}^I$ denotes the average number of idle drivers for state $(i,C_a)$. Note that for any state $(z,n)$, the average number of idle drivers in $(z,n)$, denoted as $N_{(z,n)}^I$, can be obtained from the limiting probability of the CTMC:
\begin{align}
    N_{(z,n)}^I=N_z^I\frac{P_{(z,n)}^c}{\sum_{n'} P_{(z,n')}^c},
\end{align}
where $\frac{P_{(z,n)}^c}{\sum_{n'} P_{(z,n')}^c}$ is the conditional probability for an idle driver to carry $n$ flexible packages when he/she is in zone $z$. Similar to (\ref{eq:matching_ondemand}), $\bar t_i^g$ can be characterized by Cobb-Douglas-type matching function. Let $p = 1$ and $q = 0.5$, at the stationary state the average time for drivers to pick up a flexible order in zone $i$ can be derived as::
\begin{align} \label{eq:def_bar_t_i_g}
    \bar t_i^g = \frac{L_i}{\sqrt{\bar N_i^I}}
\end{align}


To characterize $w_i^{dg}$, we note that $w_i^{dg}$ and $\bar t_i^g$ can be related by the Little's Law through the arrival rate of parcels to be picked up, which can be denoted as $\bar \lambda_{ij}^{d_f}$. However, note that $\bar \lambda_{ij}^{d_f}$ not only consists of  newly arrived flexible parcel, but also consists of flexible orders that were failed to be picked up and in all previous rounds. 
Since the probability for a parcel pickup in zone $i$ to be failed is $1-p_{i,pick}^{succ}$,  the overall arrival rate of parcels to be picked up in zone $i$ and delivered to zone $j$  can be written as
\begin{align}
    \bar \lambda_{ij}^{d_f} =\lim_{n\rightarrow\infty} \lambda_{ij}^{d_f}\left(1+(1-p_{i,pick}^{succ})+(1-p_{i,pick}^{succ})^2+\cdots+(1-p_{i,pick}^{succ})^n\right) = \frac{\lambda_{ij}^{d_f}}{p_{i,pick}^{succ}},
\end{align}
where the summation term corresponds to the accumulation of failed orders over time. Therefore, based on Little's Law, we have
\begin{align} \label{eq:def_bar_N_I_2}
    \bar N_i^I=\frac{w_i^{dg}\sum_{j=1}^M \lambda_{ij}^{d_f}}{p_{i,pick}^{succ}}
\end{align}

To characterize $w_i^{d_f}$, 
we note that since the successful rates for distinct drivers in the same zone with the same number of carried packages are independent and identical, the number of drivers that are carrying $n$ flexible parcels and  can still successfully pick up a flexible parcel in zone $z$ is subject to a Binomial distribution $B(N_{(z,n)}^I,p_{z,pick}^n)$. Therefore, the average number of idle drivers who can successfully pick up a flexible package, denoted by $N_i^{Ig}$, is given by
\begin{align}\label{def:NiIgs}
    N_i^{Ig} = \sum_nN_{(i,n)}^Ip_{i,pick}^n=\sum_nN_i^I\frac{P_{(i,n)}^c}{\sum_{n'} P_{(i,n')}^c}p_{i,pick}^n,
\end{align}
where $p_{i,pick}^n$ is the probability for a driver in zone $i$ with $n$ carried flexible packages to successfully pick up a new flexible package, defined in (\ref{eq:p_pick}).
Similar to (\ref{eq:matching_ondemand}),  the matching rate of flexible delivery services can be characterized by the Cobb-Douglas-type matching function.
Let $p = 1$ and $q = 0.5$, at the stationary state the average waiting time for flexible delivery customers in zone $i$ can be derived as:
\begin{align}
    w_i^{d_f} = \frac{L_i}{\sqrt{N_i^{Ig}}},
\end{align}

\subsection{Delivery Time for Flexible Delivery Customers}
In equation (\ref{eq:demand_flexible_delivery}), the delivery time for flexible delivery services, denoted as $t_{ij}^{d_f}$, is an endogenous variable that depends on the arrival rates of on-demand services, the travel time between different zones, the cruising time of idle drivers, and the probability of success for dropping off flexible parcels. In this subsection, we characterize the dependence of $t_{ij}^{d_f}$ on these factors. 

Note that a flexible order may be successfully dropped off to the destination zone after several trials: if the driver fails to drop off the parcel and takes another on-demand order, he/she will try to drop it off when he/she visits zone $i$ again\footnote{We comment that the possibility that zone $i$ is not visited again is negligibly small. Let $p_j^i$ be the probability that a driver in zone $j$ moves to zone $i$, then the probability for a driver in zone $k$ to never visit zone $i$ is
        $\bar p_i = \lim_{n\rightarrow \infty} \sum_{j_1\neq i}p_k^{j_1}\sum_{j_2\neq i}p_{j_1}^{j_2}\sum_{j_3\neq i}p_{j_2}^{j_3}\cdots\sum_{j_n\neq i}p_{j_{n-1}}^{j_n} = 0$.}. In this case, the average time it takes to successfully deliver a flexible order $t_{ij}^{d_f}$ crucially depends on how much time it takes for a drivers to travel from zone $i$ to zone $j$, and if the first trial of delivery is unsuccessful, how much time it takes for the driver to return to $j$ and try again. To capture this, we let $T_{ij}$ denote the time for drivers to travel from zone $i$ and reach zone $j$ ($j\neq i$) for the first time.
Let $T_{ii}$ denote the time for drivers to travel from zone $i$ and go back to zone $i$ for the first time.
Let $\mathbb{E}[\cdot]$ denote the expectation of a random variable, e.g., both $T_{ij}$ and $T_{ii}$ are random variables with expectations $\mathbb{E}[T_{ij}]$ and $\mathbb{E}[T_{jj}]$, respectively.
The following lemma gives the average delivery time for a flexible parcel to be delivered from zone $i$ to zone $j$:

\begin{lemma} \label{lemma:delivery_time}
The average delivery time for a flexible parcel with origin $i$ and destination $j$ is:

\begin{align} \label{eq:t_ij^g}
    t_{ij}^{d_f} = \mathbb{E}[T_{ij}] + \frac{1-p_{j,drop}^{succ}}{p_{j,drop}^{succ}}\mathbb{E}[T_{jj}].
\end{align}
\end{lemma}

\begin{proof}
The probability to successfully drop off the package at zone $i$ is independent each time the driver arrives at zone $i$.
For instance, the driver drops off the parcel successfully when he/she reaches zone $j$ for the first time with probability $p_{j,drop}$, and successfully drops off the parcel when he/she arrives at zone $j$ for the second time with the probability $(1-p_{j,drop})p_{j,drop}$.
Therefore, we have
\begin{align}
\label{expression_d_f}
    t_{ij}^{d_f} =& \mathbb{E}\left[\lim_{n\rightarrow\infty}\quad p_{j,drop}^{succ}T_{ij}+(1-p_{j,drop}^{succ})p_{j,drop}^{succ}(T_{ij}+T_{jj})+\cdots+(1-p_{j,drop}^{succ})^np_{j,drop}^{succ}(T_{ij}+nT_{jj})\right] \\
    &= \mathbb{E}[T_{ij}]+\mathbb{E}[T_{jj}]p_{j,drop}^{succ}\cdot\lim_{n\rightarrow \infty}\left((1-p_{j,drop}^{succ})+2(1-p_{j,drop}^{succ})^2+\cdots+n(1-p_{j,drop}^{succ})^n\right)  \\
    &= \mathbb{E}[T_{ij}]+\mathbb{E}[T_{jj}]p_{j,drop}^{succ}\cdot\lim_{n\rightarrow \infty} \frac{(1-p_{j,drop}^{succ})-(n+1)(1-p_{j,drop}^{succ})^{n+1}+n(1-p_{j,drop}^{succ})^{n+2}}{{p_{j,drop}^{succ}}^2} \\
    &= \mathbb{E}[T_{ij}] + \frac{1-p_{j,drop}^{succ}}{p_{j,drop}^{succ}}\mathbb{E}[T_{jj}],
\end{align}
which completes the proof.
\end{proof}
Lemma \ref{lemma:delivery_time} indicates that the average delivery time for flexible parcels depends on the expected time the driver spends to first arrive at the destination, and the expected time the driver spends to go back to the destination again. Next we show how to derive the expected value of $T_{ij}$ ($i\neq j$) and $T_{jj}$ for any zone $i,j$.

The movements of drivers only depend on the current zone he/she stays in, but are not related to his/her previous trajectories. Therefore, the trajectories of drivers can be modeled as a Markov chain.
The state of the Markov chain is the position of the drivers. A transition from state $i$ to state $j$ indicates that the driver in zone $i$ is matched to an on-demand order and moves to the destination $j$, $i,j = 1,\dots,M$.
The probability for a driver in state $i$ to transit to state $j$ is equivalent to the probability to be matched with an on-demand order with destination $j$, i.e., $P_{ij}$ defined in (\ref{eq:define_P_ij}).
Let $\mathbf{P}$ be the matrix of one-step transition probabilities for drivers, then we have
\begin{align}
    \mathbf{P} = 
    \begin{bmatrix}
    P_{11} & P_{12} & \cdots & P_{1M} \\
    P_{21} & P_{22} & \cdots & P_{2M} \\
    \vdots & \vdots & \ddots & \vdots \\
    P_{M1} & P_{M2} & \cdots & P_{MM} 
    \end{bmatrix}\label{transition_matrix}
\end{align}
where $\sum_{j=1}^M P_{ij} = 1$.
Assume that a driver transits to state $i$ when he/she finishes an on-demand order with destination of zone $i$.
When a driver arrives at a zone and drops off the passenger or on-demand parcel, he/she first cruises for the next passenger or on-demand parcel for an average time of $w_i^Ir$, and then travel to the destination for an average time of $t_{ij}$. Therefore  the average time for a driver to transit from state $i$ to state $j$, denoted by $s_{ij}$, is
\begin{align}
\label{definiction_s}
    s_{ij} = w_i^I+t_{ij}.
\end{align}
Given $s_{ij}$,  the probability distributions of $T_{ij}$ when $i\neq j$ can be written as:
    \begin{align} \label{eq:def_T_ij}
    T_{ij} =
    \begin{cases}
    s_{i1}+ T_{1j} & \text{with probability }P_{i1}\\
    \cdots & \\
    s_{ij} & \text{with probability }P_{ij}\\
    \cdots & \\
    s_{iM}+ T_{Mj} &  \text{with probability }P_{iM}
    \end{cases},
    \end{align}
The derivation of (\ref{eq:def_T_ij}) is straightforward: if a driver is initially in zone $i$, then the probability of traveling to zone $k$ in the next trip is $P_{ik}$. If the driver indeed travels to zone $k$, it will take an average time of $s_{ij}$ before he/she arrives at zone $k$, plus an average time of $T_{kj}$ before he/she arrives at zone $j$ for the first time from zone $k$. The only exception is the case of $k=j$, where the vehicle directly travel to zone $j$ in the next step. 



Next, we will calculate the expectation of $T_{ij}$ when $i\neq j$. Let $\mathbb{E}\left[\mathbf{T}_j\right] = (\mathbb{E}[T_{1j}],\dots,\mathbb{E}[T_{Mj}])^T$ without $\mathbb{E}[T_{jj}]$ be the vector of average travel time for a driver from each zone (other than zone $j$) to first arriving zone $j$.
Let $ \mathbf{\tilde P}$ be a truncation of $\mathbf{P}$ that removes the $j^{th}$ row and the $j^{th}$ column.
Let $ \mathbf{\hat P}$ be a truncation of $\mathbf{P}$ that removes the $j^{th}$ row.
Let $\mathbf{S}$ denote the $M\times M$ transition time matrix where $\mathbf{S}(i,j) = s_{ij}$.
Based on (\ref{eq:def_T_ij}), the expected value of $T_{ij}$ can be calculated as
\begin{align}  
    \mathbb{E}\left[\mathbf{T}_j\right] = \text{diag}\left(\mathbf{\hat P} \mathbf{S}^T\right)+\mathbf{\tilde P} \mathbb{E}\left[\mathbf{T}_j\right],
\end{align}
where $\text{diag}\left(\mathbf{\hat P} S^T\right)$ is a column vector of the diagonal element of $\mathbf{\hat P} \mathbf{S}^T$.
Therefore, average travel time for a driver to first arriving zone $j$ is given by:
\begin{align} \label{eq:stationary_T_ij}
    \mathbb{E}\left[\mathbf{T}_j\right] = (\mathbf{I}-\mathbf{\tilde P})^{-1}\text{diag}\left(\mathbf{\hat P} \mathbf{S}^T\right).
\end{align}

Note that the above expression does not consider $\mathbb{E}[T_{ij}]$ when $i=j$. To separately capture  $\mathbb{E}[T_{jj}]$, let $\mathbf{P}_j$ denote the $j^{th}$ row of the matrix of transition probabilities $\mathbf{P}$, let $\mathbf{\tilde P}_j$ denote the truncation of $\mathbf{P}_j$ without the $j^{th}$ element.
Let $\mathbf{S}_j$ be the $j^{th}$ row of transition time matrix $\mathbf{S}$.
Based on (\ref{eq:def_T_ij}), we have
\begin{align} \label{eq:stationary_T_jj}
    \mathbb{E}[T_{jj}] = \mathbf{P}_j\mathbf{S}_j^T+\mathbf{\tilde P}_j \mathbb{E}\left[\mathbf{T}_j\right].
\end{align}

Overall, the equations (\ref{expression_d_f}), (\ref{eq:stationary_T_ij}) and (\ref{eq:stationary_T_jj}) have delineated the relations between $t_{ij}^{d_f}$ and other endogenous decision variables.

\subsection{Profit Maximization of the Platform}
The integrated platform determines the ride fare $r_i^r$ for each zone $i$, the delivery fare $r_{ij}^{d_f}$ for flexible delivery services from zone $i$ to zone $j$, and the driver wage $q$, in order to maximize the profits subject to the market equilibrium constraints.
This can be formulated as the following optimization problem:
\begin{equation}
\label{optimalpricing_trip}
 \hspace{-5cm} \mathop {\max }\limits_{{\bf r}^r, {\bf r}^{d_f}, q} \quad \sum_{i=1}^M \sum_{j=1}^M r_i^rt_{ij}(\lambda_{ij}^r+\lambda_{ij}^{d_o}) + r_{ij}^{d_f}\lambda_{ij}^{d_f}-N_0F_d(q)q
\end{equation}
\begin{subnumcases}{\label{constraint_optimapricing}}
{
\lambda_{ij}^r  = {\lambda^{r,0}_{ij}}{F_r}\left(\alpha_r {w^r_i}(N^I_i)  +  r_i^rt_{ij}\right)} \label{demand_constraint_p}\\
\lambda_{ij}^{d_f} = \lambda_{ij}^{d,0}F_d^1(\alpha_dw_i^{d_f}(N_i^{Ig}) +p_d (t_{ij}^{d_f})+r_{ij}^{d_f},\alpha_dw_i^{r}(N_i^I) +p_d (t_{ij})+r_i^rt_{ij}) \label{demand_constraint_g1} \\
\lambda_{ij}^{d_o} = \lambda_{ij}^{d,0}F_d^2(\alpha_dw_i^{d_f}(N_i^{Ig}) +p_d (t_{ij}^{d_f})+r_{ij}^{d_f},\alpha_dw_i^{r}(N_i^I) +p_d (t_{ij})+r_i^rt_{ij}) \label{demand_constraint_g2} \\
w_i^r(N_i^I)\leq w_{max} \label{upper_waiting}\\
N_i^I=w_i^I\sum_{j=1}^M (\lambda_{ij}^r+\lambda_{ij}^{d_o}) \label{demand_constraint2}\\
  N_i^{Ig} = \sum_nN_{(i,n)}^Ip_{i,pick}^n=\sum_nN_i^I\frac{P_{(i,n)}^c}{\sum_{n'} P_{(i,n')}^c}p_{i,pick}^n \label{delivery_supply_constraint}  \\
  t_{ij}^{d_f} = \mathbb{E}[T_{ij}] + \frac{1-p_{j,drop}^{succ}}{p_{j,drop}^{succ}}\mathbb{E}[T_{jj}]  \label{delivery_time_equation}\\
  {N_0}{F_d}\left(q\right)  =\sum_{i=1}^M \sum_{j=1}^M (\lambda_{ij}^r+\lambda_{ij}^{d_o})t_{ij}+\sum_{i=1}^M \sum_{j=1}^M w_i^r (\lambda_{ij}^r+\lambda_{ij}^{d_o})  +  \sum_{i=1}^M \sum_{j=1}^M   w_i^I (\lambda_{ij}^r+\lambda_{ij}^{d_o}) \label{supply_conservation_const}  \\
  (\ref{eq:success_rate_drop})-(\ref{eq:stationary_T_jj})\label{constraint_CTMC}
\end{subnumcases}
where $\mathbf{r}^r=(r_1^r,r_2^r,\dots,r_M^r)$, and $\mathbf{r}^{d_f} = \begin{bmatrix}
r_{11}^{d_f} & \cdots & r_{1M}^{d_f} \\
\vdots & \ddots & \vdots \\
r_{M1}^{d_f} & \cdots & r_{MM}^{d_f}
\end{bmatrix}$.
The decision variables are the ride fare per unit of time charged to passengers ${\bf r}^r$, the average delivery fare of flexible delivery parcels under different O-D pairs ${\bf r}^{d_f}$, and the average wages paid to drivers per unit of time $q$\footnote{It is equivalent for the platform to determine the wage on a per-unit-time basis and per-order basis in the static model.}. The objective function (\ref{optimalpricing_trip}) captures the platform profits which equals the total revenues from ride-sourcing services ($\sum_{i=1}^M\sum_{j=1}^Mr_i^rt_{ij}\lambda_{ij}^r$) and both types of delivery services ($\sum_{i=1}^M\sum_{j=1}^Mr_i^rt_{ij}\lambda_{ij}^{d_o}+r_{ij}^{d_f}\lambda_{ij}^{d_f}$) minus the total wages paid to the drivers $N_0F_d(q)q$.
 Constraints (\ref{demand_constraint_p})-(\ref{demand_constraint_g2}) specify the demand of ride-sourcing passengers, flexible delivery customers, and on-demand delivery customers, respectively.
Constraint (\ref{upper_waiting}) imposes an upper bound on the waiting time for on-demand customers (apply for both passengers or on-demand delivery customers)\footnote{For some remote zones with small demand, the platform may find it more profitable not to dispatch drivers to those zones. We impose the upper-bound for waiting time to guarantee that each zone can maintain good service quality}.
Constraints (\ref{demand_constraint2}) and (\ref{delivery_supply_constraint}) capture the number of idle drivers for on-demand services (including ride-sourcing services and on-demand delivery services) and flexible delivery services, respectively.
Constraint (\ref{supply_conservation_const}) ensures that the driver supply (\ref{eq:supply}) is consistent with the sum of drivers in each operating modes (\ref{conservation}).
Constraints (\ref{constraint_CTMC}) capture the stationary distribution of the drivers in both location and number of carried parcels.
Throughout the rest of this paper, we  will use the logit model  to characterize the choices of passenger demand, delivery customer demand, and driver supply. In particular, we have:
\begin{align}
    F_r(c_{ij}^r) &= \frac{e^{-\epsilon c_{ij}^r}}{e^{-\epsilon c_{ij}^r}+e^{-\epsilon c_{ij}^{r,0}}}, \\
    F_d^1(c_{ij}^{d_f},c_{ij}^{d_o}) &= \frac{e^{-\eta c_{ij}^{d_f}}}{e^{-\eta c_{ij}^{d_f}}+e^{-\eta c_{ij}^{d_o}}+e^{-\eta c_{ij}^{d,0}}}, \\
    F_d^2(c_{ij}^{d_f},c_{ij}^{d_o}) &= \frac{e^{-\eta c_{ij}^{d_o}}}{e^{-\eta c_{ij}^{d_f}}+e^{-\eta c_{ij}^{d_o}}+e^{-\eta c_{ij}^{d,0}}}, \\
    F_d(q) &= \frac{e^{\sigma q}}{e^{\sigma q}+e^{\sigma q_0}},
\end{align}
where $\epsilon$, $\eta$ and $\sigma$ are the sensitivity parameters of the logit model, $c_{ij}^{r,0}$ is the average cost of the external options for passengers traveling from zone $i$ to zone $j$, $c_{ij}^{d,0}$ is the average costs of the outside options for delivery customers who want to send a package from zone $i$ to zone $j$, and $q_0$ is the average earnings of outside options for drivers.

\section{Solution Method} \label{sec:algo}
The platform's profit maximization problem (\ref{optimalpricing_trip}) is a non-convex program with a large number of endogenous decision variables, including $\lambda_{ij}^r,\lambda_{ij}^{d_f},\lambda_{ij}^{d_o},w_i^r,w_i^{d_f},t_{ij}^{d_f},w_i^I,w_i^{dg},P_{(i,n)}^c,p_{i,pick}^n, \bar N_i^I, N_i^{Ig}, q$, which are intimately correlated through the constraints (\ref{eq:success_rate_drop})-(\ref{eq:stationary_T_jj}). 
In this section, we will investigate the structure of the constraints by identifying the interdependence between the endogenous variables, and then exploit these structures to develop a tailored algorithm for (\ref{optimalpricing_trip}), which can converge to a locally optimal solution at accelerated speed compared to standard interior-point methods. 

A crucial challenge for solving problem (\ref{optimalpricing_trip}) is to deal with the constraints (\ref{constraint_optimapricing}), which includes a large number of nonlinear equations and endogenous variables.
These variables are intimately related with each other.
To explore the structure of the constraints and examine the interdependence between these variables, we select the ride fare $\mathbf{r}^r=(r_1^r,r_2^r,\dots,r_M^r)$,  the generalized cost for flexible delivery $\mathbf{c}^{d_f} = \begin{bmatrix}
c_{11}^{d_f} & \cdots & c_{1M}^{d_f} \\
\vdots & \ddots & \vdots \\
c_{M1}^{d_f} & \cdots & c_{MM}^{d_f}
\end{bmatrix}$ and the number of idle drivers $\mathbf{N}^I=(N_1^r,N_2^r,\dots,N_M^r)$ as elementary variables, and show hereinafter how the other variables (including decision variables and endogenous variables) in problem (\ref{optimalpricing_trip}) can be represented by these selected elementary variables.

First of all, given $r_{i}^r$ and $N_i^I$, we can obtain the arrival rate of passengers through the demand function:
\begin{align}\label{lambda_r_represent}
\lambda_{ij}^r=\lambda_{ij}^{r,0}F_r\left(\alpha_rw_i^r(N_i^I)+r_i^rt_{ij}\right).
\end{align}
Furthermore, given $r_{i}^r$, $N_i^I$, and the generalized cost for flexible delivery services $c_{ij}^{d_f}$, we can obtain the arrival rate of delivery customers through the following relations:
\begin{align}
    \lambda_{ij}^{d_f} &= \lambda_{ij}^{d,0}F_d^1(c_{ij}^{d_f},\alpha_dw_i^{d_f}(N_i^I)+p_d(t_{ij})+r_i^rt_{ij}),\\
    \lambda_{ij}^{d_o} &= \lambda_{ij}^{d,0}F_d^2(c_{ij}^{d_f},\alpha_dw_i^{d_f}(N_i^I)+p_d(t_{ij})+r_i^rt_{ij}).
\end{align}
This indicates that the arrival rates of passengers, on-demand delivery customers, and flexible delivery customers are all known once $r_{i}^r$, $N_i^I$ and $c_{ij}^{d_f}$ are given.

Next, we show how to derive the flexible delivery time $t_{ij}^{d_f}$ based on the selected variables. Recall that the delivery time $t_{ij}^g$ can be obtained through:
\begin{align}\label{eq:t_df_represent}
    t_{ij}^{d_f}=\mathbb{E}[T_{ij}]+\frac{1-p_{j,drop}^{succ}}{p_{j,drop}^{succ}}\mathbb{E}[T_{jj}],
\end{align}
where $\mathbb{E}\left[\mathbf{T}_j\right] = (\mathbf{I}-\mathbf{\tilde P})^{-1}\text{diag}\left(\mathbf{\hat P} \mathbf{S}^T\right)$ and $\mathbb{E}[T_{jj}] = \mathbf{P}_j\mathbf{S}_j^T+\mathbf{\tilde P}_j \mathbb{E}\left[\mathbf{T}_j\right]$. 
Note that the with arrival rate $\lambda_{ij}^r$ and $\lambda_{ij}^{d_o}$, we can obtain the matrix of one-step transition probabilities for drivers $\mathbf{P}$ defined in (\ref{transition_matrix}), where
\begin{align}
    P_{ij} = \frac{\lambda_{ij}^r+\lambda_{ij}^{d_o}}{\sum_{k=1}^M\left(\lambda_{ik}^r+\lambda_{ik}^{d_o}\right)}
\end{align}
From constraint (\ref{demand_constraint2}), the driver's cruising time $w_i^I$ in zone $i$ can be derived as:
\begin{align*}
    w_i^I = \frac{N_i^I}{\sum_{j=1}^M(\lambda_{ij}^r+\lambda_{ij}^{d_o})},
\end{align*}
which gives the transition time matrix $\mathbf{S}$ according to (\ref{definiction_s}).
Note that $p_{i,drop}^{succ}$ is a function of $w_i^I$ and the exogenous parameter $t_i^g$.
This indicates that the flexible delivery time $t_{ij}^{d_f}$ can be derived once $r_{i}^r$, $N_i^I$ and $c_{ij}^{d_f}$ are given.

Next, we show how to derive $N_i^{Ig}$ based on $r_{i}^r$, $N_i^I$ and $c_{ij}^{d_f}$. Recall that constraint (\ref{delivery_supply_constraint}) dictates the following relation:
\begin{align*}
    N_i^{Ig} = \sum_n N_{(i,n)}^Ip_{i,pick}^n = \sum_nN_i^I\frac{P_{(i,n)}^c}{\sum_{n'} P_{(i,n')}^c}p_{i,pick}^n.
\end{align*}
In this case, to obtain $N_i^{Ig}$, we need to express  $p_{i,pick}^n$ and $P_{(i,n)}^c$  as function of  elementary variables (i.e., $r_{i}^r$, $N_i^I$ and $c_{ij}^{d_f}$), respectively.  To this end,  we first note from (\ref{eq:success_rate_pick}) that $p_{i,pick}^n$  depends on $\bar{t}_i^g$ and $w_i^{dg}$, where the former is characterized in  (\ref{eq:def_bar_t_i_g}) and expressed as a function of $\bar{N}_i^I$, and the latter can be derived by combining (\ref{eq:success_rate_pick}) and (\ref{eq:def_bar_N_I_2}):

\begin{align} \label{eq:w_dg}
    w_i^{dg} = \frac{p_{i,pick}^{succ}(w_i^I,w_i^{dg},\bar t_i^g)\bar N_i^I}{\sum_{j=1}^M\lambda_{ij}^{d_f}},
\end{align}
by solving which we could get $w_i^{dg}$ as a function of $\bar N_i^I$ given $\bar t_i^g = \dfrac{L_i}{\bar N_i^I}$.
Since both $\bar{t}_i^g$ and $w_i^{dg}$  can be represented as a function of $\bar N_i^I$,  $p_{i,pick}^{succ}$ is also a function of $\bar N_i^I$.

Furthermore, we note
from (\ref{eq:limit_prob}) that $P_{(i,n)}^c$ can be obtained by solving:
\begin{align}
\label{equation_pzn}
\begin{cases}
    \nu_{(z,n)}P_{(z,n)}^c=\sum_{(z',n')}\nu_{(z',n')}P_{(z',n')}^cP^c_{(z',n')(z,n)}\\
     \sum_{(z,n)} P_{(z,n)}^c = 1,
\end{cases}
\end{align}
Clearly, the solution to (\ref{equation_pzn}) depends on $\nu_{(z,n)}$, which is defined by (\ref{eq:spend_time_CTMC}). In the ride-hand side of (\ref{eq:spend_time_CTMC}), all variables can be represented as functions of the elementary variables (i.e., $r_{i}^r$, $N_i^I$ and $c_{ij}^{d_f}$) except for $\bar{t}_z^g$ and $w_i^{dg}$, both of which are shown to be a function of $\bar N_i^I$. Therefore, we conclude that $P_{(i,n)}^c$ can be also expressed in terms of $\bar{N}_i^I$ under the condition that $r_{i}^r$, $N_i^I$ and $c_{ij}^{d_f}$ are given. 

Since both $p_{i,pick}^n$ and $P_{(i,n)}^c$ depends on $\bar{N}_i^I$, as far as we derive $\bar{N}_i^I$ as a function of elementary variables, then $N_i^{Ig}$ can be also expressed as a function of elementary variables. To this end, we note that by plugging (\ref{eq:def_bar_t_i_g}) into (\ref{eq:def_bar_N_I}), we can obtain the following equations:
\begin{align}\label{eq:get_bar_N_I}
    \bar N_i^I = N_i^I-t_i^g\sum_j \lambda_{ji}^{d_f}-N_i^I\frac{P_{(i,C_a)}^c}{\sum_{n'} P_{(i,n')}^c}(1-p_{flex,i}^{C_a}).
\end{align}
In the right-hand side of (\ref{eq:get_bar_N_I}), $\lambda_{ij}^{d_f}$ is a function of elementary variables, $p_{flex,i}^{C_a}$ is a function of $\lambda_{ij}^{df}$ (thus also a function of elementary variables), whereas $P_{(z,n)}^c$ is a function of $\bar{N}_i^I$. In this case, we can treat  (\ref{eq:get_bar_N_I}) as a set of equations for $\bar{N}_i^I$, and there are $M$ equations for $M$ unknowns, which if can be solved (we will discuss the existence of $\bar{N}_i^I$ later), will provide the value of  $\bar{N}_i^I$ given the elementary variables. 

Next, we can derive the service fare and driver wage as a function of elementary variables. In particular, by $c_{ij}^{d_f}=\alpha_dw_i^{d_f}(N_i^{Ig}) +p_d (t_{ij}^{d_f})+r_{ij}^{d_f}$, the price for flexible delivery services from zone $i$ to zone $j$ can be derived as:
\begin{align*}
    r_{ij}^{d_f} = c_{ij}^{d_f}-\alpha_dw_i^{d_f}(N_i^{Ig})-p_d(t_{ij}^{df}),
\end{align*}
which is a function of elementary variables as $N_i^{Ig}$ and $t_{ij}^{df}$ are functions of elementary variables. 
From (\ref{supply_conservation_const}), we have:
\begin{align}\label{q_represent}
    q = q_0+\frac{1}{\sigma}\log\left(\frac{\sum_{i=1}^M \sum_{j=1}^M (\lambda_{ij}^r+\lambda_{ij}^{d_o})t_{ij}+\sum_{i=1}^M \sum_{j=1}^M w_i^r (\lambda_{ij}^r+\lambda_{ij}^{d_o})  +  \sum_{i=1}^M \sum_{j=1}^M   w_i^I (\lambda_{ij}^r+\lambda_{ij}^{d_o})}{N_0-(\sum_{i=1}^M \sum_{j=1}^M (\lambda_{ij}^r+\lambda_{ij}^{d_o})t_{ij}+\sum_{i=1}^M \sum_{j=1}^M w_i^r (\lambda_{ij}^r+\lambda_{ij}^{d_o})  +  \sum_{i=1}^M \sum_{j=1}^M   w_i^I (\lambda_{ij}^r+\lambda_{ij}^{d_o}))}\right).
\end{align}
Therefore, all  of the endogenous variables and decision variables can be represented by the selected elementary variables. Note that based on the aforementioned derivations, the existence of the market equilibrium depends on two factors:
(1) the wellposedness of the flexible delivery time $t_{ij}^{d_f}$ in (\ref{eq:t_df_represent}),
and (2) the existence of the solution to equation (\ref{eq:w_dg}) and (\ref{eq:get_bar_N_I}) for reasonable $w_i^{dg}$ and $\bar N_i^I$.
The following proposition proves the wellposedness of the flexible delivery time $t_{ij}^{d_f}$ by showing the existence of $t_{ij}^{d_f}$ for any given demand: 
\begin{proposition}\label{prop:exist_t}
A non-negative average delivery time (\ref{eq:t_df_represent}) always exists.
\end{proposition}
\begin{proof}
From (\ref{transition_matrix}), we know that every element of $\mathbf{P}$ belongs to $(0,1)$ and the sum of any row or any column all equals to 1.
Note that $ \mathbf{\tilde P}$ in (\ref{eq:stationary_T_ij}) is a truncation of $\mathbf{P}$ which removes the $i^{th}$ column and $j^{th}$ row. Then we have that the matrix norm induced by p-norm for $\mathbf{\tilde P}$ when $p = \infty$ is
\begin{align*}
    \Vert \mathbf{\tilde P}\Vert_\infty = \max_{1\le i\le M-1} \sum_{j=1}^{M-1} P_{ij} < 1.
\end{align*}
The spectral radius $\rho(\mathbf{\tilde P})$ satisfies that
\begin{align*}
    \rho(\mathbf{\tilde P}) \le \Vert \mathbf{\tilde P}\Vert_\infty < 1.
\end{align*}
Thus the matrix $\mathbf{I}-\mathbf{\tilde P}$ is a M-matrix \cite{plemmons1977m}, which satisfies that\footnote{An M-matrix is a matrix $A\in \mathbb{R}^{n\times n}$ that satisfies $A = sI-B$ where $B\ge 0$ and $s>\rho(B)$.}
\begin{align*}
    (\mathbf{I}-\mathbf{\tilde P})^{-1}\ge 0.
\end{align*}
It implies that any element in $(\mathbf{I}-\mathbf{\tilde P})^{-1}$ is non-negative. Therefore, from (\ref{eq:t_ij^g}), (\ref{eq:stationary_T_ij}) and (\ref{eq:stationary_T_jj}), we have that for any $\lambda_{ij}^r$ and $\lambda_{ij}^{d_o}$, we can get non-negative delivery times $t_{ij}^{d_f}$ for any $i$ and $j$, which completes the proof.
\end{proof}

Since the existence criteria may depend on the probability of successful pickup and dropoff of the flexible services, here we proved the existence of $w_i^{dg}$ and $\bar N_i^I$ under a specific distribution in the following proposition:
\begin{proposition}\label{prop:exist_N} Assume that $\log{W_i^I}\sim \mathcal{N}(\mu_{w_i^I},\sigma_{w_i^I})$,  $\log{T_i^g}\sim \mathcal{N}(\mu_{t_i^g},\sigma_{t_i^g})$, $\log{W_i^{dg}}\sim \mathcal{N}(\mu_{w_i^{dg}},\sigma_{w_i^{dg}})$, $\log{\bar T_i^g}\sim\mathcal{N}(\mu_{\bar t_i^g}, \sigma_{\bar t_i^g})$. Assume the correlation parameter for $\log{W_i^I}$ and $\log{W_i^{dg}}$ is $\rho_w$, and the correlation parameter for $\log{W_i^I}$ and $\log{\bar T_i^g}$ is $\rho_t$. 
\begin{enumerate}
    \item Constraint (\ref{eq:w_dg}) is a well-posed problem, where $w_i^{dg}$ can be treated as a function of $\bar N_i^I$ in the domain of all feasible $\bar N_i^I$. 
    \item There 
      exists a feasible solution  $\bar N_i^I$  to (\ref{eq:get_bar_N_I}) treating $w_i^{dg}$ as a function of $\bar N_i^I$ derived from (i) if the following condition holds:
     \begin{align}\label{condition_exist_N}
         N_i^I\left(1-\frac{P_{(i,C_a)}^{c^*}}{\sum_{n'}P_{(i,n')}^{c^*}}(1-p_{flex,i}^{C_a})\right)-t_i^g\sum_j\lambda_{ji}^{d_f}\ge 0,
     \end{align}
     where $P_{(i,n)}^{c^*}$ is determined by:
     \begin{align}    
     &\frac{1}{{\bar t_{(i,n)}}^*}P_{(i,n)}^{c^*}=\sum_{(i',n')}\frac{1}{{\bar t_{(i',n')}}^*}P_{(i',n')}^{c^*}P_{(i',n')(i,n)}^{c^*}\quad \text{and}\quad \sum_{(i,n)} P_{(i,n)}^{c^*} = 1, \\
     &{\bar t_{(i,n)}}^* = \begin{cases}
     w_i^I & n = 0 \\
     \Phi\left(\frac{\log w_i^I-\frac{{\sigma_{t_i^g}}^2}{2}-\log t_i^g+\frac{{\sigma_{w_i^I}}^2}{2}}{\sqrt{\sigma_{w_i^I}^2+\sigma_{t_i^g}^2}}\right)p_{flex,i}^nt_i^g+ &  \\
     \left(1-\Phi\left(\frac{\log w_i^I-\frac{{\sigma_{t_i^g}}^2}{2}-\log t_i^g+\frac{{\sigma_{w_i^I}}^2}{2}}{\sqrt{\sigma_{w_i^I}^2+\sigma_{t_i^g}^2}}\right)p_{flex,i}^n\right)w_i^I & n>0
     \end{cases}, \\
     & P_{(i,n)(i',n')}^{c^*} = 
     \begin{cases}
         \Phi\left(\frac{\log w_i^I-\frac{{\sigma_{t_i^g}}^2}{2}-\log t_i^g+\frac{{\sigma_{w_i^I}}^2}{2}}{\sqrt{\sigma_{w_i^I}^2+\sigma_{t_i^g}^2}}\right)p_{flex,i}^n & i' = i, n >0, n' = n-1\\
         P_{ii'} & n=0, n'=n \\
         \left(1-\Phi\left(\frac{\log w_i^I-\frac{{\sigma_{t_i^g}}^2}{2}-\log t_i^g+\frac{{\sigma_{w_i^I}}^2}{2}}{\sqrt{\sigma_{w_i^I}^2+\sigma_{t_i^g}^2}}\right)p_{flex,i}^n\right)P_{ii'} & n>0, n'=n  \\
         0 & \text{otherwise}
     \end{cases},
     \end{align}
     where $\Phi(x) = \frac{1}{\sqrt{2\pi}}\int_{-\infty}^xe^{-\frac{t^2}{2}}dt$ denotes the cumulative distribution function (CDF) of the standard normal distribution.
\end{enumerate}

\end{proposition}
\begin{proof}
    We first prove (i). Based on the fixed point theorem, a sufficient condition under which equation $w_i^{dg} = \frac{p_{i,pick}^{succ}(w_i^I,w_i^{dg},\bar t_i^g)\bar N_i^I}{\sum_{j=1}^M\lambda_{ij}^{d_f}}$ exists a unique solution is that:
    \begin{enumerate}[(a)]
        \item Given $\bar N_i^I$, $\frac{p_{i,pick}^{succ}(w_i^I,w_i^{dg},\bar t_i^g)\bar N_i^I}{\sum_{j=1}^M\lambda_{ij}^{d_f}}$ is monotonically decreasing.
        \item The following conditions are satisfied:\begin{align} \label{exist_condition_w_dg}
        \lim_{w_i^{dg}\rightarrow 0} \frac{p_{i,pick}^{succ}(w_i^I,w_i^{dg},\bar t_i^g)\bar N_i^I}{\sum_{j=1}^M\lambda_{ij}^{d_f}} > 0\quad \text{and}\quad \lim_{w_i^{dg}\rightarrow \infty} \frac{p_{i,pick}^{succ}(w_i^I,w_i^{dg},\bar t_i^g)\bar N_i^I}{\sum_{j=1}^M\lambda_{ij}^{d_f}} < \infty  
        \end{align}
    \end{enumerate}

    Note that when the critical variables and $\bar N_i^I$ are give, we have $w_i^I$, $\bar t_i^g$ and $\lambda_{ij}^{d_f}$ are known. Then by the theory of multivariate log-normal distribution
    \begin{align}
        p_{i,pick}^{succ}(w_i^I,w_i^{dg},\bar t_i^g) = p_{i,pick}^{succ,w}(\bar t_i^g,w_i^I)\cdot p_{i,pick}^{succ,t}(w_i^{dg},w_i^I)= p_{i,pick}^{succ,w}(\bar t_i^g,w_i^I)\cdot  \Phi \left(\frac{\log w_i^I-\frac{{\sigma_{w_i^I}}^2}{2}-\log w_i^{dg}+\frac{{\sigma_{w_i^{dg}}}^2}{2}}{\sqrt{\sigma_{w_i^{dg}}^2+\sigma_{w_i^I}^2-2\rho_t\sigma_{w_i^{dg}}\sigma_{w_i^I}}}\right)
    \end{align}
    which is decreasing in $w_i^{dg}$. This completes the proof of (a).
    Since $\lim_{w_i^{dg}\rightarrow 0}p_{i,pick}^{succ} > 0$ and $\lim_{w_i^{dg}\rightarrow\infty}p_{i,pick}^{succ}= 0$, we have the conditions (\ref{exist_condition_w_dg}) always hold which completes the proof of (b).

    Next, we show the proof of (ii). The right-hand side of (\ref{eq:get_bar_N_I}) can be written as:
    \begin{align}
        RHS = N_i^I\left(1-\frac{P_{(i,C_a)}^c}{\sum_{n'}P_{(i,n')}^c}(1-p_{flex,i}^{C_a})\right)-t_i^g\sum_j\lambda_{ji}^{d_f}.
    \end{align}
    Note that $\bar N_i^I\in[0,N_i^I]$.
    To prove that $\bar N_i^I = RHS$ has a solution, based on fixed point theorem, a sufficient condition is that 
    \begin{align}
        \lim_{\mathbf{\bar N}^I\rightarrow \mathbf{N}^I}RHS \le N_i^I \quad \text{and} \quad \lim_{\mathbf{\bar N}_i^I\rightarrow 0}RHS \ge 0,       
    \end{align}
    since $RHS$ is a continuous function of $\bar N_i^I$ for all $i$.

    It is intuitive that $RHS \le N_i^I$ always holds, thus here we prove that the condition (\ref{condition_exist_N}) is equivalent to $\lim_{\mathbf{\bar N}^I\rightarrow 0}RHS \ge 0$. From (\ref{eq:w_dg}), we have 
    \begin{align}
        \lim_{\bar N_i^I\rightarrow 0} w_i^{dg} = \lim_{\bar N_i^I\rightarrow 0} \frac{p_{i,pick}^{succ}(w_i^I,w_i^{dg},\bar t_i^g)\bar N_i^I}{\sum_{j=1}^M\lambda_{ij}^{d_f}} = 0.
    \end{align}
    Then, from (\ref{eq:success_rate_pick}) and $\lim_{\bar N_i^I\rightarrow 0}\bar t_i^g = \infty$, we have $\lim_{\mathbf{\bar N}^I\rightarrow 0} p_{i,pick}^{succ} = 0$,
    which yields $\lim_{\mathbf{\bar N}^I\rightarrow 0} P_{(z,n)(z,n+1)}^c = 0$.

    Since the expected value of $T_i^g$ is $t_i^g$ and the expected value of $W_i^I$ is $w_i^I$, based on the theory of log-normal distribution we have
    \begin{align}
        \mu_{t_i^g} &= \log{t_i^g}-\frac{{\sigma_{t_i^g}}^2}{2}, \\
        \mu_{w_i^I} &= \log{w_i^I}-\frac{{\sigma_{w_i^I}}^2}{2},
    \end{align}
    By the theory of normal distribution, we have
    \begin{align}
        p_{i,drop}^{succ} = \mathbb{P}(T_i^g-W_i^I<0) = \mathbb{P}(\log{T_i^g}-\log{W_i^I}<0)= \Phi\left(\frac{\log w_i^I-\frac{{\sigma_{t_i^g}}^2}{2}-\log t_i^g+\frac{{\sigma_{w_i^I}}^2}{2}}{\sqrt{\sigma_{w_i^I}^2+\sigma_{t_i^g}^2}}\right)
    \end{align}
    
    From (\ref{eq:tran_prob_CTMC}), we have
    \begin{align}
        \lim_{\mathbf{\bar N}^I\rightarrow 0} P_{(i,n)(i',n')}^c =  P_{(i,n)(i',n')}^{c^*} = 
     \begin{cases}
         \Phi\left(\frac{\log w_i^I-\frac{{\sigma_{t_i^g}}^2}{2}-\log t_i^g+\frac{{\sigma_{w_i^I}}^2}{2}}{\sqrt{\sigma_{w_i^I}^2+\sigma_{t_i^g}^2}}\right)p_{flex,i}^n & i' = i, n >0, n' = n-1\\
         P_{ii'} & n=0, n'=n \\
         \left(1-\Phi\left(\frac{\log w_i^I-\frac{{\sigma_{t_i^g}}^2}{2}-\log t_i^g+\frac{{\sigma_{w_i^I}}^2}{2}}{\sqrt{\sigma_{w_i^I}^2+\sigma_{t_i^g}^2}}\right)p_{flex,i}^n\right)P_{ii'} & n>0, n'=n  \\
         0 & \text{otherwise}
     \end{cases}.
    \end{align}
    
    From (\ref{eq:spend_time_CTMC}) we can derive that
    \begin{enumerate}
        \item When $n = 0$, 
        \begin{align*}
            \lim_{\mathbf{\bar N}^I\rightarrow 0} \bar t_{(i,n)} &= \lim_{\mathbf{\bar N}^I\rightarrow 0} P_{(i,n)(i,n+1)}^c(w_i^{dg}+\bar t_i^g)+\sum_{j\in\{1,\dots,M\}}P_{(i,n)(j,n)}^cw_i^I\\
            &= \lim_{\mathbf{\bar N}^I\rightarrow 0}p_{i,pick}^{succ,w}p_{i,pick}^{succ,t}\cdot(w_i^{dg}+\bar t_i^g)+\sum_{j\in\{1,\dots,M\}}P_{ij}w_i^I\\
            &=\lim_{\mathbf{\bar N}^I\rightarrow 0}p_{i,pick}^{succ,w}p_{i,pick}^{succ,t}\cdot\bar t_i^g+w_i^I\\
        \end{align*}
        Note that $\lim_{\bar N_i^I\rightarrow 0}\bar t_i^g\rightarrow \infty$, and we have $\lim_{\bar t_i^g\rightarrow \infty}p_{i,pick}^{succ,t}=0$. Based on the theory of log-normal distribution and multivariate normal distribution, we have
        \begin{align}
            p_{i,pick}^{succ,t} = \mathbb{P}(\bar T_i^g < W_i^I) = \Phi \left(\frac{\log w_i^I-\frac{{\sigma_{w_i^I}}^2}{2}-\log \bar t_i^g+\frac{{\sigma_{\bar t_i^g}}^2}{2}}{\sqrt{\sigma_{\bar t_i^g}^2+\sigma_{w_i^I}^2-2\rho_t\sigma_{\bar t_i^g}\sigma_{w_i^I}}}\right) = \frac{1}{\sqrt{2\pi}}\int_{-\infty}^{\frac{\log w_i^I-\frac{{\sigma_{w_i^I}}^2}{2}-\log \bar t_i^g+\frac{{\sigma_{\bar t_i^g}}^2}{2}}{\sqrt{\sigma_{\bar t_i^g}^2+\sigma_{w_i^I}^2-2\rho_t\sigma_{\bar t_i^g}\sigma_{w_i^I}}}} e^{-\frac{t^2}{2}}dt,
        \end{align}
        
        Thus we have
        \begin{align}
            \lim_{\bar N_i^I\rightarrow 0}p_{i,pick}^{succ,t}\bar t_i^g &= \lim_{\bar t_i^g\rightarrow \infty} \frac{\bar t_i^g}{\sqrt{2\pi}}\int_{\frac{\log \bar t_i^g}{\sqrt{\sigma_{\bar t_i^g}^2+\sigma_{w_i^I}^2-2\rho_t\sigma_{\bar t_i^g}\sigma_{w_i^I}}}}^{\infty} e^{-\frac{t^2}{2}}dt \\
            & < \lim_{\bar t_i^g\rightarrow \infty} \frac{\bar t_i^g}{\sqrt{2\pi}}\int_{\frac{\log \bar t_i^g}{\sqrt{\sigma_{\bar t_i^g}^2+\sigma_{w_i^I}^2-2\rho_t\sigma_{\bar t_i^g}\sigma_{w_i^I}}}}^{\infty} te^{-\frac{t^2}{2}}dt \\
            &= \frac{1}{\sqrt{2\pi}}\cdot \lim_{\bar t_i^g\rightarrow \infty} \frac{\bar t_i^g}{\exp\left(\frac{(\log \bar t_i^g)^2}{2(\sigma_{\bar t_i^g}^2+\sigma_{w_i^I}^2-2\rho_t\sigma_{\bar t_i^g}\sigma_{w_i^I})}\right)}\\
            &= \frac{1}{\sqrt{2\pi}}\cdot \lim_{\bar t_i^g\rightarrow \infty} \frac{\bar t_i^g}{\bar t_i^{g^{\frac{\log \bar t_i^g}{2(\sigma_{\bar t_i^g}^2+\sigma_{w_i^I}^2-2\rho_t\sigma_{\bar t_i^g}\sigma_{w_i^I})}}}}\\
            & \rightarrow 0,
        \end{align}
        which yields $\lim_{\bar N_i^I\rightarrow 0}p_{i,pick}^{succ,t}\bar t_i^g = 0$.
        Therefore, we have
        \begin{align}
            \lim_{\mathbf{\bar N}^I\rightarrow 0} \bar t_{(i,n)} = \lim_{\mathbf{\bar N}^I\rightarrow 0}p_{i,pick}^{succ,w}p_{i,pick}^{succ,t}\cdot\bar t_i^g+w_i^I  = w_i^I.
        \end{align}

        \item When $0 < n< C_a$,
        \begin{align*}
            \lim_{\mathbf{\bar N}^I\rightarrow 0} \bar t_{(i,n)} &=\lim_{\mathbf{\bar N}^I\rightarrow 0} P_{(i,n)(i,n-1)}^ct_i^g+P_{(i,n)(i,n+1)}^c(w_i^{dg}+\bar t_i^g)+\sum_{j\in\{1,\dots,M\}}P_{(i,n)(j,n)}^cw_i^I\\
            &= \Phi\left(\frac{\log w_i^I-\frac{{\sigma_{t_i^g}}^2}{2}-\log t_i^g+\frac{{\sigma_{w_i^I}}^2}{2}}{\sqrt{\sigma_{w_i^I}^2+\sigma_{t_i^g}^2}}\right)p_{flex,i}^nt_i^g+ \\
            &\quad \left(1-\Phi\left(\frac{\log w_i^I-\frac{{\sigma_{t_i^g}}^2}{2}-\log t_i^g+\frac{{\sigma_{w_i^I}}^2}{2}}{\sqrt{\sigma_{w_i^I}^2+\sigma_{t_i^g}^2}}\right)p_{flex,i}^n\right)w_i^I.
        \end{align*}

        \item When $n = C_a$,
        \begin{align*}
            \lim_{\mathbf{\bar N}^I\rightarrow 0} \bar t_{(i,n)} &=\lim_{\mathbf{\bar N}^I\rightarrow 0} P_{(i,n)(i,n-1)}^ct_i^g+\sum_{j\in\{1,\dots,M\}}P_{(i,n)(j,n)}^cw_i^I\\
            &=\Phi\left(\frac{\log w_i^I-\frac{{\sigma_{t_i^g}}^2}{2}-\log t_i^g+\frac{{\sigma_{w_i^I}}^2}{2}}{\sqrt{\sigma_{w_i^I}^2+\sigma_{t_i^g}^2}}\right)p_{flex,i}^nt_i^g+\\
            &\quad \left(1-\Phi\left(\frac{\log w_i^I-\frac{{\sigma_{t_i^g}}^2}{2}-\log t_i^g+\frac{{\sigma_{w_i^I}}^2}{2}}{\sqrt{\sigma_{w_i^I}^2+\sigma_{t_i^g}^2}}\right)p_{flex,i}^n\right)w_i^I.
        \end{align*}
    \end{enumerate}
    This completes the proof.

\end{proof}

\begin{remark}
    Proposition \ref{prop:exist_t} indicates that for any given values of the crucial variables $\mathbf{r}^r$, $\mathbf{c}^{d_f}$ and $\mathbf{N}^I$, we can always find a unique non-negative average delivery time $t_{ij}^g$ for any $i$ and $j$. Proposition \ref{prop:exist_N} indicates that if $\mathbf{r}^r$,  $\mathbf{c}^{d_f}$ and $\mathbf{N}^I$ are in a reasonable range, there always exists a solution to equation (\ref{eq:get_bar_N_I}) which can yield a feasible $\bar N_i^I$, and a solution to equation (\ref{eq:w_dg}) which can yield a feasible $w_i^{dg}$.
    We have validated the condition in Proposition \ref{prop:exist_N} in the case study, which illustrates that it is fairly weak and holds for a large range of values of the model parameters.
    These two propositions are essential for developing a tailored algorithm that can compute the optimal solution efficiently.  
\end{remark}


Based on the above discussion, all the decision variables and endogenous variables, including $\lambda_{ij}^r$, $\lambda_{ij}^{d_f}$, $\lambda_{ij}^{d_o}$, $w_i^r$, $w_i^{d_f}$, $t_{ij}^{d_f}, w_i^I, w_i^{dg}, P_{(i,n)}^c, p_{i,pick}^n, \bar N_i^I, N_i^{Ig}$, and $q$, can be expressed as a function of elementary variables $\mathbf{r}^r$, $\mathbf{c}^{d_f}$ and $\mathbf{N}^I$.
This enables us to transform the original profit-maximization problem  (\ref{optimalpricing_trip})-(\ref{constraint_optimapricing}) into an equivalent form with a significantly smaller number of variables and constraints. In particular, we show that after changing the decision variables to $\mathbf{r}^r$, $\mathbf{c}^{d_f}$ and $\mathbf{N}^I$, the following optimization problem can be established: 
\begin{proposition} \label{prop:transform}
The optimization problem (\ref{optimalpricing_trip}) is equivalent to the following optimization problem:
\begin{align}\label{optimization_change_variable}
    \max_{\mathbf{r}^{r},\mathbf{c}^{d_f},\mathbf{N}^I\ge N^I_{min}, \mathbf{\bar N}^I,\mathbf{w}^{dg}} \quad \sum_{i=1}^M \sum_{j=1}^M r_i^rt_{ij}(\lambda_{ij}^r+\lambda_{ij}^{d_o}) + r_{ij}^{d_f}\lambda_{ij}^{d_f}-N_0F_d(q)q,
\end{align}
\begin{subnumcases}{\label{constraint_change_variable}}
{
\bar N_i^I = N_i^I-t_i^g\sum_j \lambda_{ji}^{d_f}-N_{(i,C_a)}^I(1-p_{flex,i}^{C_a})},\label{constraint_N_change_variable}\\
w_i^{dg} = \frac{p_{i,pick}^{succ}(w_i^I,w_i^{dg},\bar t_i^g)\bar N_i^I}{\sum_{j=1}^M\lambda_{ij}^{d_f}},\label{constraint_w_change_variable}
\end{subnumcases}
where $N_{min}^I = (L/w_{max})^2$, $\mathbf{c}^{d_f}(i,j)=c_{ij}^{d_f}=\alpha_dw_i^{d_f}(N_i^{Ig}) +p_d (t_{ij}^{d_f})+r_{ij}^{d_f}$, $\lambda_{ij}^r,\lambda_{ij}^{d_f},\lambda_{ij}^{d_o},r_{ij}^{d_f}, q, N_{(i,C_a)}^I$, $w_i^I$, $\bar t_i^g$ and $p_{flex,i}^{C_a}$ can be represented as explicit functions of the decision variables $\mathbf{r}^{r},\mathbf{c}^{d_f},\mathbf{N}^I$, $\mathbf{\bar N}^I$ and $\mathbf{w}^{dg}$.
\end{proposition}

\begin{proof}
From equations (\ref{lambda_r_represent}) to (\ref{q_represent}), it is easily to show that $\lambda_{ij}^r,\lambda_{ij}^{d_f},\lambda_{ij}^{d_o},r_{ij}^{d_f}$ and $q$ for any zone $i$ and zone $j$ can be represented by the decision variables while all of the equation constraints in (\ref{constraint_optimapricing}) are respected.
Note that from Proposition \ref{prop:exist_N}, we need to solve nonlinear equations (\ref{eq:get_bar_N_I}) to obtain $\mathbf{\bar N}^I$ and solve nonlinear equations (\ref{eq:w_dg}) to obtain $\mathbf{w^{dg}}$.
Equivalently, we incorporate equations (\ref{eq:get_bar_N_I}) and (\ref{eq:w_dg}) for all $i$ as constraints and treat $\bar N_i^I$ and $w_i^{dg}$ as decision variables.
The only remaining constraint is (\ref{upper_waiting}), which is incorporated into the box-constraint $\mathbf{N}^I\ge N_{min}^I$.
Therefore,  the objective function (\ref{optimization_change_variable}) and the constraints (\ref{constraint_change_variable}) are explicit functions of the new decision variables that can satisfy all of the constraints (\ref{constraint_optimapricing}), which completes the proof.
\end{proof}

To quickly obtain a good initial guess for the optimization problem   (\ref{optimization_change_variable}), 
we note that the third term of the ride-hand side of (\ref{constraint_N_change_variable}) depends on $\bar{N}_i^I$, making  (\ref{constraint_N_change_variable}) a fixed point equations for $\bar{N}_i^I$. However, in practice, the number of vehicles that reach the capacity for flexible parcels in zone $i$ multiplied by the probability that it does not have a packaged delivered to zone $i$ can be very small, thus by removing the third term $N_{(i,C_a)}^I(1-p_{flex,i}^{C_a})$, we can reach a good approximation of $\bar{N}_i^I$ without having to solve the fixed-point equation. In particular, instead of imposing constraint (\ref{constraint_change_variable}), we approximate $\bar N_i^I$ by:
\begin{align}\label{relax_bar_N_I}
    \bar N_i^I = N_i^I-t_i^g\sum_j\lambda_{ji}^{d_f},
\end{align}
In this case, the right-hand side of (\ref{relax_bar_N_I}) is independent from $\bar N_i^I$, and $\bar N_i^I$ is not the decision variables in the approximated problem, but an intermediate variable to get the objective function.

To deal with constraints (\ref{constraint_w_change_variable}), we note that when $\bar N_i^I$ and $\lambda_{ij}^{d_f}$ are fixed, $w_i^{dg} - \frac{p_{i,pick}^{succ}(w_i^I,w_i^{dg},\bar t_i^g)\bar N_i^I}{\sum_{j=1}^M\lambda_{ij}^{d_f}}$ is monotonously increasing in $w_i^{dg}$. Thus we can incorporate them in the objective function and use bisection method to get the solutions.

The optimization problem then becomes:
\begin{align}\label{optimization_approximate}
    \max_{\mathbf{r}^{r},\mathbf{c}^{d_f},\mathbf{N}^I\ge N^I_{min}} \quad \sum_{i=1}^M \sum_{j=1}^M r_i^rt_{ij}(\lambda_{ij}^r+\lambda_{ij}^{d_o}) + r_{ij}^{d_f}\lambda_{ij}^{d_f}-N_0F_d(q)q.
\end{align}
Note that problem (\ref{optimization_approximate}) is an unconstrained problem, which  is easier to solve. In the meanwhile, the neglected term is expected to be small and the optimal solution to (\ref{optimization_approximate}) is naturally close to that of (\ref{optimization_change_variable}). Therefore, after quickly obtaining a solution to problem (\ref{optimization_approximate}), we can resolve the reformulated problem  (\ref{optimization_change_variable}) using the solution to (\ref{optimization_approximate}) as the initial guess. In particular, given the solution to (\ref{optimization_approximate}), we will first  derive the corresponding actual $\bar N_i^I$ and $w_i^{dg}$ by (\ref{relax_bar_N_I}) and bisection method, and then treat the optimal solutions to (\ref{optimization_approximate}) and the derived $\bar N_i^I$ and $w_i^{dg}$ as the initial guess to solve the original problem (\ref{optimization_change_variable}), using standard interior-point algorithm. This enables use to solve the original  (\ref{optimization_change_variable}) with a very good initial guess. 
The overall procedure is summarized in Algorithm \ref{algo:1}.

\begin{algorithm}[ht!] 
\caption{Solution algorithm to the optimization problem} \label{algo:1}
\hspace*{0.02in} {\bf Input:} 
model parameters \\
\hspace*{0.02in} {\bf Output:} 
optimal decisions of the platform $\mathbf{r}^p$, $\mathbf{r}^g$ and $q$, 
\begin{algorithmic}[1]
\STATE Initialization: $\mathbf{y}_0=(\mathbf{r}^r_0, \mathbf{c}^{d_f}_0, \mathbf{N}^I_0)$;
\STATE Solve problem \begin{align}\nonumber
    \max_{\mathbf{r}^{r},\mathbf{c}^{d_f},\mathbf{N}^I\ge N^I_{min}} \quad \sum_{i=1}^M \sum_{j=1}^M r_i^rt_{ij}(\lambda_{ij}^r+\lambda_{ij}^{d_o}) + r_{ij}^{d_f}\lambda_{ij}^{d_f}-N_0F_d(q)q.
\end{align} with $\bar N_i^I = N_i^I-t_i^g\sum_j\lambda_{ji}^{d_f}$ from initial guess $\mathbf{y}_0$, and $w_i^{dg}$ from solving (\ref{constraint_w_change_variable}) using bisection method. Get optimal solution $\mathbf{y}_1=(\mathbf{r}^r_1, \mathbf{c}^{d_f}_1, \mathbf{N}^I_1)$;
\STATE Given $\mathbf{y}_1$, get $\mathbf{\bar N}^I_1$ from $\bar N_i^I = N_i^I-t_i^g\sum_j\lambda_{ji}^{d_f}$, and solve equations\begin{align*}
    w_i^{dg} = \frac{p_{i,pick}^{succ}(w_i^I,w_i^{dg},\bar t_i^g)\bar N_i^I}{\sum_{j=1}^M\lambda_{ij}^{d_f}}
\end{align*} to get $\mathbf{w}_1^{dg}$.
\STATE Using $\mathbf{y}_1$, $\mathbf{\bar N}^I_1$ and $\mathbf{w}_1^{dg}$ as initial guess, solve problem\begin{align*}
    \max_{\mathbf{r}^{r},\mathbf{c}^{d_f},\mathbf{N}^I\ge N^I_{min}, \mathbf{\bar N}^I,\mathbf{w}^{dg}} \quad \sum_{i=1}^M \sum_{j=1}^M r_i^rt_{ij}(\lambda_{ij}^r+\lambda_{ij}^{d_o}) + r_{ij}^{d_f}\lambda_{ij}^{d_f}-N_0F_d(q)q,
\end{align*}
\begin{subnumcases}{}\nonumber
{
\bar N_i^I = N_i^I-t_i^g\sum_j \lambda_{ji}^{d_f}-N_{(i,C_a)}^I(1-p_{flex,i}^{C_a})},\\ \nonumber
w_i^{dg} = \frac{p_{i,pick}^{succ}(w_i^I,w_i^{dg},\bar t_i^g)\bar N_i^I}{\sum_{j=1}^M\lambda_{ij}^{d_f}},
\end{subnumcases}to get optimal solution $\mathbf{y}_2=(\mathbf{r}^r_2, \mathbf{c}^{d_f}_2, \mathbf{N}^I_2, \mathbf{\bar N}^I_2, \mathbf{w}_2^{dg}),$;
\STATE Recover $\mathbf{r}^r$, $\mathbf{r}^{d_f}$ and $q$  from $\mathbf{y}_2$ based on equations (\ref{lambda_r_represent}) to (\ref{q_represent});
\RETURN $\mathbf{r}^r$, $\mathbf{r}^{d_f}$ and $q$.
\end{algorithmic}
\end{algorithm}

\begin{remark}
We would like to emphasize that our algorithm significantly outperforms directly optimizing the original problem (\ref{optimalpricing_trip}). The original problem has a large number of endogenous variables and nonlinear constraints. All the endogenous variables need to be treated as decision variables in the solvers. If the initial guess of the original problem  (\ref{optimalpricing_trip}) is poor, it may takes a relatively long time to converge, or even converge to an infeasible point.  Our algorithm improves this situation by the change of variables and warm start from an unconstrained problem. The change of variables significantly reduce the number of constraints and the number of decision variables, which accelerates the convergence of the algorithm. The warm start provide a good initial guess for the optimization problem after the change of variable, which can further speed up the computation. A comparison of the elapsed time of the proposed problem and directly solving the original problem is presented in Section \ref{sec:simulation}. 
\end{remark}

\section{Case Study}\label{sec:simulation}
This section presents a case study for San Francisco to evaluate the performance of the proposed model and algorithm.
In particular, a service region with 11 zones in San Francisco is considered, which is divided based on the zip code. The map of zip-code areas for San Francisco with the corresponding zone indexes is shown in Figure \ref{fig:SF_zip}.\footnote{Some areas with distinct zip-codes have similar level of demand and are geographically close, thus they are incorporated into a single zone for sake of simplicity.} In this case study, the benefits of integrating ride-sourcing and parcel delivery services will be quantified by comparing the integrated case with a few benchmark cases.  In particular, we compare the profits of the platform, the number of drivers, and the arrival rate of passengers\footnote{The number of drivers and arrival rate of passengers are monotonically related to driver surplus and passenger surplus, respectively. Therefore, we use them as an indicator of average surplus for each type of stakeholder. } under our proposed integrated business model and the benchmark case, to demonstrate the benefits of the proposed integrated business model. Intuitively, the extend to which stakeholders can benefit from the proposed business model depends on several crucial model parameters, such as the matching efficiency, and  the spatial distribution of demand of both ride-sourcing and parcel delivery orders. To test the robustness of the results obtained from the proposed model, we perform a few sensitivity studies, where the model parameters of the matching function and the spatial distribution of parcel delivery demand are perturbed. The details of the simulation studies will be presented hereafter.


\begin{figure}[!htb]
    \centering
    \includegraphics[width=\textwidth]{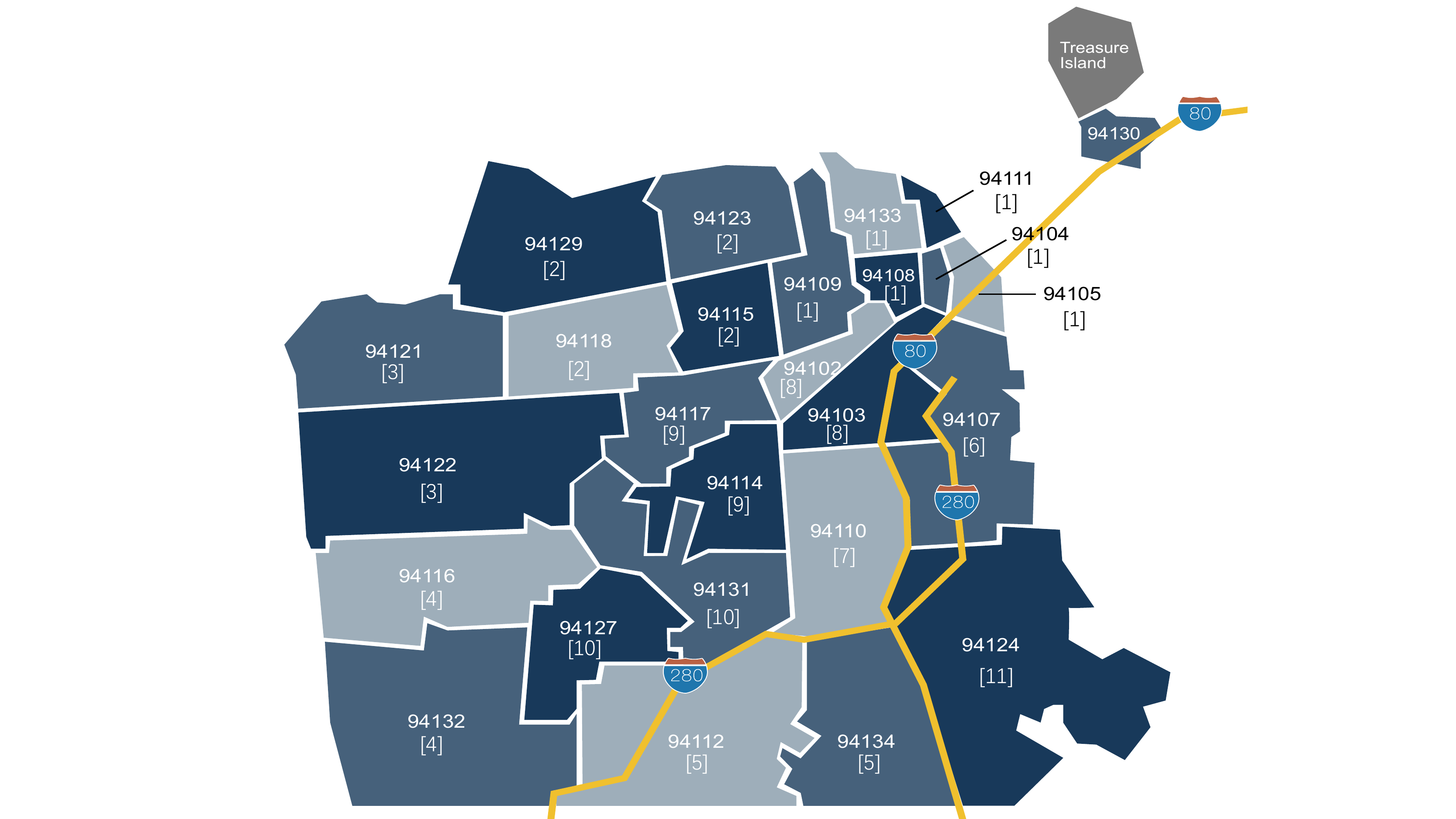}
    \caption{Zip-code areas for San Francisco labeled with corresponding zone indexes. (Figure courtesy: {\url{https://www.usmapguide.com/california/san-francisco-zip-code-map/}})}
    \label{fig:SF_zip}
\end{figure}

\subsection{Calibration of Parameters}
For the ride-sourcing market, we use the data that comprises the potential arrival rates of passengers for each origin and destination at the zip-code level, which is synthesized based on the real data of total pick-up and drop-off orders in each zone \cite{SFCTA2016data}, combined with a trip-distribution model calibrated from survey data. We calibrate the parameters so that both the parameters and the optimal solution to problem (\ref{optimalpricing_trip}) are close to the real-world data of San Francisco. The model parameters are:
\begin{align*}
    \Theta=\{\lambda_{ij}^{r,0},\lambda_{ij}^{d,0},N_0,L_i,\epsilon,\eta,\sigma,c_{ij}^{r,0},c_{ij}^{d,0},\alpha_r, \alpha_d, q_0,w_{max}\}
\end{align*}
We set $\lambda_{ij}^{r,0}$ to satisfy $0.15\lambda_{ij}^{r,0}=\lambda_{ij}^r$ where $\lambda_{ij}^r$ is the synthesized data that captures the arrival rate for ride-sourcing trips from zone $i$ to zone $j$, thus 15\% of the potential passengers choose ride-sourcing as their travel mode \cite{castiglione2016tncs}. The travel time for each origin-destination pair $t_{ij}$ is estimated in Google map, and the travel cost for the outside travel options $c_{ij}^{d,0}$ are assumed to be proportional to the travel time $t_{ij}$.
We manually tuned the sensitivity parameters for ride-sourcing passengers and drivers so that the optimal passengers demand, drivers supply, trip fares and driver wages are close to the real data without delivery services.
For instance, the average optimal ride-sourcing trip fare under the selected parameters is 
16.34 \$/trip, which is close to the fare estimates \cite{Uber2023fare} for a 2.6 mile trip \cite{castiglione2016tncs}.
The average optimal  wage for ride-sourcing drivers is 25.77 \$/hr, which is consistent with the hourly wages offered by Uber in San Francisco \cite{Uber2023wage}.
The average arrival rate of ride-sourcing passengers is 167.34 /min, and the total number of drivers on the platform is 3589, which are close to the demand and supply reported in \cite{castiglione2016tncs}.

For the parcel delivery market, there is lack of data on the OD demand of parcel orders, thus we synthesise the data based on the related real data of San Francisco. In particular, we consider two cases where the potential demand of parcel delivery customers is calibrated from the real data of San Francisco (referred to as the realistic case),  or inversely proportional to the potential demand of ride-sourcing passengers (referred to as the opposite direction case). For the realistic case, we perform order/trip generation (i.e., estimate the number of delivery order/trip originating in or destined for each zone) and order/trip distribution (i.e., determine where the trips produced in each zone will go)  of delivery services based on the realistic data of San Francisco. For simplicity, we consider delivery services that are sent from a commercial site (e.g,. grocery store, restaurants, etc) to the customer's home. 
In this case,  the trip production in zone $i$ (i.e., the number of delivery orders that have home-end in zone $i$) is determined by the zone's population, and the trip attraction in zone $i$ (i.e., the number of delivery orders that have non-home end in zone $i$) depends on the number of registered businesses in the zone. The number of registered business is obtained from The Treasurer \& Tax Collector’s Office \cite{zipcode2023business} and the population size is obtained from U.S. Census Bureau \cite{zipcode2023database}.  Let $P_i$ denote the total number of delivery orders produced in zone $i$, let $A_i$ denote the total number of delivery orders attracted to zone $i$, let $Reg_i$ denote the number of registered businesses at zone $i$ and $Pop_i$ be the population of zone $i$, then we assume
    \begin{align}
        P_i &= k_p\cdot Pop_i,\\
        A_i &= k_A\cdot Reg_i,
    \end{align}
where $P_i$ and $A_i$ are assumed to linearly depend on  $Pop_i$ and $Reg_i$ with positive coefficients $k_P$ and $k_A$.
    We employ the well-established ``gravity model'' to characterize the order distribution:
    \begin{align}
        \lambda_{ij}^{g,0} = \frac{P_jA_iF(t_{ij})}{\sum_{k=1}^MA_kF(t_{kj})},
    \end{align}
    where $F(t_{ij})$ is the friction factor for interchange $ij$, which is a decreasing function in travel time $t_{ij}$. In the case study, we define the friction factor by
    \begin{align}
        F(t_{ij}) = \frac{1}{t_{ij}}.
    \end{align}

Note that $k_p$ and $k_A$ are determined by the ratio of total potential delivery demand and total potential ride-sourcing demand.
Let $c_{ij}^{d,0}$ also be proportional to the travel time $t_{ij}$.
The average delivery fare for flexible delivery (i.e., 9.04 \$/order) is smaller than the average delivery fare for on-demand delivery services (i.e., 14.22 \$/order), which conforms to our expectation.
As reported in \cite[~p.58]{Uber2021report},
the revenue from the freight sector in 2021 is \$ 2132 millions, which is 31\% of the revenue from the ride-sourcing sector (\$ 6953 millions).
In the case with realistic data, when the ratio of the total potential delivery demand and the total potential ride-sourcing demand is 0.4, the optimal revenue of the delivery services under our selected parameters (863.4 \$/min) is 30.6\% of the revenue from ride-sourcing services (2825.7 \$/min), which is close to the ratio reported in \cite[~p.58]{Uber2021report}.
For the opposite-direction case, we let the spatial distribution of the potential delivery demand be inversely proportional to the spatial distribution of the ride-sourcing demand. For instance, if $\lambda_{ij}^{r,0}$ is the largest over all OD pairs of ride-sourcing demand, then $\lambda_{ij}^{d,0}$ is the smallest along all OD pairs of parcel delivery demand\footnote{The reason why we consider this case is that this would greatly change the spatial distribution of the overall demand, and we can evaluate whether this business model can work or not when the delivery demand is not aligned with the ride-sourcing demand.}.

The values of the parameters used in our numerical studies are summarized below:
\begin{align}
    N_0 = 10000,\ L_i = 43,\ \epsilon = 0.12,\ \eta = 0.16,\ \sigma = 0.18,\ \alpha_r = 3.2,\ \alpha_d = 0.7,\ q_0 = 29\ \$/hr,\ w_{max} = 6\ min
\end{align}

The delivery customer's disutility of the package delivery time is charaterized by:
\begin{align}\label{p_d_expression}
    p_d(t) = 25\left(\tanh\left(\frac{t}{200}-5\right)+1\right).
\end{align}
We comment that the properties of the sigmoid function (\ref{p_d_expression}) can well capture the customer's disutility to the delivery time.
When the delivery time is smaller than a threshold (e.g., one day), the growth rate of $p_d(t)$ is relatively small, which is consistent to the reality: If the parcel to be delivered is not urgent,  then customers are more sensitive to delivery fare than delivery time when it is within the tolerance (e.g., the same day).
However, when the delivery time is larger than some threshold, the growth rate of $p_d(t)$ has a dramatic increase, which indicates that if the delivery time exceeds the customer's expectation, their sensitivity to delivery time will sharply increase.

\subsection{Simulation Results under Realistic Case Parcel Demand}
We first use the calibrated model parameters to evaluate the proposed model and algorithm under the assumption that the spatial distribution of potential delivery demand depends on the commercial popularity and population of each zone.  In the rest of this subsection, we will report the computational performance of the proposed algorithm, and evaluate how the proposed integrated model affects multiple stakeholders in the ride-sourcing and parcel delivery market.

To evaluate the performance of the proposed solution method, we randomly select 10 sets of initial guesses within a specific range and set the delivery demand to ride-sourcing demand ratio at 0.4 (i.e., $\sum_{i,j}\lambda_{ij}^{d,0} = 0.4\sum_{i,j}\lambda_{ij}^{r,0}$). We then execute Algorithm 1, document its computation time and optimal profits, and compare these with the results from directly solving the original problem (problem (44) in the revised draft). The initial guesses are randomly chosen within the following ranges: \begin{align} \nonumber
        &r_i^{r}\in [1,2]\ \$/min,\quad r_{ij}^{d_f}\in [5,15]\ \$,\quad q\in [20,30]\ \$/hr, \quad\lambda_{ij}^r \in [0.15\lambda_{ij}^{r,0},0.25\lambda_{ij}^{r,0}],\\ \nonumber
        &\lambda_{ij}^{d_f}\in [0.1\lambda_{ij}^{d,0},0.2\lambda_{ij}^{d,0}],\quad\lambda_{ij}^{d_o}\in [0.1\lambda_{ij}^{d,0},0.2\lambda_{ij}^{d,0}],\quad N_i^I\in [150,250],\quad \bar N_i^I\in [50,150],\\ \label{initial_guess_range}
        &w_i^{dg}\in[5,15]\ min, c_{ij}^{d_f}\in [10,20] \$/min.
    \end{align}
The results, summarized in Table \ref{tab:compare_realistic}, show that regardless of initial guesses, Algorithm 1 can converge to similar profits (around 1900 \$/min) within 5 minutes. However, the computation and convergence performance of the original problem significantly varies with the initial guesses. In most instances, if the initial guess is inadequate, the convergence time can extend to several hours—a stark contrast to our proposed algorithm's computation time (a few minutes). The optimal profits from directly solving (44) are also generally lower than those from our algorithm. Only in rare instances (e.g., instance 7 in Table \ref{tab:compare_realistic}) does the objective of directly solving (44) outperform that of our algorithm, but at a slower rate. Therefore, we can conclude that our proposed algorithm can improve computation time, improve optimization objective values and improve algorithmic stability with respect to initial guess compared with directly solving the original problem using interior-point method.

\begin{table}[htb!]
    \centering
    \caption{Computation comparison  under realistic case.}
    \begin{tabular}{lccc}
    \hline
    \hline
       No.  &  Algorithm & Profits (\$/min) & Time (sec) \\
     \hline
     \multirow{2}{4em}{1} & Algorithm 1 & 1885.4 & 208.7 \\
     & Solve (44) & 1549.8 & 13325.5 \\
     \hline
     \multirow{2}{4em}{2} & Algorithm 1 & 1908.8 & 123.8 \\
     & Solve (44) & 1537.6 & 13369.4
     \\
     \hline
     \multirow{2}{4em}{3} & Algorithm 1 & 1923.7 & 160.3 \\
     & Solve (44) & 1622.1 & 12420.8 \\
     \hline
     \multirow{2}{4em}{4} & Algorithm 1 & 1844.6& 135.8 \\
     & Solve (44) & 1581.8 & 10293.1 \\
     \hline
     \multirow{2}{4em}{5} & Algorithm 1 & 1887.9 & 161.4 \\
     & Solve (44) & 1627.1 & 8730.0 \\
     \hline
     \multirow{2}{4em}{6} & Algorithm 1 & 1912.6 & 149.0\\
     & Solve (44) & 1629.2 & 21592.8 \\
     \hline
     \multirow{2}{4em}{7} & Algorithm 1 & 1832.5 & 179.8 \\
     & Solve (44) & 1515.2 & 14898.1 \\
     \hline
     \multirow{2}{4em}{8} & Algorithm 1 & 1902.2 & 89.0 \\
     & Solve (44) & 1591.5 & 12924.4 \\
     \hline
     \multirow{2}{4em}{9} & Algorithm 1 & 1847.7 & 133.5 \\
     & Solve (44) & 1925.4 & 412.7 \\
     \hline
     \multirow{2}{4em}{10} & Algorithm 1 & 1923.1 & 186.9 \\
     & Solve (44) & 1650.7 & 14458.0\\
     \hline
    \end{tabular}
    \label{tab:compare_realistic}
    \end{table}

To evaluate the impacts of the new business model on the stakeholders in the ride-sourcing market, including passengers, drivers and the ride-sourcing platform,
we fixed the ride-sourcing demand (using the real data), and gradually increase the level of parcel delivery demand. We normalize the level of delivery demand with respect to the ride-sourcing demand, and define it to be the ratio between these two, which is perturbed from 0 to 0.8. In particular, a higher ratio (or higher level of parcel delivery) indicates a large volume of parcel delivery orders, a ratio of 0 indicates that there is ride-sourcing demand only, and a ratio of 1 indicates that the ride-sourcing demand and parcel delivery demand are the same. The corresponding platform profits, total number of drivers on the platform, and total arrival rate of ride-sourcing passengers are presented in Figure \ref{fig:profits_same_onlyboth}-\ref{fig:passenger_same_onlyboth}.
It is clear that as the ratio increases, the platform profit, number of drivers, and passenger arrival rate both increases, indicating that integration with delivery services can lead to a Pareto improvement for the ride-sourcing markets by increased platform profits, increased average wages for drivers, and reduced costs for ride-sourcing passengers.
Intuitively, this is because the integration of the two services provide more orders to the platform, better utilize driver's idle time, and thus motivate the platform to lower the ride fare and attract more passengers.

\begin{figure}[t]
\begin{minipage}[b]{0.3\textwidth}
\centering
%
%
\definecolor{mycolor1}{rgb}{0.00000,0.44700,0.74100}%
\definecolor{mycolor2}{rgb}{0.85000,0.32500,0.09800}%
\begin{tikzpicture}

\begin{axis}[%
width=1.3in,
height=1.6in,
at={(0.729in,0.583in)},
scale only axis,
xmin=0,
xmax=0.8,
xlabel style={font=\color{white!15!black}},
xlabel={Level of delivery demand},
ymin=1000,
ymax=2600,
ylabel style={font=\color{white!15!black}},
ylabel={Platform profits},
axis background/.style={fill=white},
legend style={at={(0.03,0.97)}, anchor=north west, legend cell align=left, align=left, draw=white!15!black}
]
\addplot [color=mycolor1, line width=1.0pt]
  table[row sep=crcr]{%
0	1193.23021737378\\
0.1	1338.47399195162\\
0.2	1508.54474641935\\
0.3	1679.20752969288\\
0.4	1850.62289152593\\
0.5	2019.89753247542\\
0.6	2184.65607099236\\
0.7	2347.46010021115\\
0.8	2508.73059993365\\
};

\end{axis}
\end{tikzpicture}%
\vspace*{-0.3in}
\caption{The profits of the platform (\$/min) under different level of delivery demand (normalized with respect to ride-sourcing demand).}
\label{fig:profits_same_onlyboth}
\end{minipage}
\begin{minipage}[b]{0.03\textwidth}
\hfill
\end{minipage}
\begin{minipage}[b]{0.3\textwidth}
\centering
%
%
\definecolor{mycolor1}{rgb}{0.00000,0.44700,0.74100}%
\definecolor{mycolor2}{rgb}{0.85000,0.32500,0.09800}%
\begin{tikzpicture}

\begin{axis}[%
width=1.3in,
height=1.6in,
at={(0.729in,0.583in)},
scale only axis,
xmin=0,
xmax=0.8,
xlabel style={font=\color{white!15!black}},
xlabel={Level of delivery demand},
ymin=3500,
ymax=4800,
ylabel style={font=\color{white!15!black}},
ylabel={Number of drivers},
axis background/.style={fill=white},
legend style={at={(0.03,0.97)}, anchor=north west, legend cell align=left, align=left, draw=white!15!black}
]
\addplot [color=mycolor1, line width=1.0pt]
  table[row sep=crcr]{%
0	3586.93052445673\\
0.1	3703.08378664974\\
0.2	3825.29699850254\\
0.3	3955.66050699846\\
0.4	4091.85399355337\\
0.5	4229.03834568042\\
0.6	4360.64125601154\\
0.7	4492.35712014471\\
0.8	4613.58051828568\\
};

\end{axis}
\end{tikzpicture}%
\vspace*{-0.3in}
\caption{The total number of drivers under different level of delivery demand (normalized with respect to ride-sourcing demand).}
\label{fig:drivers_same_onlyboth}
\end{minipage}
\begin{minipage}[b]{0.03\textwidth}
\hfill
\end{minipage}
\begin{minipage}[b]{0.3\textwidth}
\centering
%
%
\definecolor{mycolor1}{rgb}{0.00000,0.44700,0.74100}%
\definecolor{mycolor2}{rgb}{0.85000,0.32500,0.09800}%
\begin{tikzpicture}

\begin{axis}[%
width=1.3in,
height=1.6in,
at={(0.729in,0.583in)},
scale only axis,
xmin=0,
xmax=0.8,
xlabel style={font=\color{white!15!black}},
xlabel={Level of delivery demand},
ymin=166,
ymax=174,
ylabel style={font=\color{white!15!black}},
ylabel={Passenger arrival rate},
axis background/.style={fill=white},
legend style={at={(0.03,0.97)}, anchor=north west, legend cell align=left, align=left, draw=white!15!black}
]
\addplot [color=mycolor1, line width=1.0pt]
  table[row sep=crcr]{%
0	167.311358207993\\
0.1	167.324026046956\\
0.2	167.820078181196\\
0.3	168.606695903133\\
0.4	169.78974642499\\
0.5	171.016057484186\\
0.6	172.001214132678\\
0.7	172.914697005952\\
0.8	173.495198531011\\
};

\end{axis}
\end{tikzpicture}%
\vspace*{-0.3in}
\caption{The arrival rate of total passengers (/min) under different level of delivery demand (normalized with respect to ride-sourcing demand).}
\label{fig:passenger_same_onlyboth}
\end{minipage}
\end{figure}

To further understand the market outcomes over the transportation network, we fix the level of delivery demand and demonstrate the results at the zonal level.  In particular, we choose the ratio of delivery demand to ride-sourcing demand to be 0.4 (e.g. $\sum_{i,j}\lambda_{ij}^{g0} = 0.4\sum_{i,j}\lambda_{ij}^{p0}$) as an example\footnote{We choose the ratio 0.4 since it conforms to the relative market size of delivery services and ride-sourcing services. Note that we have also compared the results under different demand ratios (i.e., from 0.1 to 0.8), and we found that the conclusions we drew from the ratio of ``0.4'' are also applicable to other ratios.}.
The spatial distributions of potential ride-sourcing demand and potential delivery demand are shown in Figure \ref{fig:potential_delivery_demand}.
The value corresponding to each zone is its total delivery demand as destination (e.g. the value for zone $i$ is $\sum_{j=1}^M\lambda_{ji}^{g0}$).
Figure \ref{fig:potential_delivery_demand} illustrates that the city center lies in the northeast of the city with the most ride-sourcing and parcel delivery demand, and
the remote areas, such as zone 4, 5, 10 and 11, lie in the south side of the city.

\begin{figure}[tb!]
     \centering
     \begin{subfigure}[The potential demand of ride-sourcing passengers {(/min)}.]{
         \centering
         \includegraphics[width=0.47\textwidth]{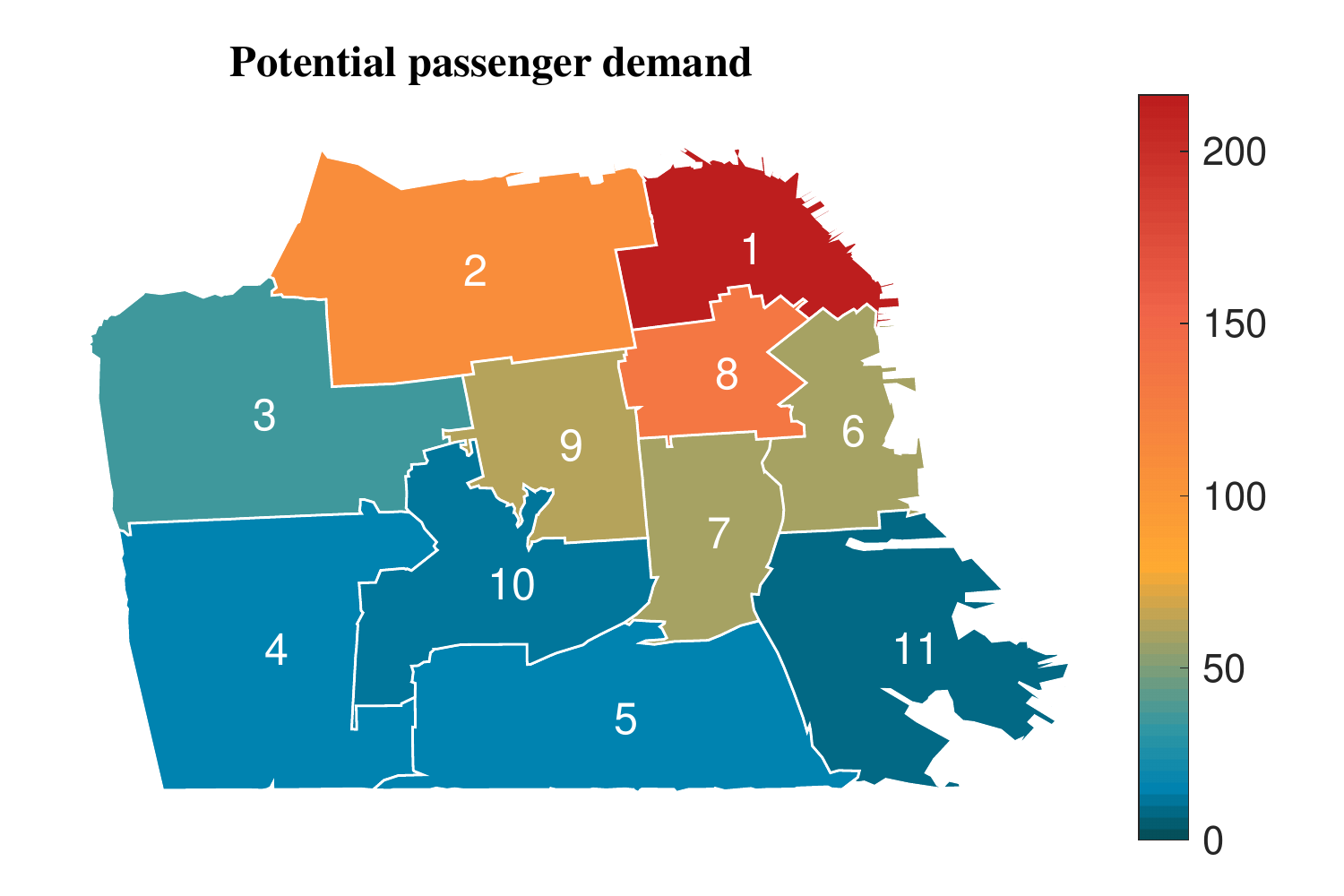}
         \label{fig:potential_passenger_demand_same}}
     \end{subfigure}
     \begin{subfigure}[The potential demand of delivery customers under realistic case {(/min)}.]{
         \centering
         \includegraphics[width=0.47\textwidth]{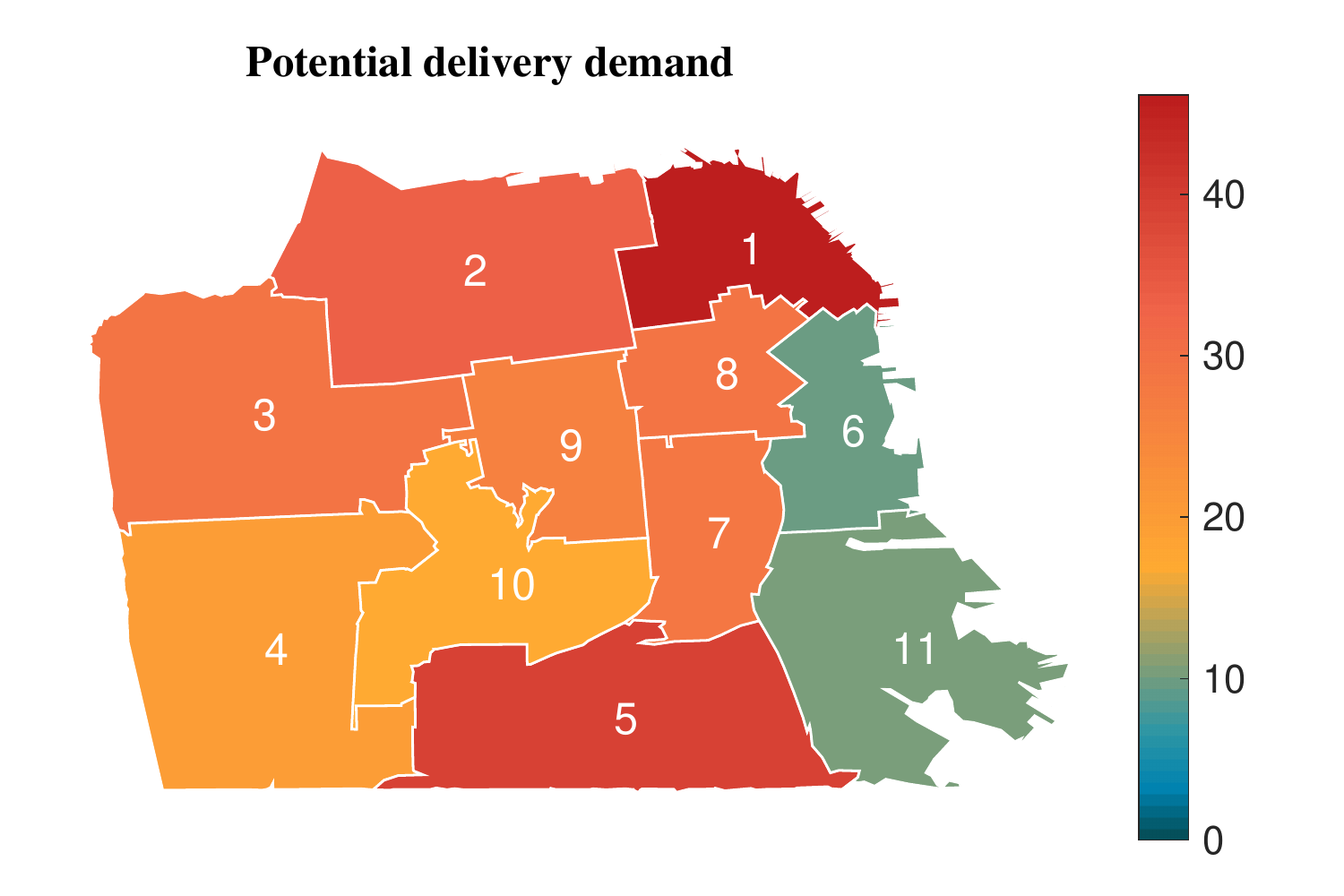}
         \label{fig:potential_delivery_demand_same}}
     \end{subfigure}
        \caption{Potential demand of ride-sourcing passengers and delivery customers under the same direction.}
        \label{fig:potential_delivery_demand}
\end{figure}

The impacts on the ride-sourcing market at the zonal level are demonstrated in Figure \ref{fig:idle_drivers_passengers_same}, where  Figure \ref{fig:idle_drivers_same} captures the change of idle drivers in each zone in percentage compared to the case of ride-sourcing services only, and Figure \ref{fig:passengers_same} captures the change of passenger trip production in each zone in percentage compared to the case with ride-sourcing services only.
The simulation results indicate that although the overall benefit to passengers is positive, its spatial distribution is highly asymmetric. In particular, the majority of increased idle drivers are concentrated in the central areas (e.g., zone 1-3 and 6-9), thus residents in these areas enjoy a reduced waiting time due to the integration of ride-sourcing and delivery service.  This is because the delivery services expand the demand of the central areas and bring more profits to the platform, who may find it profitable to dispatch more idle drivers to these areas.
However, for underserved zones (e.g., zone 10 and 11),  the ride-sourcing and delivery demand are small and the number of idle drivers remains almost unchanged. In this case, the change of the number of idle drivers in these zones are close to zero.
Since there is no increase in the number of idle drivers in the areas with lowest demand, the increased overall demand will motivate the platform to raise the prices, which may even increase the costs for passengers in these zones (e.g., zone 10 has negative change of passenger demand in Figure \ref{fig:passengers_same}). This highlights the uneven distribution of benefits among ride-sourcing passengers.

\begin{figure}[tb!]
     \centering
     \begin{subfigure}[The change of idle drivers for each zone as a ratio of the number of idle drivers at ride-only case.]{
         \centering
         \includegraphics[width=0.47\textwidth]{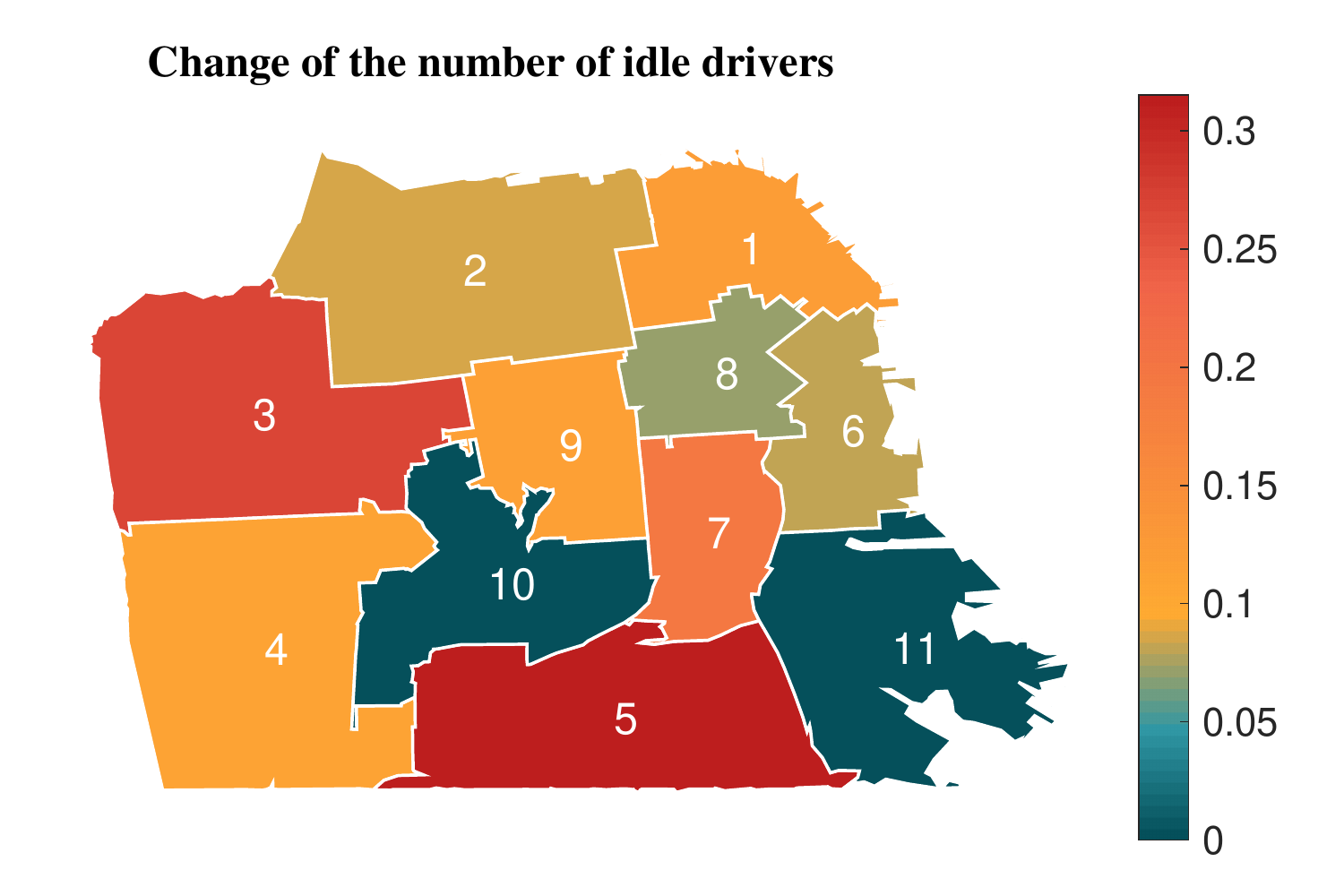}
         \label{fig:idle_drivers_same}}
     \end{subfigure}
     \begin{subfigure}[The change of passengers for each zone as a ratio of the passenger arrival rate at ride-only case.]{
         \centering
         \includegraphics[width=0.47\textwidth]{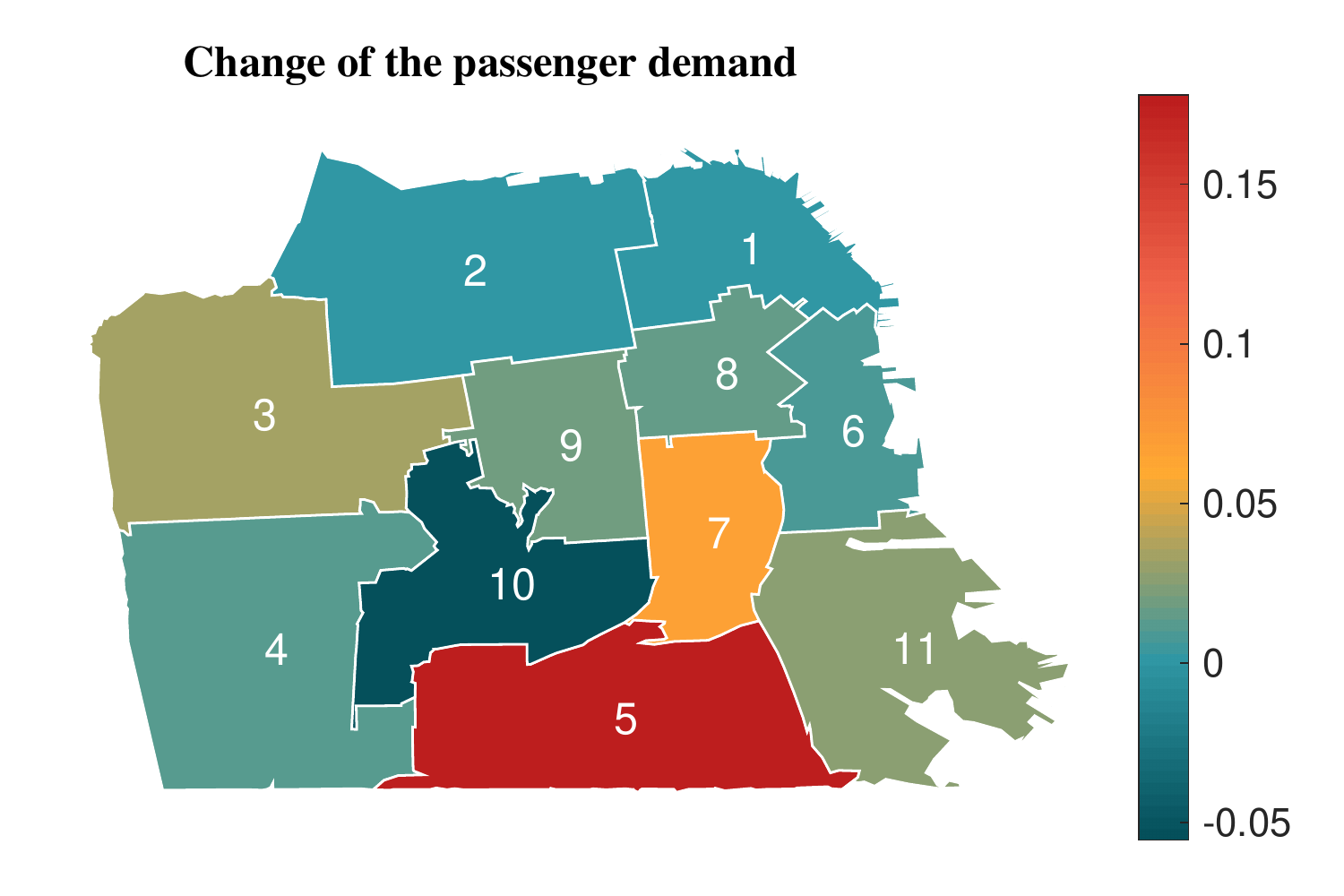}
         \label{fig:passengers_same}}
     \end{subfigure}
        \caption{The change of idle drivers and passengers for each zone under realistic case.}
        \label{fig:idle_drivers_passengers_same}
\end{figure}

Figure \ref{fig:demand_same}-\ref{fig:demand_same_origin} show the spatial equilibrium of the outcomes in the delivery sector by presenting the split of flexible and on-demand delivery orders over the transportation network.  Figure \ref{fig:demand_same} shows the total attraction\footnote{Total attraction for zone $i$ refers to the total amount of orders that arrive at zone $i$ as the destination. } of flexible/on-demand delivery orders in each zone of the city, and Figure \ref{fig:demand_same_origin} captures the total production\footnote{Total production for zone $i$ refers to the total amount of orders that originates from zone $i$. } of  flexible/on-demand delivery orders in each zone.
Interestingly, Figure \ref{fig:demand_same} indicates that when we look at the destinations of the delivery orders, the platform will give up flexible orders that are sent to the remote areas (i.e., flexible order in zone 4, 5, 10 and 11 is equal to 0 except the intra-zone flexible delivery orders.), thus customers who send orders to these zones will have no choice but to use on-demand services.  The underlying reason is straightforward: the delivery time of the flexible orders crucially depends on the probability that the drivers are dispatched to an on-demand ride-sourcing or delivery order with the same destination. Since remote areas have very low ride-sourcing and delivery demand, only a small proportion of drivers will be dispatched to these area. For this reason, the delivery time for flexible orders sent to remote areas is relatively long, which is inferior to the  on-demand delivery services even after considering the higher delivery prices.
We also note that for sending parcels to higher-demand zones (zone 1-3, 6-9), customers prefer to use flexible delivery services rather than on-demand services.
This is because  drivers have higher probability to visit the popular zones, which shortens the delivery time for flexible orders to be sent to these areas.
This attract a significant portion of customers to use flexible service, given that the delivery time is acceptable and the price is lower.
In contrast, Figure \ref{fig:demand_same_origin} indicates that when we look at the origin of the delivery orders, both flexible and on-demand delivery services are available to customers in all zones, and  flexible delivery services are more popular than on-demand delivery services. This is consistent with the intuition that most of the orders are sent to the popular areas with short delivery time, in which case the flexible delivery services outperforms on-demand delivery services because of the lower prices.

\begin{figure}[tb!]
     \centering
     \begin{subfigure}[The demand for flexible delivery (/min) under realistic case (as destination).]{
         \centering
         \includegraphics[width=0.47\textwidth]{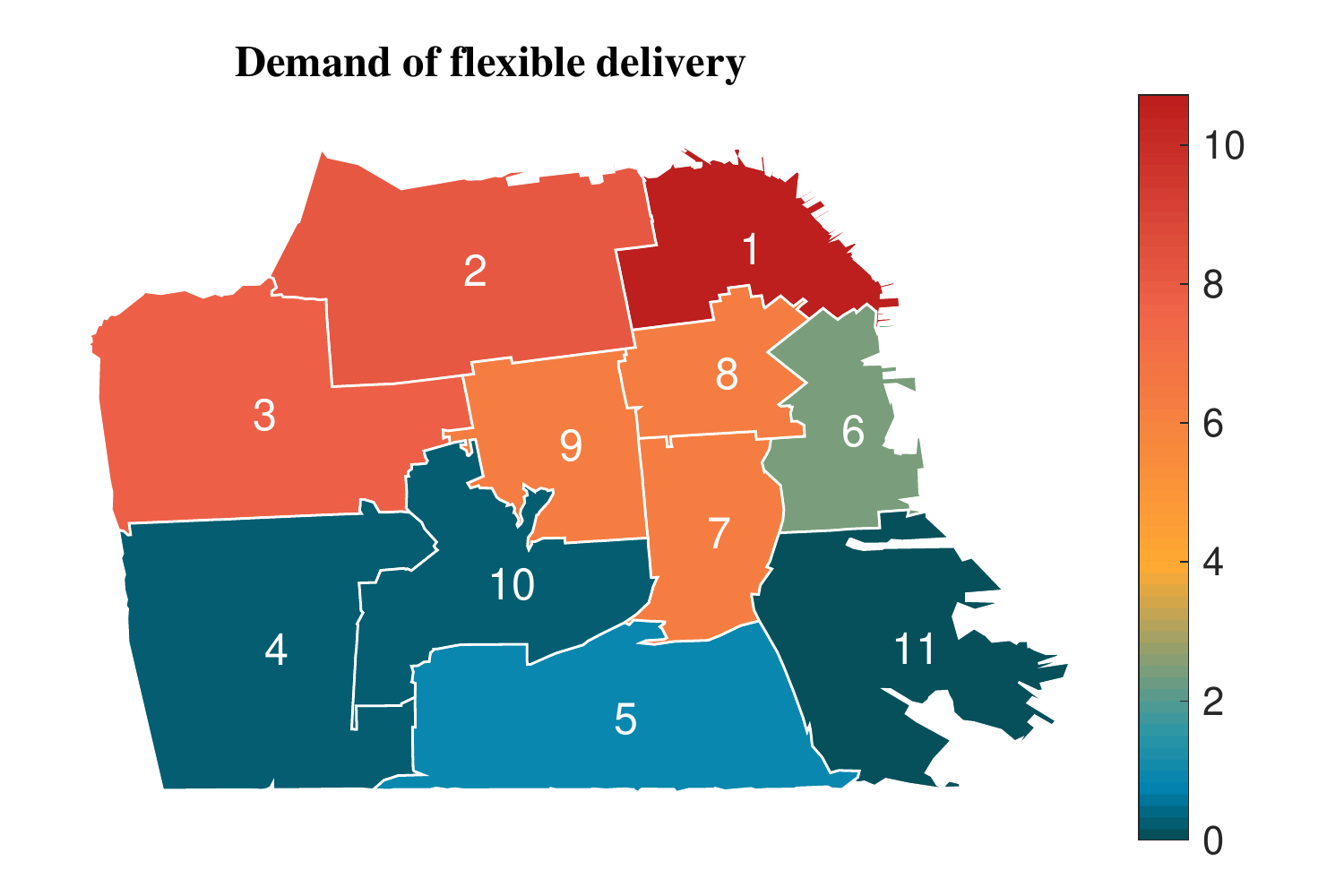}
         \label{fig:demand_flexible_same}}
     \end{subfigure}
     \begin{subfigure}[The demand for on-demand delivery {(/min)} under realistic case  (as destination).]{
         \centering
         \includegraphics[width=0.47\textwidth]{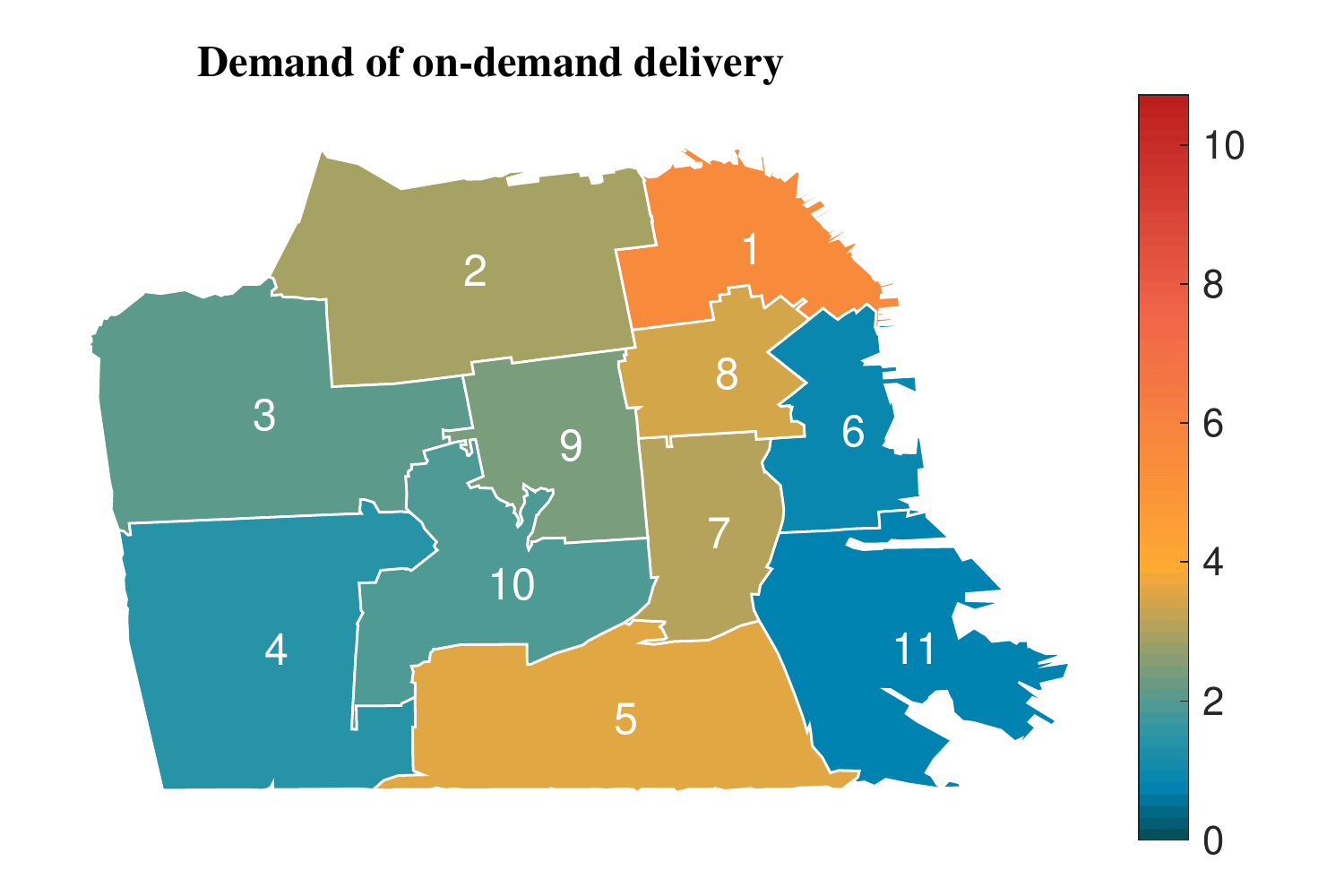}
         \label{fig:demand_fix_same}}
     \end{subfigure}
        \caption{Demand for delivery services  with destination to each zone under realistic case.}
        \label{fig:demand_same}
\end{figure}

\begin{figure}[tb!]
     \centering
     \begin{subfigure}[The demand for flexible delivery (/min) under realistic case (as origin).]{
         \centering
         \includegraphics[width=0.47\textwidth]{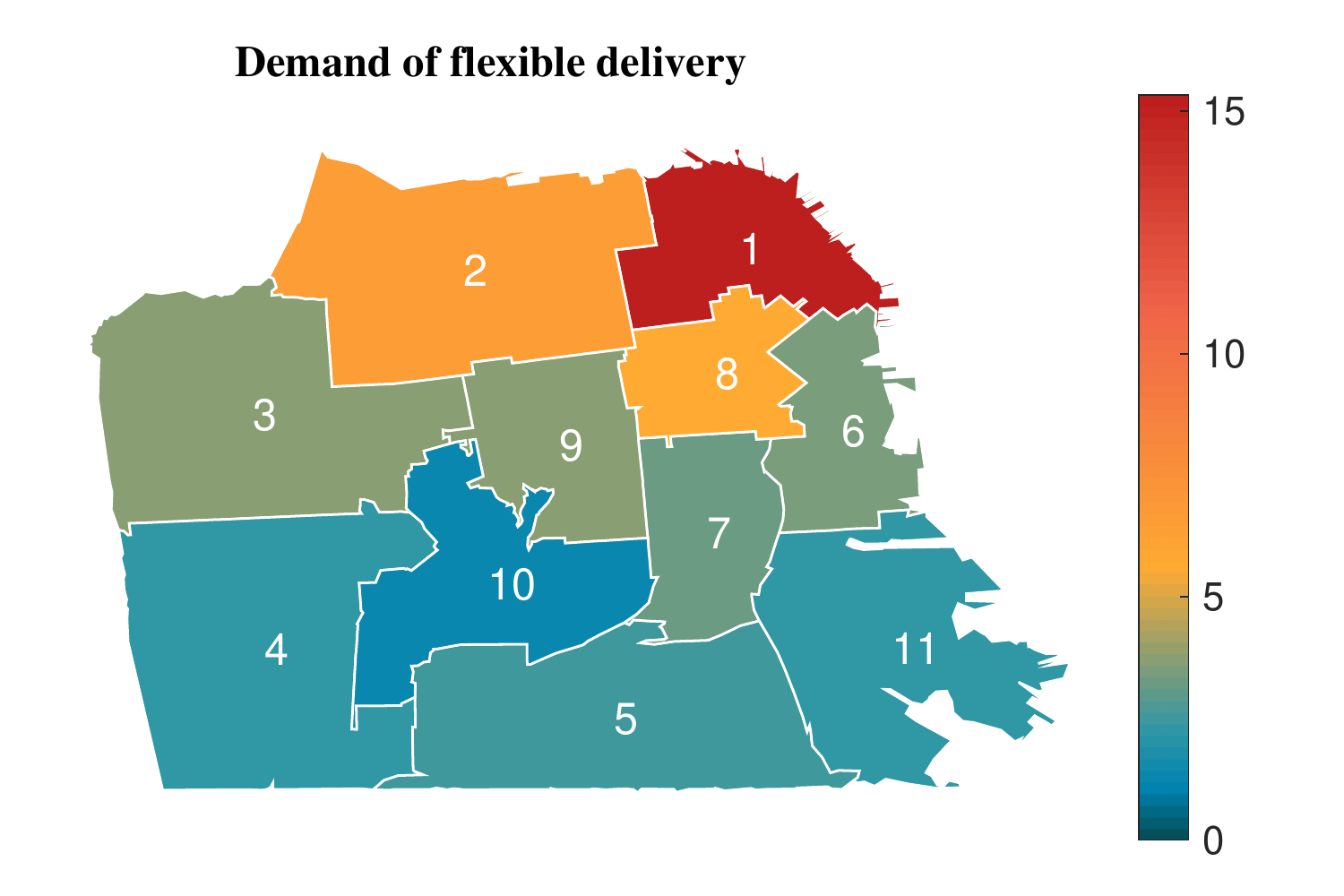}
         \label{fig:demand_flexible_origin}}
     \end{subfigure}
     \begin{subfigure}[The demand for on-demand delivery {(/min)} under realistic case (as origin).]{
         \centering
         \includegraphics[width=0.47\textwidth]{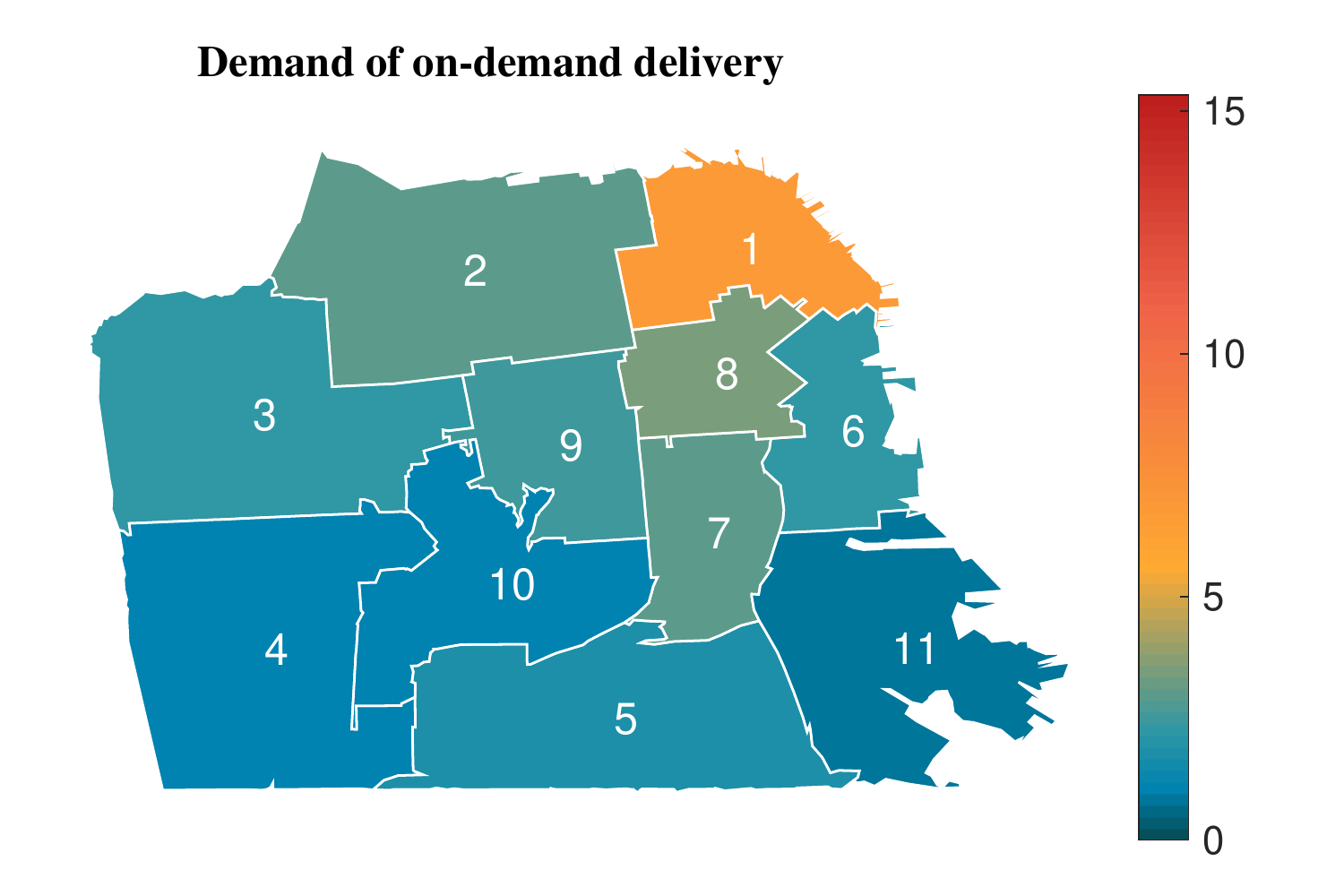}
         \label{fig:demand_fix_origin}}
     \end{subfigure}
        \caption{Demand for delivery services with origin from each zone under realistic case.}
        \label{fig:demand_same_origin}
\end{figure}

    \begin{figure}[t]
    \begin{minipage}[b]{0.48\textwidth}
    \centering
%
%
\definecolor{mycolor1}{rgb}{0.00000,0.44700,0.74100}%
\definecolor{mycolor2}{rgb}{0.85000,0.32500,0.09800}%
\begin{tikzpicture}

\begin{axis}[%
width=2in,
height=1.3in,
at={(0.729in,0.583in)},
scale only axis,
xmin=0,
xmax=0.8,
xlabel style={font=\color{white!15!black}},
xlabel={Level of delivery demand},
ymin=160,
ymax=180,
ylabel style={font=\color{white!15!black}},
ylabel={Passengers arrival rate},
axis background/.style={fill=white},
legend style={at={(0.03,0.97)}, anchor=north west, legend cell align=left, align=left, draw=white!15!black, nodes={scale=0.8, transform shape}}
]
\addplot [color=mycolor1, line width=1.0pt]
  table[row sep=crcr]{%
0	167.311358207993\\
0.1	167.324026046956\\
0.2	167.820078181196\\
0.3	168.606695903133\\
0.4	169.78974642499\\
0.5	171.016057484186\\
0.6	172.001214132678\\
0.7	172.914697005952\\
0.8	173.495198531011\\
};
\addlegendentry{integrate}

\addplot [color=mycolor2, line width=1.0pt]
  table[row sep=crcr]{%
0	167.311358207993\\
0.1	166.705818079558\\
0.2	165.875544871757\\
0.3	165.193346459961\\
0.4	164.328330081382\\
0.5	163.665018291353\\
0.6	162.878609712397\\
0.7	162.198550107802\\
0.8	161.895020512476\\
};
\addlegendentry{separate}

\end{axis}
\end{tikzpicture}%
    \vspace*{-0.3in}
    \caption{The arrival rate of total ride-sourcing passengers (/min) under different level of delivery demand under integrate case and separate case.}
    \label{fig:passengers_compare_with_separate_same}
    \end{minipage}
    \begin{minipage}[b]{0.03\textwidth}
    \hfill
    \end{minipage}
    \begin{minipage}[b]{0.48\textwidth}
    \centering
%
%
\definecolor{mycolor1}{rgb}{0.00000,0.44700,0.74100}%
\definecolor{mycolor2}{rgb}{0.85000,0.32500,0.09800}%
\begin{tikzpicture}

\begin{axis}[%
width=2in,
height=1.3in,
at={(0.729in,0.583in)},
scale only axis,
xmin=0,
xmax=0.8,
xlabel style={font=\color{white!15!black}},
xlabel={Level of delivery demand},
ymin=0,
ymax=170,
ylabel style={font=\color{white!15!black}},
ylabel={Customers arrival rate},
axis background/.style={fill=white},
legend style={at={(0.03,0.97)}, anchor=north west, legend cell align=left, align=left, draw=white!15!black, nodes={scale=0.8, transform shape}}
]
\addplot [color=mycolor1, line width=1.0pt]
  table[row sep=crcr]{%
0	0\\
0.1	19.6344532807271\\
0.2	39.9158412055956\\
0.3	59.6373124788435\\
0.4	78.9690031095784\\
0.5	98.0198817407213\\
0.6	116.737107699023\\
0.7	135.23394888834\\
0.8	153.941198548898\\
};
\addlegendentry{integrate}

\addplot [color=mycolor2, line width=1.0pt]
  table[row sep=crcr]{%
0	0\\
0.1	6.04234664388167\\
0.2	16.5007787508738\\
0.3	26.3531579381333\\
0.4	37.8873119240535\\
0.5	48.1404984275155\\
0.6	59.1457258962469\\
0.7	69.4948518886385\\
0.8	76.6650334014425\\
};
\addlegendentry{separate}

\end{axis}
\end{tikzpicture}%
    \vspace*{-0.3in}
    \caption{The arrival rate of total delivery customers (/min) under different level of delivery demand  under integrate case and separate case.}
    \label{fig:customers_compare_with_separate_same}
    \end{minipage}
    \begin{minipage}[b]{0.48\textwidth}
    \centering
%
%
\definecolor{mycolor1}{rgb}{0.00000,0.44700,0.74100}%
\definecolor{mycolor2}{rgb}{0.85000,0.32500,0.09800}%
\begin{tikzpicture}

\begin{axis}[%
width=2in,
height=1.3in,
at={(0.729in,0.583in)},
scale only axis,
xmin=0,
xmax=0.8,
xlabel style={font=\color{white!15!black}},
xlabel={Level of delivery demand},
ymin=3400,
ymax=5400,
ylabel style={font=\color{white!15!black}},
ylabel={Number of drivers},
axis background/.style={fill=white},
legend style={at={(0.03,0.97)}, anchor=north west, legend cell align=left, align=left, draw=white!15!black, nodes={scale=0.8, transform shape}}
]
\addplot [color=mycolor1, line width=1.0pt]
  table[row sep=crcr]{%
0	3586.93052445673\\
0.1	3703.08378664974\\
0.2	3825.29699850254\\
0.3	3955.66050699846\\
0.4	4091.85399355337\\
0.5	4229.03834568042\\
0.6	4360.64125601154\\
0.7	4492.35712014471\\
0.8	4613.58051828568\\
};
\addlegendentry{integrate}

\addplot [color=mycolor2, line width=1.0pt]
  table[row sep=crcr]{%
0	3586.93052445673\\
0.1	3764.37023279577\\
0.2	4000.10462631258\\
0.3	4186.38676470055\\
0.4	4414.10815028336\\
0.5	4583.00747959509\\
0.6	4775.84297077406\\
0.7	4936.78976311515\\
0.8	5006.69558796825\\
};
\addlegendentry{separate}

\end{axis}
\end{tikzpicture}%
    \vspace*{-0.3in}
    \caption{The total number of drivers on the platform(s) under different level of delivery demand  under integrate case and separate case.}
    \label{fig:drivers_compare_with_separate_same}
    \end{minipage}
    \begin{minipage}[b]{0.03\textwidth}
    \hfill
    \end{minipage}
    \begin{minipage}[b]{0.48\textwidth}
    \centering
%
%
\definecolor{mycolor1}{rgb}{0.00000,0.44700,0.74100}%
\definecolor{mycolor2}{rgb}{0.85000,0.32500,0.09800}%
\begin{tikzpicture}

\begin{axis}[%
width=2in,
height=1.3in,
at={(0.729in,0.583in)},
scale only axis,
xmin=0,
xmax=0.8,
xlabel style={font=\color{white!15!black}},
xlabel={Level of delivery demand},
ymin=1000,
ymax=2600,
ylabel style={font=\color{white!15!black}},
ylabel={Platform profits},
axis background/.style={fill=white},
legend style={at={(0.03,0.97)}, anchor=north west, legend cell align=left, align=left, draw=white!15!black, nodes={scale=0.8, transform shape}}
]
\addplot [color=mycolor1, line width=1.0pt]
  table[row sep=crcr]{%
0	1193.23021737378\\
0.1	1338.47399195162\\
0.2	1508.54474641935\\
0.3	1679.20752969288\\
0.4	1850.62289152593\\
0.5	2019.89753247542\\
0.6	2184.65607099236\\
0.7	2347.46010021115\\
0.8	2508.73059993365\\
};
\addlegendentry{integrate}

\addplot [color=mycolor2, line width=1.0pt]
  table[row sep=crcr]{%
0	1193.23021737378\\
0.1	1216.62911413014\\
0.2	1251.80739420554\\
0.3	1287.03784540467\\
0.4	1323.55962125329\\
0.5	1364.4706299093\\
0.6	1403.34558508427\\
0.7	1445.27500951391\\
0.8	1492.08550339008\\
};
\addlegendentry{separate}

\end{axis}
\end{tikzpicture}%
    \vspace*{-0.3in}
    \caption{The total platform profits from two services (\$/min) under different level of delivery demand under integrate case and separate case.}
    \label{fig:profits_compare_with_separate_same}
    \end{minipage}
    \end{figure}

\subsection{Compare with Benchmark Cases}

To better understand the benefits of the integrated business model in both ride-sourcing services and delivery services, as well as evaluating its environmental and social impacts, we compare the arrival rates of passengers, the arrival rate of delivery customers, the total number of drivers working on the platform, and the total platform profit to a benchmark case where the ride-sourcing platform and on-demand delivery platform are operated separately (referred to as separate case). Drivers choose to join which platform by comparing the perspective wages, which can be captured by a well-established multinomial logit model, while the objective of each platform is to maximize its own profits. We evaluate the total arrival rate of the ride-sourcing passengers, the total arrival rate of the delivery customers, and the total number of drivers working for ride-sourcing services and/or delivery services, and the total platform profit. The results are summarized in Figure \ref{fig:passengers_compare_with_separate_same}-\ref{fig:profits_compare_with_separate_same}. 
The findings suggest that integrating on-demand and flexible delivery services with ride-sourcing services attracts more passengers and delivery customers to the platform, resulting in increased overall platform profits. Simultaneously, this integration contributes to environmental sustainability by reducing the number of vehicles, thereby decreasing traffic congestion, vehicle miles traveled, and carbon emissions.

To examine the benefits of the flexible delivery services and  differentiated pricing strategies, we consider another benchmark case where an integrated platform provides both ride-sourcing services and on-demand delivery services, but these two services are charged at a non-differentiated price. We refer to the benchmark case as on-demand-only case, and refer to our proposed integrated platform with flexible delivery services as flexible and on-demand case. The comparison of average delivery fees in different scenarios is depicted in Figure \ref{fig:fare_same}, and the impacts on the platform, delivery customers, ride-sourcing passengers, and drivers are illustrated in Figures \ref{fig:profits_same_on_demand}-\ref{fig:drivers_same_on_demand}. Flexible delivery services yield a lower delivery fare compared to on-demand delivery services, as shown in Figure \ref{fig:fare_same}. Additionally, flexible delivery services significantly enhance platform profit (Figure \ref{fig:profits_same_on_demand}) and attract more delivery customers to the platform (Figure \ref{fig:delivery_same_on_demand}), while slightly increasing the fleet size (Figure \ref{fig:drivers_same_on_demand}) and causing negligible change in passenger arrival rates (Figure \ref{fig:passengerss_same_on_demand}). This indicates that by unlocking the value of flexibility of parcel demand, it can benefit the platform, the parcel delivery demand, slightly increase the earning of drivers, with negligible impacts on the ride-sourcing passengers.

\begin{figure}[htb!]
    \begin{minipage}[b]{\textwidth}
    \centering
%
%
%
\definecolor{mycolor1}{rgb}{0.2118,0.7647,0.7882}%
\definecolor{mycolor2}{rgb}{ 0.9882,0.6667,0.4039}%
\definecolor{mycolor3}{rgb}{0.5765,0.4392,0.8588}%
\begin{tikzpicture}

\begin{axis}[%
width=5in,
height=2.2in,
at={(0.729in,0.583in)},
scale only axis,
bar shift auto,
xmin=0.511111111111111,
xmax=11.4888888888889,
xtick={ 1,  2,  3,  4,  5,  6,  7,  8,  9, 10, 11},
xlabel style={font=\color{white!15!black}},
xlabel={Zone},
ymin=0,
ymax=25,
ylabel style={font=\color{white!15!black}},
ylabel={delivery fare (\$)},
axis background/.style={fill=white},
legend style={legend cell align=left, align=left, draw=white!15!black}
]
\addplot[ybar, bar width=0.178, fill=mycolor1,  draw = white,area legend] table[row sep=crcr] {%
1	14.392889272439\\
2	16.2205325801597\\
3	14.4929601263129\\
4	15.393591578454\\
5	12.4650069068841\\
6	14.252030424069\\
7	12.0148152424449\\
8	13.8874914714862\\
9	14.0314580418791\\
10	12.0287443739743\\
11	15.5986357369495\\
};
\addplot[forget plot, color=white!15!black] table[row sep=crcr] {%
0.511111111111111	0\\
11.4888888888889	0\\
};
\addlegendentry{on-demand delivery fare in on-demand-only case}

\addplot[ybar, bar width=0.178, fill=mycolor2,  draw = white,area legend] table[row sep=crcr] {%
1	14.5508204181618\\
2	16.6861038215972\\
3	14.7693703928379\\
4	15.3630740999319\\
5	12.3353559298659\\
6	14.2510506160999\\
7	12.0189265478077\\
8	14.0923654461562\\
9	14.1140063516401\\
10	11.8993065079676\\
11	15.2234055518633\\
};
\addplot[forget plot, color=white!15!black] table[row sep=crcr] {%
0.511111111111111	0\\
11.4888888888889	0\\
};
\addlegendentry{on-demand delivery fare in flexible \& on-demand case}

\addplot[ybar, bar width=0.178, fill=mycolor3,  draw = white,area legend] table[row sep=crcr] {%
1	9.28008787880208\\
2	9.06242731835546\\
3	9.58578468675224\\
4	8.78900202232102\\
5	8.85180396431536\\
6	8.96510688111165\\
7	8.75087359591805\\
8	8.63299115787405\\
9	8.78628926953709\\
10	8.4030473970203\\
11	9.21725903023733\\
};
\addplot[forget plot, color=white!15!black] table[row sep=crcr] {%
0.511111111111111	0\\
11.4888888888889	0\\
};
\addlegendentry{flexible delivery fare in flexible \& on-demand case}

\end{axis}
\end{tikzpicture}%
    \vspace*{-0.3in}
    \caption{The average delivery fare for customers sending parcels from each zone using on-demand delivery services in the on-demand-only case, on-demand delivery services in the flexible \& on-demand case , and flexible delivery services in the flexible \& on-demand case.}
    \label{fig:fare_same}
    \end{minipage}
\end{figure}
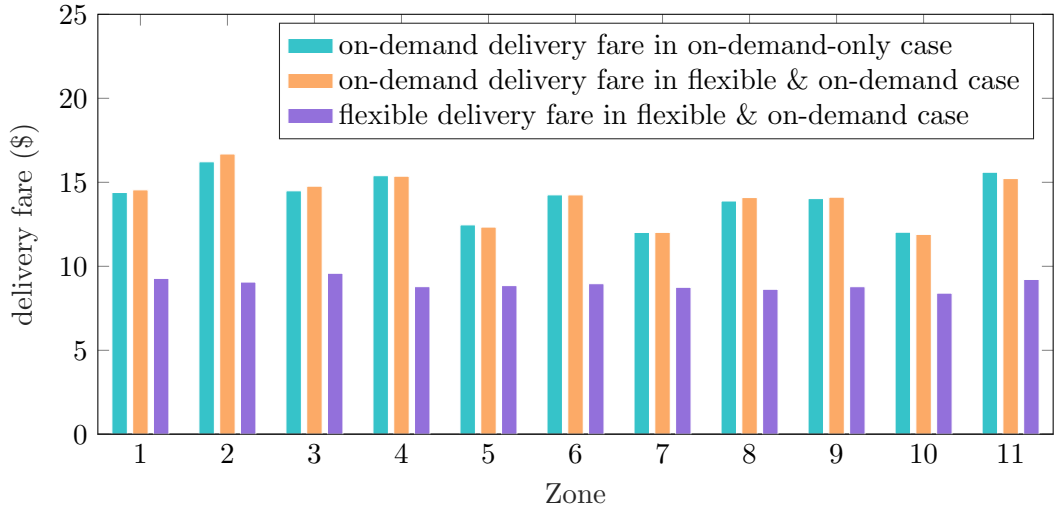

\begin{figure}[htb!]
    \begin{minipage}[b]{0.48\textwidth}
    \centering
%
%
\definecolor{mycolor1}{rgb}{0.00000,0.44700,0.74100}%
\definecolor{mycolor2}{rgb}{0.85000,0.32500,0.09800}%
\begin{tikzpicture}

\begin{axis}[%
width=2in,
height=1.3in,
at={(0.729in,0.583in)},
scale only axis,
xmin=0,
xmax=0.8,
xlabel style={font=\color{white!15!black}},
xlabel={Level of delivery demand},
ymin=1000,
ymax=2700,
ylabel style={font=\color{white!15!black}},
ylabel={Platform profits},
axis background/.style={fill=white},
legend style={at={(0.03,0.97)}, anchor=north west, legend cell align=left, align=left, nodes={scale=0.8, transform shape},draw=white!15!black}
]
\addplot [color=mycolor1, line width=1.0pt]
  table[row sep=crcr]{%
0	1193.23021737378\\
0.1	1338.47399195162\\
0.2	1508.54474641935\\
0.3	1679.20752969288\\
0.4	1850.62289152593\\
0.5	2019.89753247542\\
0.6	2184.65607099236\\
0.7	2347.46010021115\\
0.8	2508.73059993365\\
};
\addlegendentry{flexible \& on-demand}

\addplot [color=mycolor2, line width=1.0pt]
  table[row sep=crcr]{%
0	1193.23021737378\\
0.1	1259.25304528238\\
0.2	1324.74969369085\\
0.3	1389.68532354321\\
0.4	1454.03209378168\\
0.5	1517.78588385838\\
0.6	1580.93796559167\\
0.7	1643.46926402087\\
0.8	1705.36467207539\\
};
\addlegendentry{on-demand only}

\end{axis}
\end{tikzpicture}%
    \vspace*{-0.3in}
    \caption{The profits of the platform (\$/min) under different level of delivery demand under on-demand-only case and flexible \& on-demand case.}
    \label{fig:profits_same_on_demand}
    \end{minipage}
    \begin{minipage}[b]{0.03\textwidth}
    \hfill
    \end{minipage}
    \begin{minipage}[b]{0.48\textwidth}
    \centering
%
%
\definecolor{mycolor1}{rgb}{0.00000,0.44700,0.74100}%
\definecolor{mycolor2}{rgb}{0.85000,0.32500,0.09800}%
\begin{tikzpicture}

\begin{axis}[%
width=2in,
height=1.3in,
at={(0.729in,0.583in)},
scale only axis,
xmin=0,
xmax=0.8,
xlabel style={font=\color{white!15!black}},
xlabel={Level of delivery demand},
ymin=0,
ymax=180,
ylabel style={font=\color{white!15!black}},
ylabel={Customers arrival rate},
axis background/.style={fill=white},
legend style={at={(0.03,0.97)}, anchor=north west, legend cell align=left, align=left, draw=white!15!black, nodes={scale=0.8, transform shape}}
]
\addplot [color=mycolor1, line width=1.0pt]
  table[row sep=crcr]{%
0	0\\
0.1	19.6344532807271\\
0.2	39.9158412055956\\
0.3	59.6373124788435\\
0.4	78.9690031095784\\
0.5	98.0198817407213\\
0.6	116.737107699023\\
0.7	135.23394888834\\
0.8	153.941198548898\\
};
\addlegendentry{flexible \& on-demand}

\addplot [color=mycolor2, line width=1.0pt]
  table[row sep=crcr]{%
0	0\\
0.1	8.92433318054794\\
0.2	17.9157295232993\\
0.3	26.9442930422103\\
0.4	35.9849818289965\\
0.5	45.0278954972281\\
0.6	54.044620359041\\
0.7	63.0172515544058\\
0.8	71.9306104514998\\
};
\addlegendentry{on-demand only}

\end{axis}
\end{tikzpicture}%
    \vspace*{-0.3in}
    \caption{The arrival rate of total delivery customers (/min) under different level of delivery demand under on-demand-only case and flexible \& on-demand case.}
    \label{fig:delivery_same_on_demand}
    \end{minipage}
    \begin{minipage}[b]{0.48\textwidth}
    \centering
%
%
\definecolor{mycolor1}{rgb}{0.00000,0.44700,0.74100}%
\definecolor{mycolor2}{rgb}{0.85000,0.32500,0.09800}%
\begin{tikzpicture}

\begin{axis}[%
width=2in,
height=1.3in,
at={(0.729in,0.583in)},
scale only axis,
xmin=0,
xmax=0.8,
xlabel style={font=\color{white!15!black}},
xlabel={Level of delivery demand},
ymin=150,
ymax=190,
ylabel style={font=\color{white!15!black}},
ylabel={Passengers arrival rate},
axis background/.style={fill=white},
legend style={at={(0.03,0.97)}, anchor=north west, legend cell align=left, align=left, draw=white!15!black, nodes={scale=0.8, transform shape}}
]
\addplot [color=mycolor1, line width=1.0pt]
  table[row sep=crcr]{%
0	167.311358207993\\
0.1	167.324026046956\\
0.2	167.820078181196\\
0.3	168.606695903133\\
0.4	169.78974642499\\
0.5	171.016057484186\\
0.6	172.001214132678\\
0.7	172.914697005952\\
0.8	173.495198531011\\
};
\addlegendentry{flexible \& on-demand}

\addplot [color=mycolor2, line width=1.0pt]
  table[row sep=crcr]{%
0	167.311358207993\\
0.1	168.277735298262\\
0.2	169.085297566662\\
0.3	169.747269431272\\
0.4	170.279785125616\\
0.5	170.725838330216\\
0.6	171.055715461038\\
0.7	171.278633035188\\
0.8	171.404115265038\\
};
\addlegendentry{on-demand only}

\end{axis}
\end{tikzpicture}%
    \vspace*{-0.3in}
    \caption{The arrival rate of ride-sourcing passengers (/min) under different level of delivery demand under on-demand-only case and flexible \& on-demand case.}
    \label{fig:passengerss_same_on_demand}
    \end{minipage}
    \begin{minipage}[b]{0.03\textwidth}
    \hfill
    \end{minipage}
    \begin{minipage}[b]{0.48\textwidth}
    \centering
%
%
\definecolor{mycolor1}{rgb}{0.00000,0.44700,0.74100}%
\definecolor{mycolor2}{rgb}{0.85000,0.32500,0.09800}%
\begin{tikzpicture}

\begin{axis}[%
width=2in,
height=1.4in,
at={(0.729in,0.583in)},
scale only axis,
xmin=0,
xmax=0.8,
xlabel style={font=\color{white!15!black}},
xlabel={Level of delivery demand},
ymin=3400,
ymax=4900,
ylabel style={font=\color{white!15!black}},
ylabel={Total number of drivers},
axis background/.style={fill=white},
legend style={at={(0.03,0.97)}, anchor=north west, legend cell align=left, align=left, draw=white!15!black, nodes={scale=0.8, transform shape}}
]
\addplot [color=mycolor1, line width=1.0pt]
  table[row sep=crcr]{%
0	3586.93052445673\\
0.1	3703.08378664974\\
0.2	3825.29699850254\\
0.3	3955.66050699846\\
0.4	4091.85399355337\\
0.5	4229.03834568042\\
0.6	4360.64125601154\\
0.7	4492.35712014471\\
0.8	4613.58051828568\\
};
\addlegendentry{flexible \& on-demand}

\addplot [color=mycolor2, line width=1.0pt]
  table[row sep=crcr]{%
0	3586.93052445673\\
0.1	3705.57836132982\\
0.2	3822.80650014913\\
0.3	3938.39703546719\\
0.4	4052.36386883936\\
0.5	4166.24216227189\\
0.6	4277.87412810839\\
0.7	4387.16556235464\\
0.8	4494.0642443485\\
};
\addlegendentry{on-demand only}

\end{axis}
\end{tikzpicture}%
    \vspace*{-0.3in}
    \caption{The total number of drivers working on the platform(s) under different level of delivery demand under on-demand-only case and flexible \& on-demand case.}
    \label{fig:drivers_same_on_demand}
    \end{minipage}
    \end{figure}

\subsection{Sensitivity Analysis}
\subsubsection{Perturbing matching functions}

To test the robustness of our results under distinct parameters of the matching functions, we evaluate the change in platform profits, number of drivers on the platform, and passenger arrival rates with increasing potential delivery demand under three instances of the Cobb-Douglas matching function (\ref{eq:matching_ondemand}) with distinct elastic parameters:
    \begin{enumerate}[(i)]
        \item $p=1$, $q=0.5$ (increasing returns to scale): the average waiting time for on-demand customers is:
        \begin{align}\label{w_r_increase}
            w_i^r = \frac{L_i}{\sqrt{N_i^I}}.
        \end{align}
        \item $p=0.5$, $q=0.5$ (constant returns to scale): the average waiting time for on-demand customers is:
        \begin{align}\label{w_r_constant}
            w_i^r = \bar L_i\frac{\sum_{j=1}^M(\lambda_{ij}^r+\lambda_{ij}^{d_o})}{N_i^I}
        \end{align}
        \item $p=0.5$, $q=0.25$ (decreasing returns to scale): the average waiting time for on-demand customers is:
        \begin{align}\label{w_r_decrease}
            w_i^r = \tilde L_i\frac{\sum_{j=1}^M(\lambda_{ij}^r+\lambda_{ij}^{d_o})}{\sqrt{N_i^I}}.
        \end{align}
    \end{enumerate}
    The results are summarized in Figure \ref{fig:profits_p_q} - \ref{fig:passenger_p_q}.  It indicates that as potential delivery demand rises, so do platform profits and the number of drivers. This suggests the integrated business model benefits both the platform and drivers, irrespective of the matching function's elasticity parameters. However, the effects on passengers vary slightly with differing model parameters. Figure \ref{fig:passenger_p_q} shows that when the matching function exhibits increasing or constant returns to scale (i.e., cases (i) and (ii)), passenger costs decrease, thereby increasing ride-sourcing demand as delivery demand escalates. Conversely, when the matching function demonstrates decreasing returns to scale, the passenger arrival rate remains largely unaffected by an increase in delivery demand. We conjecture that this is because when the matching function is decreasing returns to scale, the impacts of the arrival rate on the customer waiting time is much more significant (as shown in case (iii), Equation (\ref{w_r_decrease})), then increased passenger demand can greatly increase the customer's waiting time, which adversely affects the passenger arrival rate. Hence, the passenger arrival rate does not increase as significantly as in the other two cases.

    \begin{figure}[tb!]
    \begin{minipage}[b]{0.3\textwidth}
    \centering
%
%
\definecolor{mycolor1}{rgb}{0.00000,0.44700,0.74100}%
\definecolor{mycolor2}{rgb}{0.85000,0.32500,0.09800}%
\definecolor{mycolor3}{rgb}{0.92900,0.69400,0.12500}%
\begin{tikzpicture}

\begin{axis}[%
width=1.3in,
height=1.6in,
at={(0.729in,0.583in)},
scale only axis,
xmin=0,
xmax=0.8,
xlabel style={font=\color{white!15!black}},
xlabel={Level of delivery demand},
ymin=800,
ymax=2600,
ylabel style={font=\color{white!15!black}},
ylabel={Platform profits},
axis background/.style={fill=white},
legend style={at={(0.97,0.03)}, anchor=south east, legend cell align=left, align=left, draw=white!15!black, nodes={scale=0.6, transform shape}}
]
\addplot [color=mycolor1, line width=1.0pt]
  table[row sep=crcr]{%
0	1193.23021737378\\
0.1	1338.47399195162\\
0.2	1508.54474641935\\
0.3	1679.20752969288\\
0.4	1850.62289152593\\
0.5	2019.89753247542\\
0.6	2184.65607099236\\
0.7	2347.46010021115\\
0.8	2508.73059993365\\
};
\addlegendentry{p=1,q=0.5}

\addplot [color=mycolor2, line width=1.0pt]
  table[row sep=crcr]{%
0	1152.32982267924\\
0.1	1300.81951158684\\
0.2	1440.48760082047\\
0.3	1569.62258094221\\
0.4	1688.0273136384\\
0.5	1798.93602369735\\
0.6	1913.25159977867\\
0.7	2029.03906262223\\
0.8	2141.61500454081\\
};
\addlegendentry{p=0.5,q=0.5}

\addplot [color=mycolor3, line width=1.0pt]
  table[row sep=crcr]{%
0	1063.85933390805\\
0.1	1215.74947060293\\
0.2	1357.27010582111\\
0.3	1508.17613618847\\
0.4	1631.25700857979\\
0.5	1755.6721171329\\
0.6	1871.31331116083\\
0.7	1979.06337377263\\
0.8	2080.01302171024\\
};
\addlegendentry{p=0.5,q=0.25}

\end{axis}
\end{tikzpicture}%
    \vspace*{-0.3in}
    \caption{The profits of the platform (\$/min) under different level of delivery demand.}
    \label{fig:profits_p_q}
    \end{minipage}
    \begin{minipage}[b]{0.03\textwidth}
    \hfill
    \end{minipage}
    \begin{minipage}[b]{0.3\textwidth}
    \centering
%
%
\definecolor{mycolor1}{rgb}{0.00000,0.44700,0.74100}%
\definecolor{mycolor2}{rgb}{0.85000,0.32500,0.09800}%
\definecolor{mycolor3}{rgb}{0.92900,0.69400,0.12500}%
\begin{tikzpicture}

\begin{axis}[%
width=1.3in,
height=1.6in,
at={(0.729in,0.583in)},
scale only axis,
xmin=0,
xmax=0.8,
xlabel style={font=\color{white!15!black}},
xlabel={Level of delivery demand},
ymin=600,
ymax=4800,
ylabel style={font=\color{white!15!black}},
ylabel={Number of drivers},
axis background/.style={fill=white},
legend style={at={(0.97,0.03)}, anchor=south east, legend cell align=left, align=left, draw=white!15!black, nodes={scale=0.6, transform shape}}
]
\addplot [color=mycolor1, line width=1.0pt]
  table[row sep=crcr]{%
0	3586.93052445673\\
0.1	3703.08378664974\\
0.2	3825.29699850254\\
0.3	3955.66050699846\\
0.4	4091.85399355337\\
0.5	4229.03834568042\\
0.6	4360.64125601154\\
0.7	4492.35712014471\\
0.8	4613.58051828568\\
};
\addlegendentry{p=1,q=0.5}

\addplot [color=mycolor2, line width=1.0pt]
  table[row sep=crcr]{%
0	2858.56733295984\\
0.1	2972.97191962651\\
0.2	3103.00369568455\\
0.3	3228.76445408287\\
0.4	3355.19997147429\\
0.5	3486.42495430238\\
0.6	3627.81246937046\\
0.7	3753.60973906537\\
0.8	3884.54350697233\\
};
\addlegendentry{p=0.5,q=0.5}

\addplot [color=mycolor3, line width=1.0pt]
  table[row sep=crcr]{%
0	1772.60967886192\\
0.1	1858.98256459651\\
0.2	1911.51480888152\\
0.3	2024.33995268074\\
0.4	2082.05312065028\\
0.5	2202.98386101754\\
0.6	2323.96572008508\\
0.7	2438.80382220689\\
0.8	2549.41365659861\\
};
\addlegendentry{p=0.5,q=0.25}

\end{axis}
\end{tikzpicture}%
    \vspace*{-0.3in}
    \caption{The total number of drivers under different level of delivery demand.}
    \label{fig:drivers_p_q}
    \end{minipage}
    \begin{minipage}[b]{0.03\textwidth}
    \hfill
    \end{minipage}
    \begin{minipage}[b]{0.3\textwidth}
    \centering
%
%
\definecolor{mycolor1}{rgb}{0.00000,0.44700,0.74100}%
\definecolor{mycolor2}{rgb}{0.85000,0.32500,0.09800}%
\definecolor{mycolor3}{rgb}{0.92900,0.69400,0.12500}%
\begin{tikzpicture}

\begin{axis}[%
width=1.3in,
height=1.6in,
at={(0.729in,0.583in)},
scale only axis,
xmin=0,
xmax=0.8,
xlabel style={font=\color{white!15!black}},
xlabel={Level of delivery demand},
ymin=20,
ymax=190,
ylabel style={font=\color{white!15!black}},
ylabel={Passenger arrival rate},
axis background/.style={fill=white},
legend style={at={(0.97,0.03)}, anchor=south east, legend cell align=left, align=left, draw=white!15!black, nodes={scale=0.6, transform shape}}
]
\addplot [color=mycolor1, line width=1.0pt]
  table[row sep=crcr]{%
0	167.311358207993\\
0.1	167.324026046956\\
0.2	167.820078181196\\
0.3	168.606695903133\\
0.4	169.78974642499\\
0.5	171.016057484186\\
0.6	172.001214132678\\
0.7	172.914697005952\\
0.8	173.495198531011\\
};
\addlegendentry{p=1,q=0.5}

\addplot [color=mycolor2, line width=1.0pt]
  table[row sep=crcr]{%
0	108.754163094824\\
0.1	110.638366812742\\
0.2	112.863855007047\\
0.3	113.777888799065\\
0.4	114.356561599684\\
0.5	115.297083849742\\
0.6	116.957909360002\\
0.7	117.223692989524\\
0.8	117.403259710235\\
};
\addlegendentry{p=0.5,q=0.5}

\addplot [color=mycolor3, line width=1.0pt]
  table[row sep=crcr]{%
0	76.0614305870863\\
0.1	77.2457330960522\\
0.2	76.2738199346913\\
0.3	76.675588981775\\
0.4	75.94786197492\\
0.5	75.8711555842308\\
0.6	75.5903433992803\\
0.7	75.0981175991613\\
0.8	74.4979500892596\\
};
\addlegendentry{p=0.5,q=0.25}

\end{axis}
\end{tikzpicture}%
    \vspace*{-0.3in}
    \caption{The arrival rate of total passengers (/min) under different level of delivery demand.}
    \label{fig:passenger_p_q}
    \end{minipage}
    \end{figure}

\subsubsection{Parcel demand in the opposite direction}

To test the robustness of our proposed framework, we conduct the case study under another extreme demand pattern where the potential delivery demand and the ride-sourcing demand are in opposite directions (i.e., the potential delivery demand is inversely proportional to the potential ride-sourcing demand).
In this case, we rank all origin-destination (OD) pairs of ride-sourcing demand based on their magnitude  and then reshuffle their order to generate parcel delivery demand, so that the OD pair with the largest demand in ride-sourcing market after reshuffling becomes the OD pair with the smallest demand in the parcel delivery market.

To evaluate the impacts of the proposed model on the ride-sourcing market, we gradually increase the level of parcel delivery demand, which is defined to be the ratio between the delivery demand to the inverse of the ride-sourcing demand. We perturb the demand level from 0 to 0.8. The corresponding platform profits, total number of drivers on the platform, and total arrival rate of ride-sourcing passengers are presented in Figure \ref{fig:profits_oppo_onlyboth}-\ref{fig:passenger_oppo_onlyboth}.
The results indicate that the platform and drivers can benefit from the integrated business model under the opposite-direction demand pattern, but passengers will only receive benefits when the delivery demand is not large. We comment that the benefits to passengers, drivers, and the platform are intuitive because the proposed model brings more revenues to the platform, leads to better utilization of driver's idle time, which reduces the costs for drivers and potentially benefit ride-sourcing passengers. However, since the delivery order is not in the same direction with ride-sourcing orders, when the number of delivery order is large, the integration of delivery serves may change the spatial distribution of the drivers and motivate significant amount of drivers to spend time in remote areas where ride-sourcing demand is low but delivery demand is high. In this case, since the ride-sourcing passengers are mainly concentrated in the city center, the change of driver's spatial distribution may negatively affects the passengers.

\begin{figure}[t!]
\begin{minipage}[b]{0.27\textwidth}
\centering
%
%
\definecolor{mycolor1}{rgb}{0.00000,0.44700,0.74100}%
\definecolor{mycolor2}{rgb}{0.85000,0.32500,0.09800}%
\begin{tikzpicture}

\begin{axis}[%
width=1.3in,
height=1.6in,
at={(0.729in,0.583in)},
scale only axis,
xmin=0,
xmax=0.8,
xlabel style={font=\color{white!15!black}},
xlabel={Level of delivery demand},
ymin=1000,
ymax=3200,
ylabel style={font=\color{white!15!black}},
ylabel={Platform profits},
axis background/.style={fill=white},
legend style={at={(0.03,0.97)}, anchor=north west, legend cell align=left, align=left, draw=white!15!black}
]
\addplot [color=mycolor1, line width=1.0pt]
  table[row sep=crcr]{%
0	1193.23021737378\\
0.1	1309.16969034233\\
0.2	1473.92373928806\\
0.3	1788.45343075055\\
0.4	2054.51203468242\\
0.5	2322.65071759292\\
0.6	2592.05282655\\
0.7	2862.53744865801\\
0.8	3134.25979215423\\
};

\end{axis}
\end{tikzpicture}%
\vspace*{-0.3in}
\caption{The profits of the platform (\$/min) under different delivery demand size with opposite directions.}
\label{fig:profits_oppo_onlyboth}
\end{minipage}
\begin{minipage}[t]{0.03\textwidth}
\hfill
\end{minipage}
\begin{minipage}[b]{0.27\textwidth}
\centering
%
%
\definecolor{mycolor1}{rgb}{0.00000,0.44700,0.74100}%
\definecolor{mycolor2}{rgb}{0.85000,0.32500,0.09800}%
\begin{tikzpicture}

\begin{axis}[%
width=1.3in,
height=1.6in,
at={(0.729in,0.583in)},
scale only axis,
xmin=0,
xmax=0.8,
xlabel style={font=\color{white!15!black}},
xlabel={Level of delivery demand},
ymin=3400,
ymax=4700,
ylabel style={font=\color{white!15!black}},
ylabel={Number of drivers},
axis background/.style={fill=white},
legend style={at={(0.03,0.97)}, anchor=north west, legend cell align=left, align=left, draw=white!15!black}
]
\addplot [color=mycolor1, line width=1.0pt]
  table[row sep=crcr]{%
0	3586.93052445673\\
0.1	3678.92777461182\\
0.2	3841.55541556484\\
0.3	4036.20237338647\\
0.4	4178.97353379793\\
0.5	4308.4677347341\\
0.6	4432.68454882366\\
0.7	4548.29775216884\\
0.8	4662.83091648347\\
};

\end{axis}
\end{tikzpicture}%
\vspace*{-0.3in}
\caption{The total number of drivers under different delivery demand size with opposite directions.}
\label{fig:drivers_oppo_onlyboth}
\end{minipage}
\begin{minipage}[t]{0.03\textwidth}
\hfill
\end{minipage}
\begin{minipage}[b]{0.27\textwidth}
\centering
%
%
\definecolor{mycolor1}{rgb}{0.00000,0.44700,0.74100}%
\definecolor{mycolor2}{rgb}{0.85000,0.32500,0.09800}%
\begin{tikzpicture}

\begin{axis}[%
width=1.3in,
height=1.6in,
at={(0.729in,0.583in)},
scale only axis,
xmin=0,
xmax=0.8,
xlabel style={font=\color{white!15!black}},
xlabel={Level of delivery demand},
ymin=166,
ymax=180,
ylabel style={font=\color{white!15!black}},
ylabel={Passenger arrival rate},
axis background/.style={fill=white},
legend style={at={(0.03,0.97)}, anchor=north west, legend cell align=left, align=left, draw=white!15!black}
]
\addplot [color=mycolor1, line width=1.0pt]
  table[row sep=crcr]{%
0	167.311358207993\\
0.1	168.950090438665\\
0.2	173.822785482057\\
0.3	178.227097973819\\
0.4	178.139276067127\\
0.5	177.787650692898\\
0.6	177.12933146037\\
0.7	175.968717929863\\
0.8	174.653558230395\\
};

\end{axis}
\end{tikzpicture}%
\vspace*{-0.3in}
\caption{Arrival rate of passengers (/min) under different delivery demand size with opposite directions.}
\label{fig:passenger_oppo_onlyboth}
\end{minipage}

     \centering
     \begin{subfigure}[The potential demand (/min) of ride-sourcing passengers.]{
         \centering
         \includegraphics[width=0.47\textwidth]{Figure_V2/potential_passenger_demand.pdf}
         \label{fig:potential_passenger_demand_oppo}}
     \end{subfigure}
     \begin{subfigure}[The potential demand of delivery customers (/min) under opposite directions.]{
         \centering
         \includegraphics[width=0.47\textwidth]{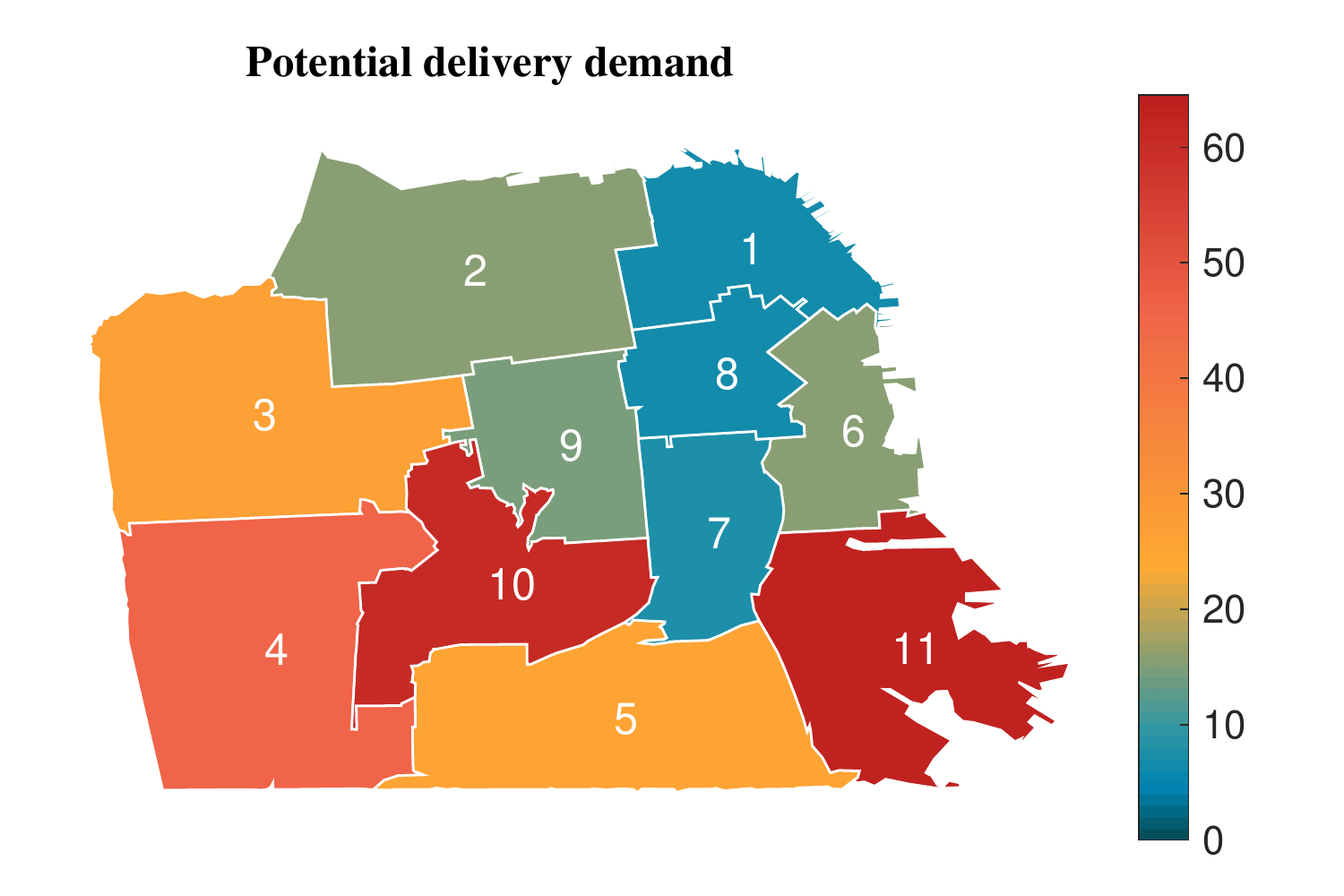}
         \label{fig:potential_delivery_demand_oppo}}
     \end{subfigure}
        \caption{Potential demand of ride-sourcing passengers and delivery customers under the opposite directions.}
        \label{fig:potential_demand_oppo}

     \centering
     \begin{subfigure}[The change of idle drivers for each zone as a ratio of the number of idle drivers at ride-only case]{
         \centering
         \includegraphics[width=0.47\textwidth]{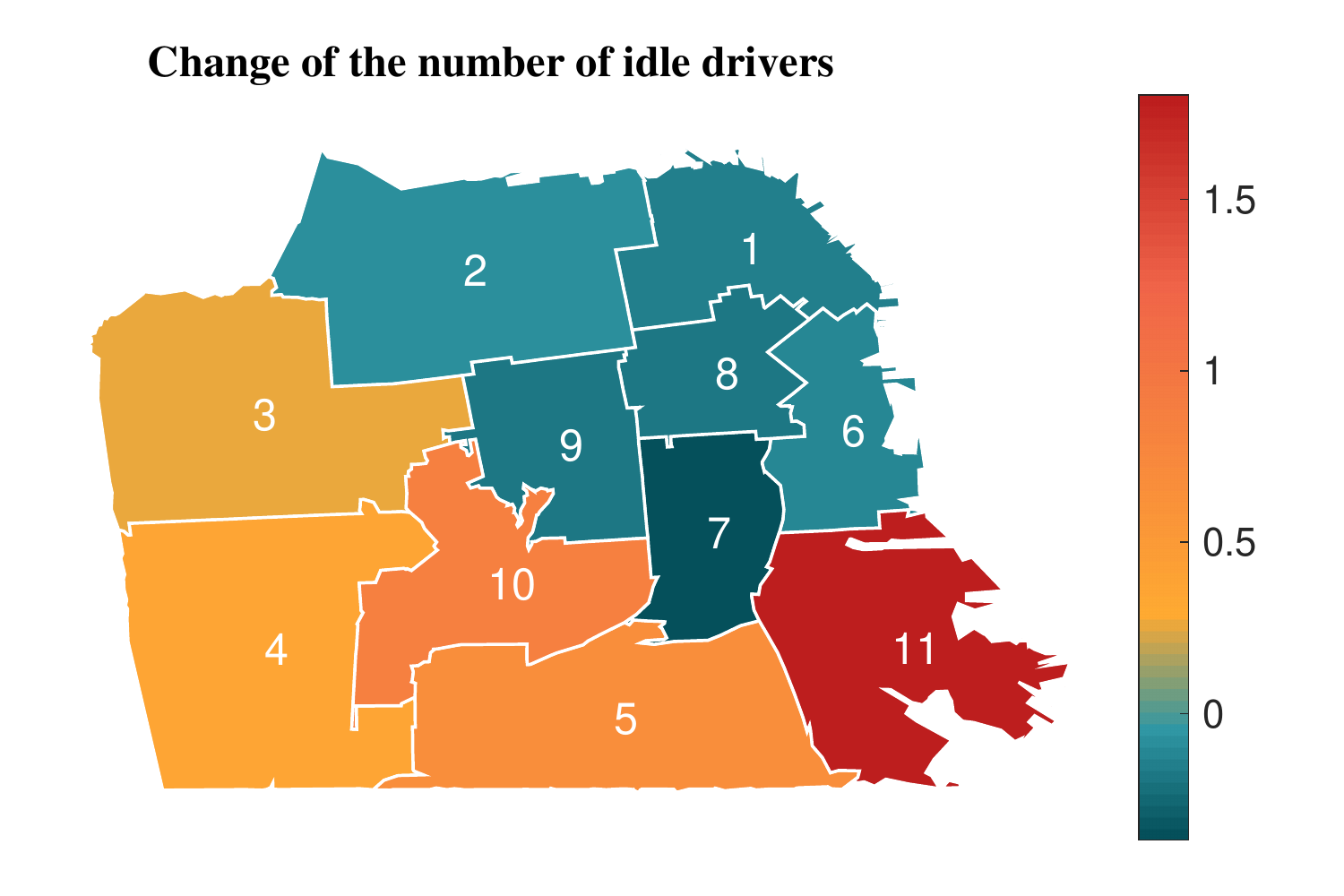}
         \label{fig:idle_drivers_oppo}}
     \end{subfigure}
     \begin{subfigure}[The change of passengers arrival rate for each zone as a ratio of the passenger arrival rate at ride-only case.]{
         \centering
         \includegraphics[width=0.47\textwidth]{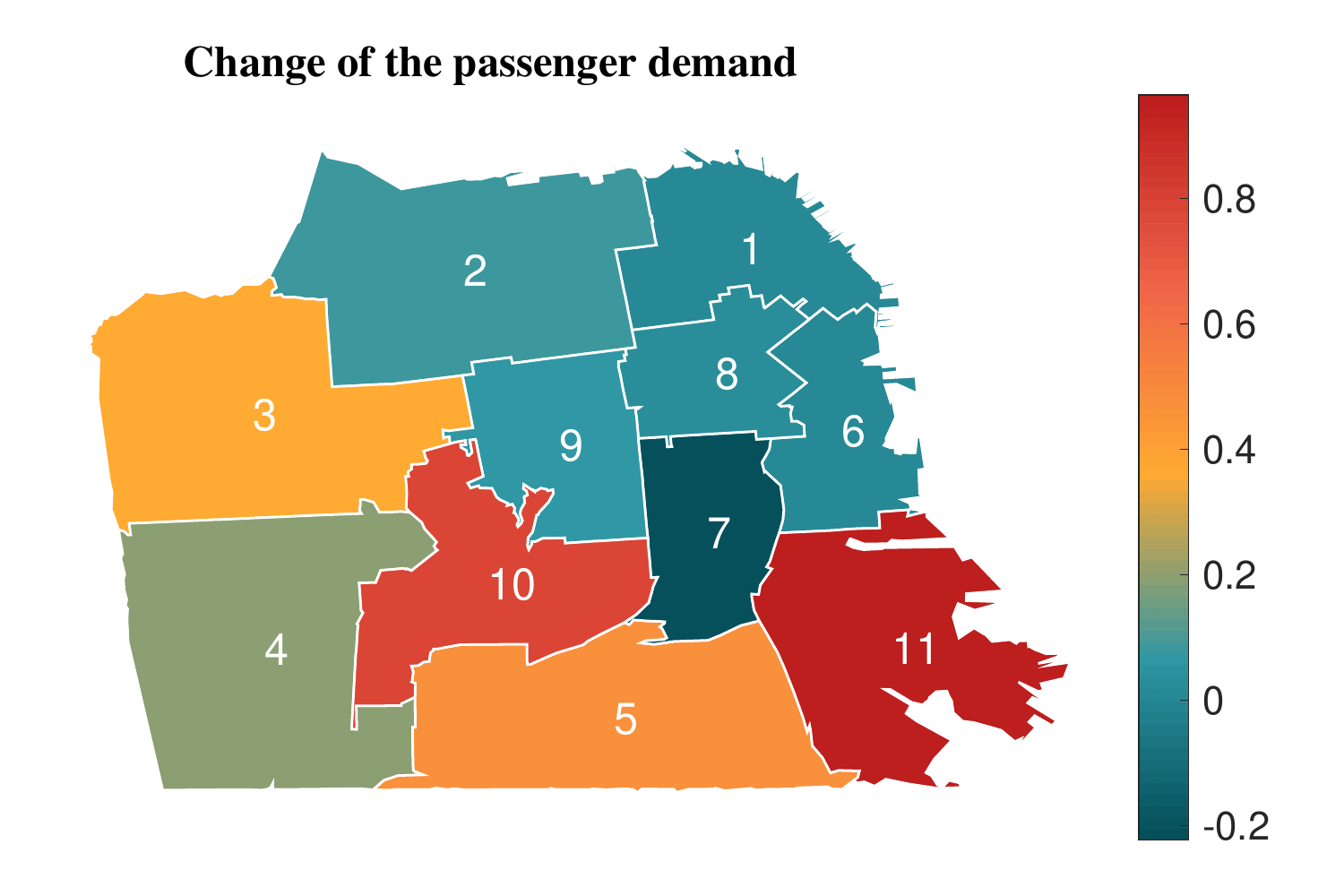}
         \label{fig:passengers_oppo}}
     \end{subfigure}
        \caption{The change of idle drivers and passenger demand for each zone in percentage.}
        \label{fig:idle_drivers_passengers_oppo}
\end{figure}

\begin{figure}[t]
     \centering
     \begin{subfigure}[The demand for flexible delivery (/min) under opposite direction (as destination).]{
         \centering
         \includegraphics[width=0.47\textwidth]{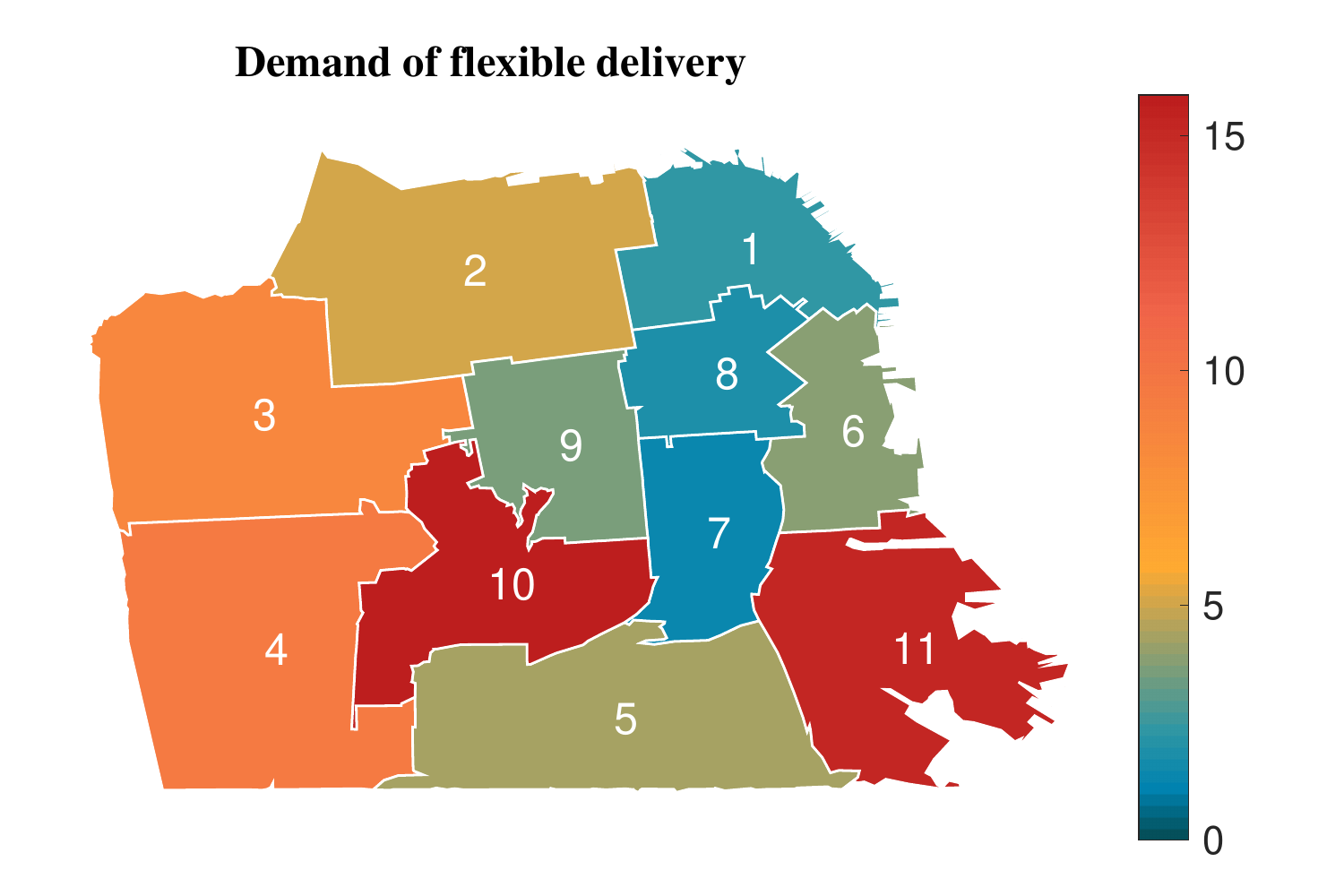}
         \label{fig:demand_flexible_oppo}}
     \end{subfigure}
     \begin{subfigure}[The demand for on-demand delivery (/min) under opposite direction (as destination).]{
         \centering
         \includegraphics[width=0.47\textwidth]{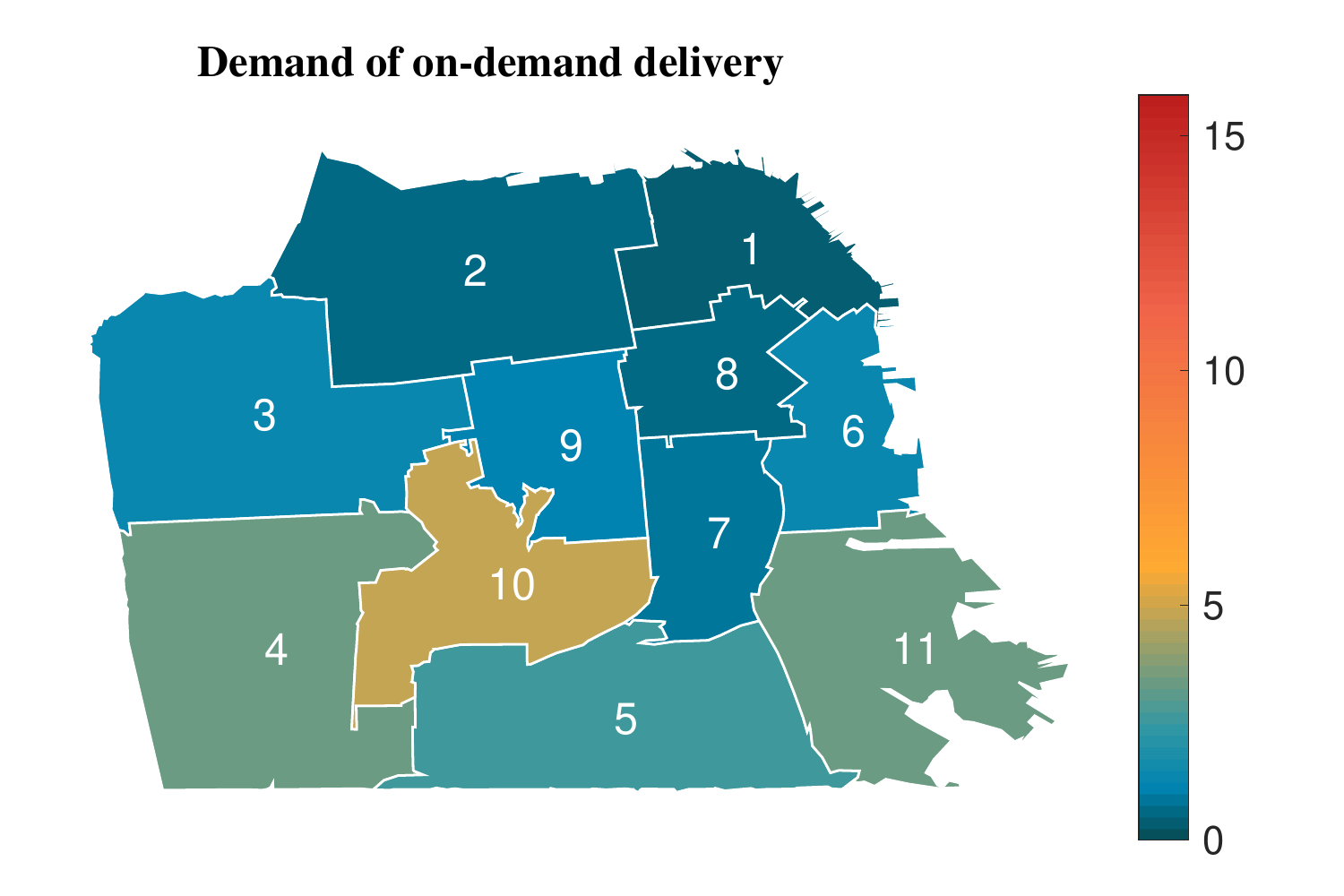}
         \label{fig:demand_fix_oppo}}
     \end{subfigure}
        \caption{Demand for delivery services with destination to each zone under opposite directions.}
        \label{fig:demand_oppo}
\end{figure}

\begin{figure}[t]
     \centering
     \begin{subfigure}[The demand for flexible delivery (/min) under opposite direction (as origin).]{
         \centering
         \includegraphics[width=0.47\textwidth]{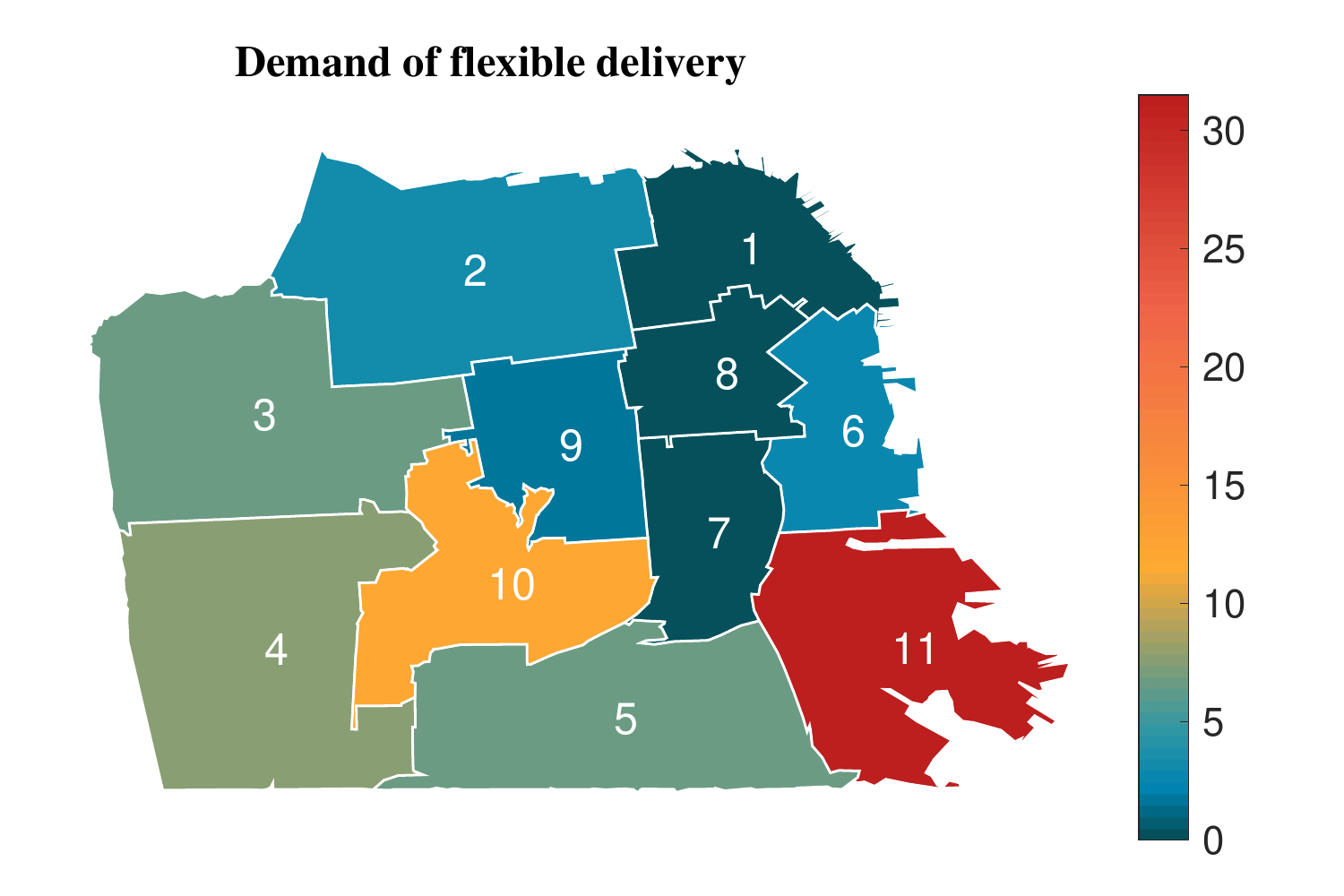}
         \label{fig:demand_flexible_oppo_origin}}
     \end{subfigure}
     \begin{subfigure}[The demand for on-demand delivery (/min) under opposite direction (as origin).]{
         \centering
         \includegraphics[width=0.47\textwidth]{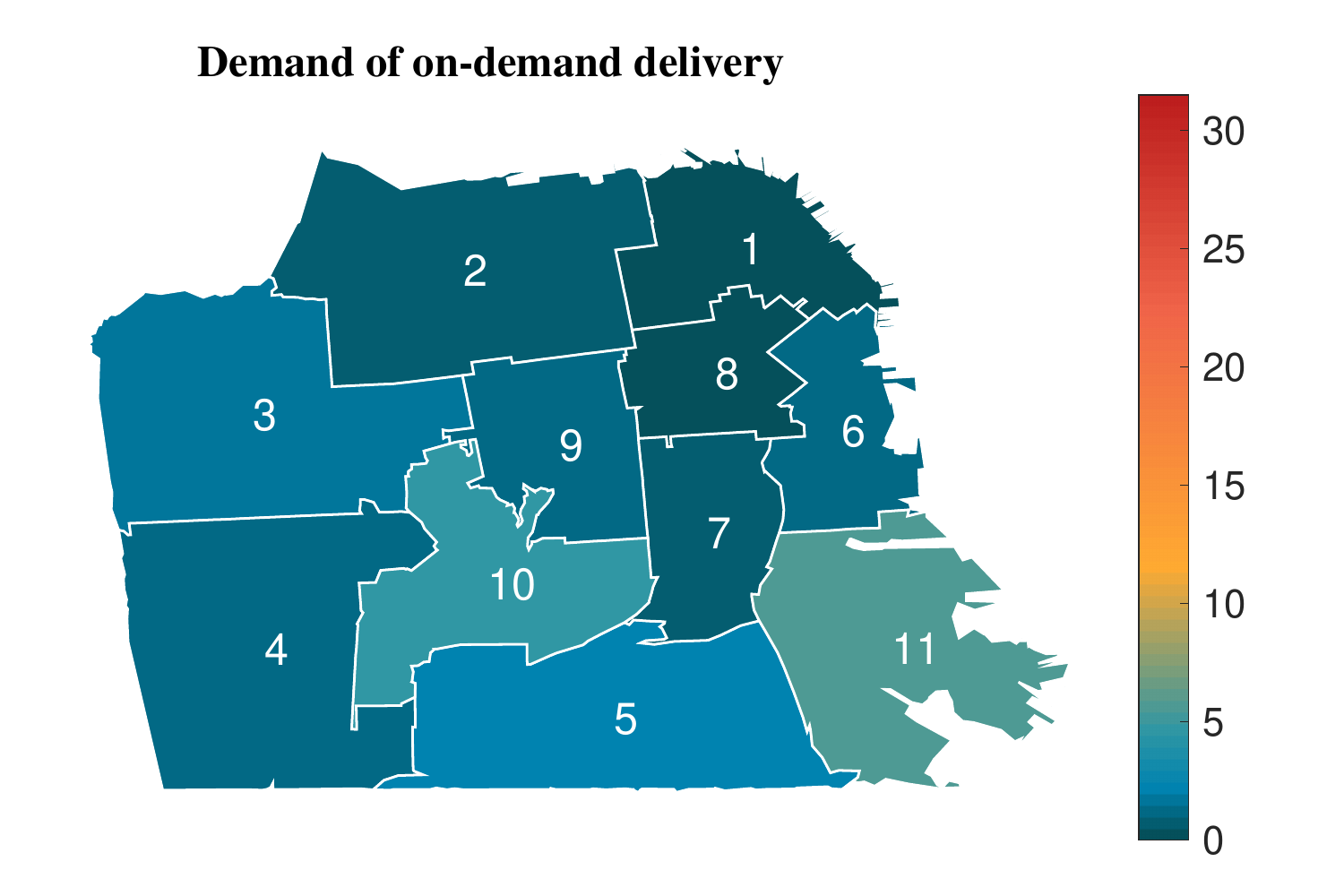}
         \label{fig:demand_fix_oppo_origin}}
     \end{subfigure}
        \caption{Demand for delivery services with origin from each zone under opposite directions.}
        \label{fig:demand_oppo_origin}
\end{figure}

To further understand the market outcomes over the transportation network, we fix the level of delivery demand, and demonstrate the results at the zonal level.  Similar to the case of same direction, we choose the ratio of delivery demand to ride-sourcing demand to be 0.4 (e.g. $\lambda_{ij}^{d,0} = 0.4\lambda_{ij}^{r,0}$) as an example.
The spatial distributions of potential ride-sourcing demand and potential delivery demand are shown in Figure \ref{fig:potential_demand_oppo}.
The value corresponding to each zone is its total delivery demand as destination (e.g. the value for zone $i$ is $\sum_{j=1}^M\lambda_{ji}^{d,0}$).
Figure \ref{fig:potential_demand_oppo} illustrates that the city center (i.e., the northeast of the city) has highest ride-sourcing demand, but lowest parcel-delivery demand.

Figure \ref{fig:idle_drivers_passengers_oppo} shows the impacts of the integrated model on the ride-sourcing market at the zonal level under the case of opposite directions.
Specifically, Figure \ref{fig:idle_drivers_oppo} and \ref{fig:passengers_oppo} demonstrate the change of idle drivers and arrival rate of ride-sourcing passengers in each zone in percentage compared to the case of ride-sourcing services only, respectively.
We observe that the number of idle drivers in high-demand areas (zone 1, 2, 6-9) decreases, while the number of idle drivers in lower-demand areas increases (Figure \ref{fig:idle_drivers_oppo}).
For the remote zones with less ride-sourcing passenger demand (e.g., zone 3, 4, 5, 10 and 11), the delivery customer demand are relatively high, thus the platform is motivated to dispatch idle drivers in high-demand areas to move to the these remote areas so that the drivers supply can meet the expansive delivery customer demand. The ride-sourcing demand in the remote areas also experience a large increase as shown in Figure \ref{fig:passengers_oppo}. This is consistent with the fact that the increase of idle drivers in these zones can reduce passengers' waiting time and attract more passengers to use ride-sourcing services. The above observation highlights that within certain regime, delivery demand in the opposite direction with the ride-sourcing demand may benefit all stakeholders  in lower demand areas by increased platform profits, idle drivers and arrival rate of passengers at the expense of reduced passenger demand in higher-demand areas. At the same time it could improve the spatial equity of ride-sourcing services over the transportation networks.


We also evaluate the outcomes of the delivery market at the zonal level by presenting the split of orders between  on-demand and flexible delivery orders in each zone. 
Figure \ref{fig:demand_oppo} represents the total delivery order attractions to each zone, and  Figure \ref{fig:demand_oppo_origin} represents the total delivery order productions corresponding to each zone.
Interestingly, the results are significantly different from that in the realistic case. In particular, the results of Figure \ref{fig:demand_oppo} indicate that when looking at the destination of the delivery order, no matter where the destination is, both flexible and on-demand services are available  for all customers, and there are more customers choosing to use flexible delivery than on-demand delivery in all zones including the remote areas.
We comment that this is because of a large number of delivery orders was sent out from the remote areas, which increase the probability for drivers to arrive at remote zones, shortening the delivery time of flexible delivery orders. Under a lower price, it makes the flexible delivery services more attractable compared to the on-demand delivery services.
On the other hand, Figure \ref{fig:demand_oppo_origin} illustrates that when we look at origins of the delivery orders, the platform gives up the flexible orders sent from high-demand areas (zone 1, 2, 6, 7, 8, 9), thus customers who send parcels from these areas can only use on-demand services.This is because in those zones, the high ride-sourcing demand attracts a large number of idle drivers, but under the opposite-direction demand pattern, the delivery demand in those areas is relatively low. The waiting time for idle drivers to be matched with a flexible delivery order is too long such that the successful rate for a idle driver to pick up a flexible delivery parcel before be dispatched to a ride-souring order is very small. Therefore, the expensive cost of the long waiting time motivates customers to choose on-demand delivery services instead of flexible services.
This is contrary to the case of same direction, where the platform only offers on-demand services to remote areas when we look at destination, but the flexible delivery services are available everywhere when we look at the origin.

\section{Conclusions}

This paper investigates the optimal pricing and fleet management strategies for an integrated platform to provide integrated ride-sourcing and intracity package delivery services over a transportation network utilizing the idle time of ride-sourcing drivers. 
We consider an integration of ride-sourcing services and two modes of parcel delivery services on a single platform. In this model, parcel delivery services consist of (1) on-demand delivery, in which drivers must immediately pick up and deliver the goods when the delivery order is placed; and (2) flexible delivery, in which drivers can pick up (or drop off) the parcel only when drivers are idle and waiting for the next ride-sourcing request, and close to the origin or destination of the delivery order. A continuous-time Markov Chain (CTMC) model is proposed to capture the simultaneous movement of passengers and parcels across a transportation network with limited vehicle capacity. The model allows for the characterization of the service quality of ride-sourcing services, on-demand delivery services, and flexible delivery services, as well as the economic equilibrium arising from the incentives of ride-sourcing passengers, delivery customers, drivers, and the platform. The platform's profit maximization problem is formulated as a non-convex optimization. We utilize the structure of interdependence between distinct endogenous variables to prove the well-posedness of the model and develop a customized algorithm that can rapidly compute the optimal platform decisions. The proposed model is validated by a comprehensive numerical studied. We show  that the joint management of ride-sourcing services and intracity package delivery services can result in a Pareto improvement that is beneficial to all market participants under realistic ride-sourcing and parcel delivery demand patterns.

This work can be extended along several directions. First, the static model can be extended to capture temporal dynamics of the demand and supply, and provide distinct managerial strategies at different time  of a day. In addition, a more comprehensive matching and routing mechanisms for drivers can be examined in order to capture more details of the transportation network. Finally, this work can be extended to a more complex multi-modal transportation network taking public transit as well as other forms of shared mobility into account.

\section*{Acknowledgments}  {This research was supported by Hong Kong Research Grants Council under project 26200420.}

\bibliographystyle{unsrt}
\bibliography{references}

\begin{thebibliography}{10}

\bibitem{cherry2012truck}
Christopher~R Cherry and Adebola~A Adelakun.
\newblock Truck driver perceptions and preferences: Congestion and conflict,
  managed lanes, and tolls.
\newblock {\em Transport Policy}, 24:1--9, 2012.

\bibitem{kafle2017design}
Nabin Kafle, Bo~Zou, and Jane Lin.
\newblock Design and modeling of a crowdsource-enabled system for urban parcel
  relay and delivery.
\newblock {\em Transportation research part B: methodological}, 99:62--82,
  2017.

\bibitem{jacobs2019lastmile}
Kees Jacobs, Shannon Warner, Marc Rietra, Lindsey Mazza, Jerome Buvat, Amol
  Khadikar, Sumit Cherian, and Yashwardhan Khemka.
\newblock The last-mile delivery challenge.
\newblock {\em Capgemini Research Institute}, 2019.

\bibitem{sun2019optimal}
Luoyi Sun, Ruud~H Teunter, M~Zied Babai, and Guowei Hua.
\newblock Optimal pricing for ride-sourcing platforms.
\newblock {\em European Journal of Operational Research}, 278(3):783--795,
  2019.

\bibitem{wang2016pricing}
Xiaolei Wang, Fang He, Hai Yang, and H~Oliver Gao.
\newblock Pricing strategies for a taxi-hailing platform.
\newblock {\em Transportation Research Part E: Logistics and Transportation
  Review}, 93:212--231, 2016.

\bibitem{zha2016economic}
Liteng Zha, Yafeng Yin, and Hai Yang.
\newblock Economic analysis of ride-sourcing markets.
\newblock {\em Transportation Research Part C: Emerging Technologies},
  71:249--266, 2016.

\bibitem{xu2021generalized}
Zhengtian Xu, Yafeng Yin, Xiuli Chao, Hongtu Zhu, and Jieping Ye.
\newblock A generalized fluid model of ride-hailing systems.
\newblock {\em Transportation Research Part B: Methodological}, 150:587--605,
  2021.

\bibitem{bai2019coordinating}
Jiaru Bai, Kut~C So, Christopher~S Tang, Xiqun Chen, and Hai Wang.
\newblock Coordinating supply and demand on an on-demand service platform with
  impatient customers.
\newblock {\em Manufacturing \& Service Operations Management}, 21(3):556--570,
  2019.

\bibitem{taylor2018demand}
Terry~A Taylor.
\newblock On-demand service platforms.
\newblock {\em Manufacturing \& Service Operations Management}, 20(4):704--720,
  2018.

\bibitem{hu2020price}
Ming Hu and Yun Zhou.
\newblock Price, wage, and fixed commission in on-demand matching.
\newblock {\em Available at SSRN 2949513}, 2020.

\bibitem{zhu2021mean}
Zheng Zhu, Jintao Ke, and Hai Wang.
\newblock A mean-field markov decision process model for spatial-temporal
  subsidies in ride-sourcing markets.
\newblock {\em Transportation Research Part B: Methodological}, 150:540--565,
  2021.

\bibitem{zha2018geometric}
Liteng Zha, Yafeng Yin, and Zhengtian Xu.
\newblock Geometric matching and spatial pricing in ride-sourcing markets.
\newblock {\em Transportation Research Part C: Emerging Technologies},
  92:58--75, 2018.

\bibitem{chen2021spatial}
Chuqiao Chen, Fugen Yao, Dong Mo, Jiangtao Zhu, and Xiqun~Michael Chen.
\newblock Spatial-temporal pricing for ride-sourcing platform with
  reinforcement learning.
\newblock {\em Transportation Research Part C: Emerging Technologies},
  130:103272, 2021.

\bibitem{cachon2017role}
Gerard~P Cachon, Kaitlin~M Daniels, and Ruben Lobel.
\newblock The role of surge pricing on a service platform with self-scheduling
  capacity.
\newblock {\em Manufacturing \& Service Operations Management}, 19(3):368--384,
  2017.

\bibitem{banerjee2015pricing}
Siddhartha Banerjee, Carlos Riquelme, and Ramesh Johari.
\newblock Pricing in ride-share platforms: A queueing-theoretic approach.
\newblock {\em Available at SSRN 2568258}, 2015.

\bibitem{zha2017surge}
Liteng Zha, Yafeng Yin, and Yuchuan Du.
\newblock Surge pricing and labor supply in the ride-sourcing market.
\newblock {\em Transportation Research Procedia}, 23:2--21, 2017.

\bibitem{li2021spatial}
Sen Li, Hai Yang, Kameshwar Poolla, and Pravin Varaiya.
\newblock Spatial pricing in ride-sourcing markets under a congestion charge.
\newblock {\em Transportation Research Part B: Methodological}, 152:18--45,
  2021.

\bibitem{tang2022bi}
Wei Tang, Heng Wang, Yang Wang, Chuqiao Chen, Xiqun~Michael Chen, et~al.
\newblock A bi-level optimization model for ride-sourcing platform’s spatial
  pricing strategy.
\newblock {\em Journal of Advanced Transportation}, 2022, 2022.

\bibitem{chen2020dynamic}
Xiqun~Michael Chen, Hongyu Zheng, Jintao Ke, and Hai Yang.
\newblock Dynamic optimization strategies for on-demand ride services platform:
  Surge pricing, commission rate, and incentives.
\newblock {\em Transportation Research Part B: Methodological}, 138:23--45,
  2020.

\bibitem{nourinejad2020ride}
Mehdi Nourinejad and Mohsen Ramezani.
\newblock Ride-sourcing modeling and pricing in non-equilibrium two-sided
  markets.
\newblock {\em Transportation Research Part B: Methodological}, 132:340--357,
  2020.

\bibitem{wang2019ridesourcing}
Hai Wang and Hai Yang.
\newblock Ridesourcing systems: A framework and review.
\newblock {\em Transportation Research Part B: Methodological}, 129:122--155,
  2019.

\bibitem{zhang2021pool}
Kenan Zhang and Yu~Marco Nie.
\newblock To pool or not to pool: Equilibrium, pricing and regulation.
\newblock {\em Transportation Research Part B: Methodological}, 151:59--90,
  2021.

\bibitem{bahrami2022optimal}
Sina Bahrami, Mehdi Nourinejad, Mahmood~Mahmoodi Nesheli, and Yafeng Yin.
\newblock Optimal composition of solo and pool services for on-demand
  ride-hailing.
\newblock {\em Transportation Research Part E: Logistics and Transportation
  Review}, 161:102680, 2022.

\bibitem{ke2020pricing}
Jintao Ke, Hai Yang, Xinwei Li, Hai Wang, and Jieping Ye.
\newblock Pricing and equilibrium in on-demand ride-pooling markets.
\newblock {\em Transportation Research Part B: Methodological}, 139:411--431,
  2020.

\bibitem{ke2022coordinating}
Jintao Ke, Xiqun~Michael Chen, Hai Yang, and Sen Li.
\newblock Coordinating supply and demand in ride-sourcing markets with
  pre-assigned pooling service and traffic congestion externality.
\newblock {\em Transportation Research Part E: Logistics and Transportation
  Review}, 166:102887, 2022.

\bibitem{taniguchi2000evaluation}
Eiichi Taniguchi and Rob~ECM Van Der~Heijden.
\newblock An evaluation methodology for city logistics.
\newblock {\em Transport Reviews}, 20(1):65--90, 2000.

\bibitem{quak2009delivering}
HJ~Quak and M(Ren{\'e})~BM de~Koster.
\newblock Delivering goods in urban areas: how to deal with urban policy
  restrictions and the environment.
\newblock {\em Transportation science}, 43(2):211--227, 2009.

\bibitem{crainic2009models}
Teodor~Gabriel Crainic, Nicoletta Ricciardi, and Giovanni Storchi.
\newblock Models for evaluating and planning city logistics systems.
\newblock {\em Transportation science}, 43(4):432--454, 2009.

\bibitem{cattaruzza2017vehicle}
Diego Cattaruzza, Nabil Absi, Dominique Feillet, and Jes{\'u}s
  Gonz{\'a}lez-Feliu.
\newblock Vehicle routing problems for city logistics.
\newblock {\em EURO Journal on Transportation and Logistics}, 6(1):51--79,
  2017.

\bibitem{taniguchi2004intelligent}
Eiichi Taniguchi and Hiroshi Shimamoto.
\newblock Intelligent transportation system based dynamic vehicle routing and
  scheduling with variable travel times.
\newblock {\em Transportation Research Part C: Emerging Technologies},
  12(3-4):235--250, 2004.

\bibitem{liu2019minimizing}
Bingbing Liu, Xiaolong Guo, Yugang Yu, and Qiang Zhou.
\newblock Minimizing the total completion time of an urban delivery problem
  with uncertain assembly time.
\newblock {\em Transportation Research Part E: Logistics and Transportation
  Review}, 132:163--182, 2019.

\bibitem{gross2019cost}
Patrick-Oliver Gro{\ss}, Jan~F Ehmke, and Dirk~C Mattfeld.
\newblock Cost-efficient and reliable city logistics vehicle routing with
  satellite locations under travel time uncertainty.
\newblock {\em Transportation Research Procedia}, 37:83--90, 2019.

\bibitem{gayialis2022city}
Sotiris~P Gayialis, Evripidis~P Kechagias, and Grigorios~D Konstantakopoulos.
\newblock A city logistics system for freight transportation: Integrating
  information technology and operational research.
\newblock {\em Operational Research}, 22(5):5953--5982, 2022.

\bibitem{janjevic2020characterizing}
Milena Janjevic and Matthias Winkenbach.
\newblock Characterizing urban last-mile distribution strategies in mature and
  emerging e-commerce markets.
\newblock {\em Transportation Research Part A: Policy and Practice},
  133:164--196, 2020.

\bibitem{afeche2016optimal}
Philipp Afeche and J~Michael Pavlin.
\newblock Optimal price/lead-time menus for queues with customer choice:
  Segmentation, pooling, and strategic delay.
\newblock {\em Management Science}, 62(8):2412--2436, 2016.

\bibitem{savelsbergh201650th}
Martin Savelsbergh and Tom Van~Woensel.
\newblock 50th anniversary invited article—city logistics: Challenges and
  opportunities.
\newblock {\em Transportation Science}, 50(2):579--590, 2016.

\bibitem{konstantakopoulos2020vehicle}
Grigorios~D Konstantakopoulos, Sotiris~P Gayialis, and Evripidis~P Kechagias.
\newblock Vehicle routing problem and related algorithms for logistics
  distribution: a literature review and classification.
\newblock {\em Operational research}, pages 1--30, 2020.

\bibitem{archetti2016vehicle}
Claudia Archetti, Martin Savelsbergh, and M~Grazia Speranza.
\newblock The vehicle routing problem with occasional drivers.
\newblock {\em European Journal of Operational Research}, 254(2):472--480,
  2016.

\bibitem{raviv2018crowd}
Tal Raviv and Eyal~Z Tenzer.
\newblock Crowd-shipping of small parcels in a physical internet.
\newblock {\em Workingpaper, Tel Aviv University}, 2018.

\bibitem{yildiz2021express}
Bar{\i}{\c{s}} Y{\i}ld{\i}z.
\newblock Express package routing problem with occasional couriers.
\newblock {\em Transportation Research Part C: Emerging Technologies},
  123:102994, 2021.

\bibitem{dayarian2020crowdshipping}
Iman Dayarian and Martin Savelsbergh.
\newblock Crowdshipping and same-day delivery: Employing in-store customers to
  deliver online orders.
\newblock {\em Production and Operations Management}, 29(9):2153--2174, 2020.

\bibitem{ahamed2021deep}
Tanvir Ahamed, Bo~Zou, Nahid~Parvez Farazi, and Theja Tulabandhula.
\newblock Deep reinforcement learning for crowdsourced urban delivery.
\newblock {\em Transportation Research Part B: Methodological}, 152:227--257,
  2021.

\bibitem{mousavi2022stochastic}
Kianoush Mousavi, Merve Bodur, and Matthew~J Roorda.
\newblock Stochastic last-mile delivery with crowd-shipping and mobile depots.
\newblock {\em Transportation Science}, 56(3):612--630, 2022.

\bibitem{nieto2022value}
Santiago Nieto-Isaza, Pirmin Fontaine, and Stefan Minner.
\newblock The value of stochastic crowd resources and strategic location of
  mini-depots for last-mile delivery: A benders decomposition approach.
\newblock {\em Transportation Research Part B: Methodological}, 157:62--79,
  2022.

\bibitem{arslan2019crowdsourced}
Alp~M Arslan, Niels Agatz, Leo Kroon, and Rob Zuidwijk.
\newblock Crowdsourced delivery—a dynamic pickup and delivery problem with ad
  hoc drivers.
\newblock {\em Transportation Science}, 53(1):222--235, 2019.

\bibitem{savelsbergh2022challenges}
Martin~WP Savelsbergh and Marlin~W Ulmer.
\newblock Challenges and opportunities in crowdsourced delivery planning and
  operations.
\newblock {\em 4OR}, 20(1):1--21, 2022.

\bibitem{fatnassi2015planning}
Ezzeddine Fatnassi, Jouhaina Chaouachi, and Walid Klibi.
\newblock Planning and operating a shared goods and passengers on-demand rapid
  transit system for sustainable city-logistics.
\newblock {\em Transportation Research Part B: Methodological}, 81:440--460,
  2015.

\bibitem{pimentel2018integrated}
Carina Pimentel and Filipe Alvelos.
\newblock Integrated urban freight logistics combining passenger and freight
  flows--mathematical model proposal.
\newblock {\em Transportation research procedia}, 30:80--89, 2018.

\bibitem{azcuy2021designing}
Irecis Azcuy, Niels Agatz, and Ricardo Giesen.
\newblock Designing integrated urban delivery systems using public transport.
\newblock {\em Transportation Research Part E: Logistics and Transportation
  Review}, 156:102525, 2021.

\bibitem{hu2022mass}
Wanjie Hu, Jianjun Dong, Bon-Gang Hwang, Rui Ren, and Zhilong Chen.
\newblock Is mass rapid transit applicable for deep integration of
  freight-passenger transport? a multi-perspective analysis from urban china.
\newblock {\em Transportation Research Part A: Policy and Practice},
  165:490--510, 2022.

\bibitem{masson2017optimization}
Renaud Masson, Anna Trentini, Fabien Lehu{\'e}d{\'e}, Nicolas Malh{\'e}n{\'e},
  Olivier P{\'e}ton, and Houda Tlahig.
\newblock Optimization of a city logistics transportation system with mixed
  passengers and goods.
\newblock {\em EURO Journal on Transportation and Logistics}, 6(1):81--109,
  2017.

\bibitem{behiri2018urban}
Walid Behiri, Sana Belmokhtar-Berraf, and Chengbin Chu.
\newblock Urban freight transport using passenger rail network: Scientific
  issues and quantitative analysis.
\newblock {\em Transportation Research Part E: Logistics and Transportation
  Review}, 115:227--245, 2018.

\bibitem{kizil2023public}
Kerim~U K{\i}z{\i}l and Bar{\i}{\c{s}} Y{\i}ld{\i}z.
\newblock Public transport-based crowd-shipping with backup transfers.
\newblock {\em Transportation Science}, 57(1):174--196, 2023.

\bibitem{machado2023integration}
Bruno Machado, Carina Pimentel, and Amaro de~Sousa.
\newblock Integration planning of freight deliveries into passenger bus
  networks: Exact and heuristic algorithms.
\newblock {\em Transportation Research Part A: Policy and Practice},
  171:103645, 2023.

\bibitem{qi2018shared}
Wei Qi, Lefei Li, Sheng Liu, and Zuo-Jun~Max Shen.
\newblock Shared mobility for last-mile delivery: Design, operational
  prescriptions, and environmental impact.
\newblock {\em Manufacturing \& Service Operations Management}, 20(4):737--751,
  2018.

\bibitem{liu2023economic}
Yang Liu and Sen Li.
\newblock An economic analysis of on-demand food delivery platforms: Impacts of
  regulations and integration with ride-sourcing platforms.
\newblock {\em Transportation Research Part E: Logistics and Transportation
  Review}, 171:103019, 2023.

\bibitem{ghilas2013integrating}
Veaceslav Ghilas, Emrah Demir, and Tom Van~Woensel.
\newblock Integrating passenger and freight transportation: Model formulation
  and insights.
\newblock {\em BETA publicatie: working papers}, 441, 2013.

\bibitem{perboli2021simulation}
Guido Perboli, Mariangela Rosano, and Qu~Wei.
\newblock A simulation-optimization approach for the management of the
  on-demand parcel delivery in sharing economy.
\newblock {\em IEEE Transactions on Intelligent Transportation Systems}, 2021.

\bibitem{fehn2021ride}
Fabian Fehn, Roman Engelhardt, and Klaus Bogenberger.
\newblock Ride-parcel-pooling-assessment of the potential in combining
  on-demand mobility and city logistics.
\newblock In {\em 2021 IEEE International Intelligent Transportation Systems
  Conference (ITSC)}, pages 3366--3372. IEEE, 2021.

\bibitem{li2014share}
Baoxiang Li, Dmitry Krushinsky, Hajo~A Reijers, and Tom Van~Woensel.
\newblock The share-a-ride problem: People and parcels sharing taxis.
\newblock {\em European Journal of Operational Research}, 238(1):31--40, 2014.

\bibitem{chen2016crowddeliver}
Chao Chen, Daqing Zhang, Xiaojuan Ma, Bin Guo, Leye Wang, Yasha Wang, and Edwin
  Sha.
\newblock Crowddeliver: Planning city-wide package delivery paths leveraging
  the crowd of taxis.
\newblock {\em IEEE Transactions on Intelligent Transportation Systems},
  18(6):1478--1496, 2016.

\bibitem{li2016adaptive}
Baoxiang Li, Dmitry Krushinsky, Tom Van~Woensel, and Hajo~A Reijers.
\newblock An adaptive large neighborhood search heuristic for the share-a-ride
  problem.
\newblock {\em Computers \& Operations Research}, 66:170--180, 2016.

\bibitem{li2016share}
Baoxiang Li, Dmitry Krushinsky, Tom Van~Woensel, and Hajo~A Reijers.
\newblock The share-a-ride problem with stochastic travel times and stochastic
  delivery locations.
\newblock {\em Transportation Research Part C: Emerging Technologies},
  67:95--108, 2016.

\bibitem{manchella2021flexpool}
Kaushik Manchella, Abhishek~K Umrawal, and Vaneet Aggarwal.
\newblock Flexpool: A distributed model-free deep reinforcement learning
  algorithm for joint passengers and goods transportation.
\newblock {\em IEEE Transactions on Intelligent Transportation Systems},
  22(4):2035--2047, 2021.

\bibitem{bosse2023dynamic}
Alexander Bosse, Marlin~W Ulmer, Emanuele Manni, and Dirk~C Mattfeld.
\newblock Dynamic priority rules for combining on-demand passenger
  transportation and transportation of goods.
\newblock {\em European Journal of Operational Research}, 2023.

\bibitem{fehn2023integrating}
Fabian Fehn, Roman Engelhardt, Florian Dandl, Klaus Bogenberger, and Fritz
  Busch.
\newblock Integrating parcel deliveries into a ride-pooling service—an
  agent-based simulation study.
\newblock {\em Transportation Research Part A: Policy and Practice},
  169:103580, 2023.

\bibitem{beirigo2018integrating}
Breno~A Beirigo, Frederik Schulte, and Rudy~R Negenborn.
\newblock Integrating people and freight transportation using shared autonomous
  vehicles with compartments.
\newblock {\em IFAC-PapersOnLine}, 51(9):392--397, 2018.

\bibitem{lu2022combined}
Chung-Cheng Lu, Ali Diabat, Yi-Ting Li, and Yu-Min Yang.
\newblock Combined passenger and parcel transportation using a mixed fleet of
  electric and gasoline vehicles.
\newblock {\em Transportation Research Part E: Logistics and Transportation
  Review}, 157:102546, 2022.

\bibitem{zhan2023ride}
Xingbin Zhan, WY~Szeto, and Yue Wang.
\newblock The ride-hailing sharing problem with parcel transportation.
\newblock {\em Transportation Research Part E: Logistics and Transportation
  Review}, 172:103073, 2023.

\bibitem{Uber2023Guide}
Brett Helling.
\newblock Your guide to uber connect: A speedy local package delivery service.
\newblock \url{https://www.ridester.com/uber-connect/}, 2023.

\bibitem{yang2011equilibrium}
Hai Yang and Teng Yang.
\newblock Equilibrium properties of taxi markets with search frictions.
\newblock {\em Transportation Research Part B: Methodological}, 45(4):696--713,
  2011.

\bibitem{arnott1996taxi}
Richard Arnott.
\newblock Taxi travel should be subsidized.
\newblock {\em Journal of Urban Economics}, 40(3):316--333, 1996.

\bibitem{li2019regulating}
Sen Li, Hamidreza Tavafoghi, Kameshwar Poolla, and Pravin Varaiya.
\newblock Regulating tncs: Should uber and lyft set their own rules?
\newblock {\em Transportation Research Part B: Methodological}, 129:193--225,
  2019.

\bibitem{ross1996stochastic}
Sheldon~M Ross, John~J Kelly, Roger~J Sullivan, William~James Perry, Donald
  Mercer, Ruth~M Davis, Thomas~Dell Washburn, Earl~V Sager, Joseph~B Boyce, and
  Vincent~L Bristow.
\newblock {\em Stochastic processes}, volume~2.
\newblock Wiley New York, 1996.

\bibitem{plemmons1977m}
Robert~J Plemmons.
\newblock M-matrix characterizations. i—nonsingular m-matrices.
\newblock {\em Linear Algebra and its applications}, 18(2):175--188, 1977.

\bibitem{SFCTA2016data}
SFCTA.
\newblock Tnc pickup and dropoff data in san francisco.
\newblock \url{http://tncsandcongestion.sfcta.org/}, 2016.

\bibitem{castiglione2016tncs}
Joe Castiglione, Tilly Chang, Drew Cooper, Jeff Hobson, Warren Logan, Eric
  Young, Billy Charlton, Christo Wilson, Alan Mislove, Le~Chen, and Shan Jiang.
\newblock Tncs today: a profile of san francisco transportation network company
  activity.
\newblock {\em San Francisco County Transportation Authority}, 2017.

\bibitem{Uber2023fare}
Uber.
\newblock Uber fare san francisco.
\newblock \url{https://taxis-fare.com/uber-fare-city-san-francisco}, 2023.

\bibitem{Uber2023wage}
Uber.
\newblock Uber driver in san francisco, california.
\newblock
  \url{https://www.salary.com/tools/salary-calculator/uber-driver-hourly/san-francisco-ca},
  2023.

\bibitem{zipcode2023business}
Treasurer \& Tax~Collector’s Office.
\newblock Map of registered businesses - san francisco.
\newblock
  \url{https://data.sfgov.org/Economy-and-Community/Map-of-Registered-Businesses-San-Francisco/9tqe-vdng},
  2023.

\bibitem{zipcode2023database}
U.S.~Census Bureau.
\newblock Us zip codes database.
\newblock \url{https://simplemaps.com/data/us-zips}, 2023.

\bibitem{Uber2021report}
Uber.
\newblock Uber 2021 annual report.
\newblock
  \url{https://s23.q4cdn.com/407969754/files/doc_financials/2022/ar/2021-Annual-Report.pdf},
  2021.

\end{thebibliography}

\begin{appendices}

\section{Notations}\label{append:notation}
\begin{table}[H]
  \centering
  \caption{Summary of Notations for parameters}
    \begin{tabular}{ll}
    \toprule
    Notation & Definition \\
    \midrule
    \bf{Parameters} \\
    $\alpha_r$ & Passengers' value of time  \\
    $\alpha_d$ & Delivery customers' value of time\\
    $\epsilon$ & Sensitivity parameter for the passenger's choice model \\
    $\eta$ & Sensitivity parameter for the delivery customer's choice model\\
    $\sigma$ & Sensitivity parameter for the driver's choice model \\
    $\lambda_{ij}^{r,0}$ & Potential arrival rate of ride-sourcing passengers \\
    $\lambda_{ij}^{d,0}$ & Potential arrival rate of delivery customers\\
    $L_i$ & Scaling parameter in the matching function\\
    $M$ & The number of zones\\
    $N_0$ & Total number of for-hire drivers\\
    $c_{ij}^{r,0}$ & Average generalized cost of outside options for ride-sourcing passengers\\
    $c_{ij}^{d,0}$ & Average generalized cost of outside options for delivery customers\\
    $t_{ij}$ & Average ride-sourcing travel time from zone $i$ to zone $j$ \\
    $t_i^g$ & Average package drop-off time in zone $j$ \\
    $q_0$ & Average wages of outside options for drivers\\
    $w_{max}$ & Upper bound of passenger waiting time\\
        \bottomrule
    \end{tabular}%
  \label{tab:notation}%
\end{table}%

\begin{table}[H]
  \centering
  \caption{Summary of Notations for variables}
    \begin{tabular}{ll}
    \toprule
    Notation & Definition \\
    \midrule
    \bf{Endogenous variables} \\
    $\lambda_{ij}^r$ & Arrival rate of ride-sourcing passengers\\
    $\lambda_{ij}^{d_f}$ & Arrival rate of flexible delivery customers \\
    $\lambda_{ij}^{d_o}$ & Arrival rate of on-demand delivery customers \\
    $\nu_{z,n}$ & The average spent time in state $(z,n)$ in the CTMC\\
    $N_i^I$ & The number of idle drivers in zone $i$\\
    $\bar N_i^I$  & The number of drivers available for flexible pickup\\
    $N_i^{Ig}$  & The number of drivers that can successfully pick up a flexible package\\
    $N_{(z,n)}^I$ & The number of idle drivers in zone $z$ carrying $n$ flexible packages\\
    $P_{ij}$ & The probability for a driver in zone $i$ to move to zone $j$ \\
    $P_{(z,n)(z',n')}^c$ & The transition probability for a driver in the CTMC \\
    $P_{(z,n)}^c$ & The limiting probability of the CTMC \\
    $S_{ij}$ & The transition time for a driver in zone $i$ to zone $j$\\
    $T_{ij}$ & The time for a driver from zone $i$ to first arrive at zone $j$\\
    $c_{ij}^r$  & The average generalized cost for the ride-sourcing passengers \\
    $c_{ij}^{d_f}$ & The average generalized cost for the flexible delivery customers \\
    $c_{ij}^{d_o}$ & The average generalized cost for the on-demand delivery customers \\
    $w_i^I$ & Average waiting time for drivers to receive an on-demand order \\
    $w_i^r$ & Average waiting time for ride-sourcing passengers/on-demand delivery customers\\
    $w_i^{d_f}$ & Average waiting time for flexible delivery customers \\ 
    $w_i^{dg}$ & Average waiting time for idle drivers to be matched with a flexible delivery order\\
    $p_{i,pick}^{succ}$/$p_{i,drop}^{succ}$ & Successful rate for flexible package pickup/drop-off\\
    $p_{i,pick}^n$/$p_{i,drop}^n$ & Probability for an idle driver in zone $i$ carrying $n$ flexible parcels to \\
     & successfully pick up/drop off a new flexible package\\
     $p_{flex,i}^n$ & Probability for a driver with $n$ carrying flexible packages to drop off at zone $i$\\
     $t_{ij}^{d_f}$ & Average delivery time for a flexible package from zone $i$ to zone $j$\\
     $\bar t_i^g$ & Average time for a driver in zone $i$ to pick up a flexible package\\
    \bf{Decision variables} \\
    $r_i^r$ & Average per-time trip fare for ride-sourcing passengers \\
    $r_{ij}^{d_f}$ & Average delivery fee for flexible delivery services \\
    $q$ & Average hourly wages paid to drivers\\
    \bottomrule
    \end{tabular}%
  \label{tab:notation_2}%
\end{table}%

\end{appendices}


\end{document}